\documentclass[preprint,3p,12pt]{elsarticle}

\usepackage{amsmath,amssymb,amsfonts,bm,mathtools}
\usepackage{graphicx}
\usepackage{booktabs,multirow}
\usepackage{algorithm}
\usepackage[noend]{algpseudocode}
\algrenewcommand{\algorithmicrequire}{\textbf{Input:}}
\algrenewcommand{\algorithmicensure}{\textbf{Output:}}
\usepackage{float}
\usepackage{flafter}
\usepackage{placeins}
\usepackage{xcolor}
\usepackage{hyperref}
\usepackage{amsthm}

\AddToHook{env/table/begin}{\setlength{\belowcaptionskip}{6pt}}

\newtheorem{proposition}{Proposition}

\theoremstyle{remark}

\theoremstyle{plain}

\newcommand{\dd}{\mathrm{d}}
\newcommand{\panelcaption}[1]{\par\vspace{2pt}{\small #1\par}}

\newcommand{\fvec}{\bm{f}}
\newcommand{\Flux}{\bm{F}}
\newcommand{\mFlux}{\bm{\mathcal F}}
\newcommand{\W}{\bm{W}}
\newcommand{\M}{\bm{M}}
\newcommand{\rvec}{\bm{r}}
\newcommand{\avec}{\bm{a}}
\newcommand{\dvec}{\bm{d}}
\newcommand{\uvec}{\bm{u}}
\newcommand{\cvec}{\bm{c}}
\newcommand{\Cmat}{\mathbf C}

\newcommand{\Rop}{\bm{R}}
\newtheorem{corollary}{Corollary}

\begin{document}
\begin{frontmatter}
\title{A nonequilibrium-capturing unified gas-kinetic scheme I: discrete orthogonal basis construction and expansion on velocity space}
\author[HKUST1]{Hongyu Liu}
\ead{hliudv@connect.ust.hk}
\author[HKUST1,HKUST2,HKUST3]{Kun Xu\corref{cor1}}
\ead{makxu@ust.hk}

\cortext[cor1]{Corresponding author}

\address[HKUST1]{Department of Mathematics, Hong Kong University of Science and Technology, Clear Water Bay, Kowloon, Hong Kong}
\address[HKUST2]{Department of Mechanical and Aerospace Engineering, Hong Kong University of Science and Technology, Clear Water Bay, Kowloon, Hong Kong}
\address[HKUST3]{Shenzhen Research Institute, Hong Kong University of Science and Technology, Shenzhen, China}
\begin{abstract}
Velocity space discretization is a primary source of computational cost in deterministic methods for multiscale gas flows. A fixed velocity grid often retains redundant kinetic degrees of freedom (DOFs), particularly under near equilibrium conditions. To address this issue, this work develops a nonequilibrium-capturing formulation based on the unified gas kinetic scheme (UGKS), in which the number of kinetic DOFs adapts locally to the degree of nonequilibrium. A discrete orthogonal basis is constructed for numerical implementation, and the velocity distribution function is decomposed into an analytical local equilibrium and a basis expanded nonequilibrium component. The retained, locally adapted DOFs are evolved directly, substantially lowering computational overhead while preserving the discrete conservation of mass, momentum, and energy, as well as accurately capturing nonequilibrium states. The formulation retains the core multiscale coupling of molecular transport and collision inherent to UGKS, providing explicit and implicit schemes for unsteady and steady flows, respectively. Numerical tests across a wide range of flow speeds and rarefaction regimes (from continuum to highly rarefied) demonstrate close agreement with benchmark solutions. For unsteady flows, the explicit scheme achieves speedups of 12.64 to 58.50 over the standard UGKS. In steady state tests, only about $2\%$ of the original kinetic DOFs are required in continuum regions, accelerating each implicit iteration by a factor of 33.87. Even in the most rarefied cases, the scheme captures nonequilibrium behavior using only about $27\%$ of the original DOFs. 
Combined with steady acceleration techniques, the method achieves overall speedups ranging from 13.7 to approximately 2,344 relative to the standard implicit UGKS. By dramatically reducing the cost of resolving nonequilibrium effects, the proposed framework provides an efficient and practical path toward the deterministic simulation of multiscale gas dynamics.

\end{abstract}

\end{frontmatter}

\section{Introduction}\label{sec:intro}

In multiscale gas flows, near equilibrium regions frequently coexist with strong nonequilibrium structures such as shock waves, Knudsen layers, and sharp expansions. Resolving such nonequilibrium phenomena necessitates a kinetic description, whose computational overhead is heavily dictated by the discretization of velocity space. A central challenge lies in preserving the multiscale coupling of molecular transport and collision across all flow regimes while utilizing only the kinetic degrees of freedom (DOFs) required by the local degree of nonequilibrium.

Existing multiscale kinetic methods have made substantial progress in capturing this transport--collision coupling. The unified gas kinetic scheme (UGKS) incorporates coupled molecular transport and collision into the numerical fluxes at cell interfaces via an analytical integral solution of the Bhatnagar--Gross--Krook (BGK) model~\cite{xu2010ugks,huang2012ugks}. Consequently, it recovers the hydrodynamic Navier Stokes limit without resolving the particle mean free path, while preserving high physical fidelity in rarefied regimes. Following a similar principle, the discrete unified gas kinetic scheme (DUGKS) achieves this multiscale coupling through a discrete characteristic line integration~\cite{guo2013dugks,guo2015dugks}.
In particle based frameworks, this physical coupling has also been leveraged to curb computational cost in near equilibrium regions. The unified gas kinetic wave-particle (UGKWP) method decomposes the solution into an analytical hydrodynamic wave and stochastic particles~\cite{zhu2019ugkwp2}, where the analytical component predominates in the continuum limit. The unified stochastic particle (USP) method similarly couples molecular convection and collisional relaxation within a unified particle step~\cite{fei2020usp}. More recently, the multiscale stochastic particle method~\cite{cui2025msp} and multiscale simulation Monte Carlo method~\cite{yang2026msmc} mitigate numerical discretization errors in Fokker--Planck and Boltzmann particle frameworks, respectively, enabling the use of coarser spatial and temporal grids. These advances clearly demonstrate the power of transport--collision coupling in facilitating efficient multiscale simulations.
For steady state problems, computational efficiency is additionally governed by the convergence rate. The general synthetic iterative scheme (GSIS) accelerates steady convergence by coupling the microscopic kinetic equation with macroscopic synthetic governing equations that incorporate kinetic stress and heat flux corrections~\cite{zhu2021gsis}. Implicit UGKS and its synergy with GSIS have also substantially enhanced steady state convergence performance~\cite{zhu2016implicit,xu2021implicit,liu2026coupled}. Nevertheless, in tandem with accelerated convergence, reducing the computational effort per kinetic iteration remains an essential avenue for achieving scalable efficiency.

For the numerical methods described above, computational efficiency can be further augmented through spatial domain adaptation, which tailors the governing physical description to the local flow regime. Adaptive UGKS employs continuous Chapman--Enskog distribution functions in near equilibrium regions and discrete velocity distributions elsewhere~\cite{xiao2018adaptive}, with an implicit formulation that accelerates steady state calculations~\cite{long2024adaptive}. Similarly, adaptive DUGKS couples a Navier Stokes solver in continuum regions with a BGK solver in non-equilibrium zones~\cite{yang2023adaptive}, while adaptive GSIS identifies kinetic domains via a local Knudsen number criterion~\cite{zhang2025adaptive}. Domain adaptation has also proven effective in particle based methods: adaptive UGKWP dynamically tunes the wave--particle ratio based on local nonequilibrium~\cite{wei2023adaptive}, and adaptive USP adjusts spatial grids and time steps adaptively~\cite{feng2024adaptive}. Nonetheless, in these discrete velocity methods, the number of discrete velocity points remains fixed across all designated kinetic regions.

Additional computational gains can be achieved by adapting the velocity space resolution directly within kinetic domains. Existing strategies dynamically adjust the velocity grid according to the local flow conditions~\cite{chen2012adaptive,baranger2014refined} or employ spatially and temporally varying local meshes~\cite{brull2014local}. However, transport across non-conforming velocity grids necessitates intricate interpolation to guarantee the conservation of macroscopic moments. Another promising avenue exploits the intrinsic mathematical structure of the velocity distribution function: its local equilibrium state is determined by a few macroscopic conserved variables, leaving only the remaining component to represent nonequilibrium effects. In this spirit, the micro--macro decomposition of Bennoune et al.\cite{bennoune2008micromacro} isolates the local Maxwellian from a nonequilibrium distribution constrained to possess vanishing mass, momentum, and energy moments. Dynamical low-rank approaches have also been utilized to approximate distribution functions after factoring out a Maxwellian background\cite{einkemmer2021lowrank}. Furthermore, Liu et al.~\cite{liu2026hermite} coupled Hermite spectral expansions with UGKS for flows under external forcing. These developments motivate an analytical treatment of local equilibrium paired with an adaptive kinetic representation of the nonequilibrium state.

Building on this viewpoint, we present a nonequilibrium-capturing formulation of the unified gas kinetic scheme (NCUGKS). In this framework, the local Maxwellian is evaluated analytically, whereas the nonequilibrium component is expanded in a discrete orthogonal basis defined on each velocity quadrant. The DOF are selected in blocks based on their contribution to the nonequilibrium distribution function, with truncation governed by a prescribed tolerance. This mechanism allows the kinetic DOFs to dynamically adapt to the intensity and anisotropy of local nonequilibrium, allocating minimal DOFs near equilibrium and progressively more where steep kinetic gradients occur.

Crucially, the UGKS interface flux and macroscopic/microscopic updates are formulated directly in terms of these adaptive DOFs. Polynomial recurrence relations yield sparse transport operators, while the separable structure of the Maxwellian enables efficient analytical flux evaluation. By evolving the truncated DOFs directly, the reduction in velocity space representation translates directly into actual computational speedup while preserving the multiscale transport--collision coupling of UGKS. Notably, the number of active kinetic DOFs can vary from cell to cell without altering the underlying global velocity grid. A local constrained least-squares correction enforces that the nonequilibrium DOFs contribute strictly zero to the conserved macroscopic moments, while a shared interface flux maintains strict numerical conservation across adjacent cells regardless of differences in their active DOFs.

The framework incorporates both explicit and implicit formulations for unsteady and steady flows, respectively. For steady problems, the implicit scheme is further coupled with multi step macro prediction~\cite{liu2026waveparticleII} and conservative Anderson acceleration to expedite convergence. \mbox{Numerical} benchmarks across low to high Mach numbers and continuum to rarefied regimes demonstrate excellent agreement with reference solutions. In steady state tests, the proposed method requires only about $2\%$ of the original kinetic DOFs in continuum regimes and approximately $27\%$ even in the most rarefied cases. When integrated with steady acceleration techniques, NCUGKS achieves an overall speedup of approximately $2344$ in continuum flow relative to the standard implicit UGKS.

In this paper, Section~\ref{sec:model} introduces the reference UGKS. Sections~\ref{sec:modal}--\ref{sec:properties} develop the nonequilibrium formulation, adaptive evolution, and discrete properties. Section~\ref{sec:implicit-extension} presents the implicit extension, Section~\ref{sec:results} reports numerical results, and Section~\ref{sec:conclusions} concludes the paper.

\section{UGKS discretization}\label{sec:model}

This section introduces the kinetic model, the conservative update scheme, and the interface flux formulations that constitute the foundation of the nonequilibrium capturing method.
\subsection{Kinetic model and reduced distribution functions}\label{sec:distribution-notation}
Consider a gas flow in two spatial dimensions, defined by the position vector $\bm x=(x,y)^{\mathsf T}$ and the molecular velocity $\uvec=(u,v)^{\mathsf T}$. The evolution of the distribution function $f(\bm x,t,\uvec,\bm\xi)$ is governed by the BGK equation~\cite{bgk1954}:
\begin{equation}
 \partial_t f+\uvec\cdot\nabla_{\bm x}f=\frac{g-f}{\tau},
 \qquad \tau=\frac{\mu}{p}.
 \label{eq:bgk}
\end{equation}
Here, $\bm\xi\in\mathbb R^K$ represents the third velocity component and the internal molecular degrees of freedom. The variables $\tau$, $\mu$, and $p$ denote the relaxation time, dynamic viscosity, and pressure, respectively. The local equilibrium state $g$ is determined by the density $\rho$, flow velocity $\bm U=(U,V)^{\mathsf T}$, and temperature $T$. In terms of the peculiar velocity $\cvec=\uvec-\bm U$, the equilibrium distribution takes the Maxwellian form:
\begin{equation}
 g=\frac{\rho}{(2\pi\theta)^{(K+2)/2}}
   \exp\!\left[-\frac{|\cvec|^2+|\bm\xi|^2}{2\theta}\right],
 \qquad \theta=RT,
 \label{eq:scalar-maxwellian}
\end{equation}
where $R$ is the specific gas constant and $p=\rho\theta$. For a monatomic gas, $K=1$ and $\xi$ corresponds to the molecular velocity in the $z$-direction~\cite{zhu2016implicit}. More generally, the constant $K$ determines the specific heat ratio via $\gamma=(K+4)/(K+2)$. The viscosity is modeled using the power law $\mu=\mu_{\mathrm{ref}}(T/T_{\mathrm{ref}})^\omega$, where $\mu_{\mathrm{ref}}$ is the viscosity at the reference temperature $T_{\mathrm{ref}}$ and $\omega$ is the viscosity exponent.

Integrating over $\bm\xi$ gives two reduced distribution functions:
\begin{equation}
 h=\int_{\mathbb R^K}f\dd\bm\xi,
 \qquad
 b=\int_{\mathbb R^K}|\bm\xi|^2f\dd\bm\xi.
 \label{eq:reduction}
\end{equation}
The conservative variables are $\W=(\rho,\rho U,\rho V,\rho E)^{\mathsf T}$, with specific total energy $E=|\bm U|^2/2+(K+2)\theta/2$. Writing $\fvec=(h,b)^{\mathsf T}$, integration of Eq.~\eqref{eq:bgk} yields
\begin{equation}
 \partial_t\fvec+u\partial_x\fvec+v\partial_y\fvec
 =\frac{\M(\W)-\fvec}{\tau},
 \label{eq:reduced-bgk}
\end{equation}
where the reduced equilibrium $\M(\W,\uvec)=(M_h,M_b)^{\mathsf T}$ is given by
\begin{equation}
 \begin{aligned}
 M_h&=\frac{\rho}{2\pi\theta}
       \exp\!\left[-\frac{(u-U)^2+(v-V)^2}{2\theta}\right],\\
 M_b&=K\theta M_h.
 \end{aligned}
 \label{eq:maxwellian}
\end{equation}
The conservative variables follow from the velocity moments
\begin{equation}
 \W=\int_{\mathbb R^2}\bm\Psi(\uvec)\fvec(\uvec)\dd\uvec,
 \qquad
 \bm\Psi(\uvec)=
 \begin{pmatrix}
  1&0\\
  u&0\\
  v&0\\
  \tfrac12(u^2+v^2)&\tfrac12
 \end{pmatrix}.
 \label{eq:reduced-moments}
\end{equation}
The four rows of $\bm\Psi$ give the weights for mass, the two momentum components, and total energy. The equilibrium has the same moments,
\[
 \int_{\mathbb R^2}\bm\Psi(\uvec)\M(\W,\uvec)\dd\uvec=\W.
\]

\subsection{Finite volume discretization}\label{sec:finite-volume}
The physical domain is divided into cells $\Omega_i$ of area $|\Omega_i|$. The faces of cell $i$ form the set $\mathcal J_i$, with boundary faces in $\mathcal J_i^\partial$. A face $ij$ has length $A_{ij}$ and outward unit normal $\bm n_{ij}$. For an interior face, $j$ denotes the adjacent cell.

The velocity space discretization uses numerical quadrature with nodes $\uvec_\nu=(u_\nu,v_\nu)^{\mathsf T}$ and positive weights $w_\nu$, $\nu=0,\ldots,N_{\mathrm{vel}}-1$. With $\bm\Psi_\nu=\bm\Psi(\uvec_\nu)$, the nodal vector of the reduced distribution functions is
\[
 \fvec^{\mathrm v}
 =\bigl(h_0,\ldots,h_{N_{\mathrm{vel}}-1},
        b_0,\ldots,b_{N_{\mathrm{vel}}-1}\bigr)^{\mathsf T}
 \in\mathbb R^{2N_{\mathrm{vel}}},
\]
where $h_\nu=h(\uvec_\nu)$ and $b_\nu=b(\uvec_\nu)$. The superscript $\mathrm v$ denotes this nodal representation for the reduced distribution functions and microscopic fluxes.

Over the time step $\Delta t=t^{n+1}-t^n$, the interface distribution function $\fvec_{ij}(t,\uvec)$ determines the time averaged microscopic flux
\begin{equation}
 \Flux_{ij}^n(\uvec)
 =\frac{1}{\Delta t}\int_0^{\Delta t}
                    u_{ij}\fvec_{ij}(t,\uvec)\dd t,
 \qquad u_{ij}=\uvec\cdot\bm n_{ij}.
 \label{eq:microscopic-face-flux}
\end{equation}
Here $t$ is measured from $t^n$, and $u_{ij,\nu}=\uvec_\nu\cdot\bm n_{ij}$ is the normal velocity. The macroscopic flux follows by taking the conservative moments of the microscopic flux with the same quadrature
\begin{equation}
 \mFlux_{ij}^n
 =\sum_{\nu=0}^{N_{\mathrm{vel}}-1}
       w_\nu\bm\Psi_\nu\Flux_{ij}^n(\uvec_\nu).
 \label{eq:macroscopic-face-flux}
\end{equation}
With the microscopic flux divergence
\begin{equation}
 \bm{\mathcal D}_i^n
 =\frac{1}{|\Omega_i|}\sum_{j\in\mathcal J_i}A_{ij}\Flux_{ij}^n,
 \label{eq:flux-divergence}
\end{equation}
the cell averages evolve as
\begin{gather}
 \W_i^{n+1}
 =\W_i^n-\frac{\Delta t}{|\Omega_i|}
            \sum_{j\in\mathcal J_i}A_{ij}\mFlux_{ij}^n,
 \label{eq:macro-update}\\
 \frac{\fvec_i^{n+1}-\fvec_i^n}{\Delta t}
 +\bm{\mathcal D}_i^n
 =\frac{\M(\W_i^{n+1})-\fvec_i^{n+1}}{\tau_i^{n+1}}.
 \label{eq:distribution-update}
\end{gather}
The superscript $n$ is omitted in the local interface formulas below.

\subsection{Interface evolution}
The BGK integral solution couples molecular transport and collision in the interface flux. With $\tau$ fixed over the time step, the solution at $\bm x=\bm0$ is~\cite{xu2010ugks}
\begin{equation}
 \fvec_{ij}(t)=e^{-t/\tau}\fvec_{\mathrm{init}}(-\uvec t)
 +\frac{1}{\tau}\int_0^t e^{-(t-t')/\tau}\M(-\uvec(t-t'),t')\dd t'.
 \label{eq:integral-solution}
\end{equation}
Here $\fvec_{\mathrm{init}}$ is the initial distribution function, and $\M(\bm x,t)=\M(\W(\bm x,t),\uvec)$ for fixed $\uvec$. Linear expansions of the initial distribution function in space and the equilibrium in space and time give
\begin{equation}
 \begin{aligned}
 \fvec_{\mathrm{init}}(-\uvec t)&\simeq\fvec_0-t(u\fvec_x+v\fvec_y),\\
 \M(-\uvec(t-t'),t')
 &\simeq\M_0-(t-t')(u\M_x+v\M_y)+t'\M_t.
 \end{aligned}
 \label{eq:characteristic-expansions}
\end{equation}
The spatial gradients of the distribution function and the macroscopic variables $\W$ are reconstructed with the van Leer limiter. For the left and right interface values $\fvec_0^L$ and $\fvec_0^R$, upwinding gives
\begin{equation}
 \fvec_0(\uvec)=
 \begin{cases}
  \fvec_0^L(\uvec),&u_{ij}\ge0,\\
  \fvec_0^R(\uvec),&u_{ij}<0,
 \end{cases}
 \qquad
 \W_0=\sum_{\nu=0}^{N_{\mathrm{vel}}-1}
             w_\nu\bm\Psi_\nu\fvec_0(\uvec_\nu).
 \label{eq:interface-half-space-state}
\end{equation}
The same upwind choice is used for $\fvec_x$ and $\fvec_y$. The interface variables $\W_0$ determine $\M_0=\M(\W_0)$ and $\tau$, while their spatial gradients give $\M_x$ and $\M_y$ through the chain rule. The compatibility condition determines the temporal derivative:
\begin{equation}
 \begin{aligned}
 \W_t&=-\int_{\mathbb R^2}\bm\Psi(\uvec)
                    \bigl(u\M_x+v\M_y\bigr)\dd\uvec,\\
 \M_t&=\left.\frac{\partial\M}{\partial\W}\right|_{\W_0}\W_t.
 \end{aligned}
 \label{eq:interface-time-derivatives}
\end{equation}
Substituting Eq.~\eqref{eq:characteristic-expansions} into Eq.~\eqref{eq:integral-solution} and applying the time average in Eq.~\eqref{eq:microscopic-face-flux} gives the UGKS flux
\begin{equation}
 \begin{aligned}
 \Flux_{ij}
 &=u_{ij}\bigl[\bar c_1\M_0+\bar c_2(u\M_x+v\M_y)+\bar c_3\M_t\bigr]\\
 &\quad+u_{ij}\bigl[\bar c_4\fvec_0
                         +\bar c_5(u\fvec_x+v\fvec_y)\bigr],
 \end{aligned}
 \label{eq:reference-flux}
\end{equation}
with time coefficients
\begin{equation}
 \begin{aligned}
 \bar c_1&=1-\bar c_4,
 &\bar c_4&=\frac{\tau}{\Delta t}\bigl(1-e^{-\Delta t/\tau}\bigr),\\
 \bar c_2&=-\bar c_5-\tau\bar c_1,
 &\bar c_5&=\tau\bigl(e^{-\Delta t/\tau}-\bar c_4\bigr),\\
 \bar c_3&=\tfrac12\Delta t-\tau\bar c_1.
 \end{aligned}
 \label{eq:time-coefficients}
\end{equation}

\section{Discrete basis representation of the nonequilibrium distribution function}\label{sec:modal}
In this section, we separate the distribution function in cell $i$ into the local Maxwellian determined by $\W_i$ and a nonequilibrium part:
\begin{equation}
 \fvec_i^{\mathrm v}=\M_i^{\mathrm v}+\rvec_i^{\mathrm v},
 \qquad \M_i^{\mathrm v}=\M^{\mathrm v}(\W_i).
 \label{eq:cell-macro-micro-decomposition}
\end{equation}
The nonequilibrium vector $\rvec_i^{\mathrm v}$ is expanded in a discrete orthonormal basis.

\subsection{Projection of the nonequilibrium distribution function}\label{sec:modal-projection}
The midpoint velocity grid in \ref{app:basis-discrete-space} is divided into four quadrants:
\[
 \begin{aligned}
 Q_1&=\{\uvec_\nu:u_\nu>0,\ v_\nu>0\},
 &Q_2&=\{\uvec_\nu:u_\nu<0,\ v_\nu>0\},\\
 Q_3&=\{\uvec_\nu:u_\nu<0,\ v_\nu<0\},
 &Q_4&=\{\uvec_\nu:u_\nu>0,\ v_\nu<0\}.
 \end{aligned}
\]
The index $q\in\{1,2,3,4\}$ identifies the velocity quadrant throughout. The reason for this partition is explained in Corollary~\ref{cor:quadrant-upwind}.

On $Q_q$, the velocity nodes $(u_k^q,v_l^q)$, $0\leq k<N_u^q$ and $0\leq l<N_v^q$, form a Cartesian grid of $N_q=N_u^qN_v^q$ nodes. Here $N_u^q=N_u$ and $N_v^q=N_v$. The values of the cell averaged distribution functions $h_i$ and $b_i$ on this grid form the column vectors
\begin{equation}
 \begin{aligned}
 \fvec_i^{\mathrm v,q,h}
 &=\bigl(h_i(u_0^q,v_0^q),\ h_i(u_0^q,v_1^q),\ \ldots,\
 h_i(u_{N_u^q-1}^q,v_{N_v^q-1}^q)\bigr)^{\mathsf T}
 \in\mathbb R^{N_q},\\
 \fvec_i^{\mathrm v,q,b}
 &=\bigl(b_i(u_0^q,v_0^q),\ b_i(u_0^q,v_1^q),\ \ldots,\
 b_i(u_{N_u^q-1}^q,v_{N_v^q-1}^q)\bigr)^{\mathsf T}
 \in\mathbb R^{N_q}.
 \end{aligned}
 \label{eq:quadrant-nodal-vectors}
\end{equation}
In these vectors, $h_i(u_k^q,v_l^q)=h_{i,\nu}^q$ and $b_i(u_k^q,v_l^q)=b_{i,\nu}^q$, where $\nu=kN_v^q+l$ and $0\le\nu<N_q$. The notation $(u_\nu,v_\nu)$ and $w_\nu$ refers to these local nodes and weights when $q$ is fixed.
The corresponding quadrant components of Eq.~\eqref{eq:cell-macro-micro-decomposition} satisfy
\[
 \fvec_i^{\mathrm v,q,s}
 =\M_i^{\mathrm v,q,s}+\rvec_i^{\mathrm v,q,s},
 \qquad q=1,\ldots,4,\quad s=h,b,
\]
where $\M_i^{\mathrm v,q,s}$ is the column vector of Maxwellian values $M_s(\W_i,u_k^q,v_l^q)$, and $\rvec_i^{\mathrm v,q,s}$ is the nonequilibrium vector.

The discrete orthonormal basis is defined with respect to the following quadrature inner products. The positive quadrature weights satisfy $w_\nu=w_{u,k}^qw_{v,l}^q$, with $w_{u,k}^q=\Delta u$ and $w_{v,l}^q=\Delta v$. For $\mathbf z(u,v),\mathbf y(u,v)\in\mathbb R^{N_q}$, the two dimensional inner product is
\begin{equation}
 \begin{aligned}
 \langle\mathbf z(u,v),\mathbf y(u,v)\rangle_q
 &=\sum_{\uvec_\nu\in Q_q}w_\nu z(u_\nu,v_\nu)y(u_\nu,v_\nu)\\
 &=\mathbf z^{\mathsf T}\mathbf\Omega_q\mathbf y,
 \qquad
 \mathbf\Omega_q=\operatorname{diag}(w_\nu)_{\uvec_\nu\in Q_q}.
 \end{aligned}
 \label{eq:quadrant-inner-product}
\end{equation}
For coordinate vectors $\mathbf z_u,\mathbf y_u\in\mathbb R^{N_u^q}$ and $\mathbf z_v,\mathbf y_v\in\mathbb R^{N_v^q}$, the one dimensional inner products are
\begin{equation}
 \langle\mathbf z_u,\mathbf y_u\rangle_{w_u^q}
 =\sum_{k=0}^{N_u^q-1}w_{u,k}^q(\mathbf z_u)_k(\mathbf y_u)_k,
 \qquad
 \langle\mathbf z_v,\mathbf y_v\rangle_{w_v^q}
 =\sum_{l=0}^{N_v^q-1}w_{v,l}^q(\mathbf z_v)_l(\mathbf y_v)_l.
 \label{eq:axis-inner-products}
\end{equation}

Let $\mathbf p_{u,m_1}^q\in\mathbb R^{N_u^q}$ and $\mathbf p_{v,m_2}^q\in\mathbb R^{N_v^q}$ be the one dimensional orthonormal basis vectors of degrees $m_1$ and $m_2$, respectively, based on the inner products~\eqref{eq:axis-inner-products}. Their degree pairs $\bm m=(m_1,m_2)$ belong to
\begin{equation}
 \mathcal V_q=\{0,\ldots,N_u^q-1\}\times\{0,\ldots,N_v^q-1\},
 \label{eq:quadrant-degree-set}
\end{equation}
where $\times$ denotes the Cartesian product. The tensor product basis $\bm\phi_{\bm m}^q\in\mathbb R^{N_q}$ satisfies
\begin{equation}
 \begin{aligned}
 \bm\phi_{\bm m}^q
 &=\bm\phi_{(m_1,m_2)}^q
 =\mathbf p_{u,m_1}^q\otimes\mathbf p_{v,m_2}^q,\\
 (\bm\phi_{\bm m}^q)_\nu
 &=(\bm\phi_{(m_1,m_2)}^q)_\nu
 =(\mathbf p_{u,m_1}^q)_k(\mathbf p_{v,m_2}^q)_l,
 \qquad \nu=kN_v^q+l,\\
 \langle\bm\phi_{\bm m}^q,\bm\phi_{\bm n}^q\rangle_q
 &=\langle\bm\phi_{(m_1,m_2)}^q,\bm\phi_{(n_1,n_2)}^q\rangle_q
 =\delta_{m_1n_1}\delta_{m_2n_2}.
 \end{aligned}
 \label{eq:tensor-basis}
\end{equation}
Here $\bm n=(n_1,n_2)\in\mathcal V_q$, and $\otimes$ denotes the Kronecker product, with the velocity node order $\nu=kN_v^q+l$ used for $\fvec_i^{\mathrm v,q,s}$ in Eq.~\eqref{eq:quadrant-nodal-vectors}. The one dimensional basis construction is given in \ref{app:basis-one-dimensional}.

The nonequilibrium distribution function has the expansion
\begin{equation}
 \rvec_i^{\mathrm v,q,s}
 =\sum_{(m_1,m_2)\in\mathcal V_q}a_{i,(m_1,m_2)}^{q,s}\bm\phi_{(m_1,m_2)}^q.
 \label{eq:nonequilibrium-expansion}
\end{equation}
The coefficient $a_{i,(m_1,m_2)}^{q,s}$ is the projection of the nonequilibrium vector onto $\bm\phi_{(m_1,m_2)}^q$:
\begin{equation}
 a_{i,(m_1,m_2)}^{q,s}
 =\langle\rvec_i^{\mathrm v,q,s},\bm\phi_{(m_1,m_2)}^q\rangle_q
 =(\bm\phi_{(m_1,m_2)}^q)^{\mathsf T}\mathbf\Omega_q\rvec_i^{\mathrm v,q,s}.
 \label{eq:representation}
\end{equation}

\begin{samepage}
Equation~\eqref{eq:representation} gives the vector projection
\[
 \avec_i^{q,s}=\mathbf P_q^{\mathsf T}\mathbf\Omega_q\rvec_i^{\mathrm v,q,s},
\]
where
\[
 \begin{aligned}
 \mathbf P_q&=\bigl(
   \bm\phi_{(0,0)}^q,\ \bm\phi_{(0,1)}^q,\ \bm\phi_{(1,0)}^q,\ \ldots,\
   \bm\phi_{(N_u^q-1,N_v^q-1)}^q\bigr)
   \in\mathbb R^{N_q\times N_q},\\
 \avec_i^{q,s}&=\bigl(
   a_{i,(0,0)}^{q,s},\ a_{i,(0,1)}^{q,s},\ a_{i,(1,0)}^{q,s},\ \ldots,\
   a_{i,(N_u^q-1,N_v^q-1)}^{q,s}\bigr)^{\mathsf T}
   \in\mathbb R^{N_q}.
 \end{aligned}
\]
\end{samepage}
In both expressions, the degree pairs are ordered by increasing $m_1+m_2$, then by increasing $m_1$. The expansion and orthonormality give
\begin{equation}
 \rvec_i^{\mathrm v,q,s}=\mathbf P_q\avec_i^{q,s},
 \qquad
 \mathbf P_q^{\mathsf T}\mathbf\Omega_q\mathbf P_q=\mathbf I_{N_q}.
 \label{eq:quadrant-basis-identities}
\end{equation}

Ordering the nodes by quadrant within each reduced distribution function, the expansions take the block form
\[
 \underbrace{\begin{pmatrix}
   \rvec_i^{\mathrm v,1,h}\\\rvec_i^{\mathrm v,2,h}\\
   \rvec_i^{\mathrm v,3,h}\\\rvec_i^{\mathrm v,4,h}\\[3pt]
   \rvec_i^{\mathrm v,1,b}\\\rvec_i^{\mathrm v,2,b}\\
   \rvec_i^{\mathrm v,3,b}\\\rvec_i^{\mathrm v,4,b}
 \end{pmatrix}}_{\rvec_i^{\mathrm v}}
 =
 \underbrace{\begin{pmatrix}
   \mathbf P_1&&&&&&&0\\
   &\mathbf P_2&&&&&&\\
   &&\mathbf P_3&&&&&\\
   &&&\mathbf P_4&&&&\\[3pt]
   &&&&\mathbf P_1&&&\\
   &&&&&\mathbf P_2&&\\
   &&&&&&\mathbf P_3&\\
   0&&&&&&&\mathbf P_4
 \end{pmatrix}}_{\mathbf P}
 \underbrace{\begin{pmatrix}
   \avec_i^{1,h}\\\avec_i^{2,h}\\\avec_i^{3,h}\\\avec_i^{4,h}\\[3pt]
   \avec_i^{1,b}\\\avec_i^{2,b}\\\avec_i^{3,b}\\\avec_i^{4,b}
 \end{pmatrix}}_{\avec_i}.
\]
Here $\rvec_i^{\mathrm v},\avec_i\in\mathbb R^{2N_{\mathrm{vel}}}$ and $\mathbf P\in\mathbb R^{2N_{\mathrm{vel}}\times2N_{\mathrm{vel}}}$, where $N_{\mathrm{vel}}=\sum_{q=1}^4N_q$. The weight matrix $\mathbf\Omega\in\mathbb R^{2N_{\mathrm{vel}}\times2N_{\mathrm{vel}}}$ has diagonal entries $w_{u,k}^qw_{v,l}^q$, repeated for the two reduced distribution functions:
\[
 \mathbf\Omega=
 \begin{pmatrix}
  w_{u,0}^1w_{v,0}^1&&&&&0\\
  &\ddots&&&&\\
  &&w_{u,N_u^4-1}^4w_{v,N_v^4-1}^4&&&\\[3pt]
  &&&w_{u,0}^1w_{v,0}^1&&\\
  &&&&\ddots&\\
  0&&&&&w_{u,N_u^4-1}^4w_{v,N_v^4-1}^4
 \end{pmatrix}.
\]
Projecting the vector $\mathbf z^{\mathrm v}$ onto the discrete basis gives the coefficient vector $\hat{\mathbf z}$:
\[
 \hat{\mathbf z}=\mathbf P^{\mathsf T}\mathbf\Omega\mathbf z^{\mathrm v},
 \qquad
 \hat{\mathbf z}^{q,s}=\mathbf P_q^{\mathsf T}\mathbf\Omega_q\mathbf z^{\mathrm v,q,s}.
\]
Here $s=h,b$. The decomposition~\eqref{eq:cell-macro-micro-decomposition} becomes
\[
 \hat{\fvec}_i=\hat{\M}_i+\avec_i,
 \qquad \hat{\rvec}_i=\avec_i.
\]

\subsection{Conservative moments}\label{sec:modal-moments}
The $\psi$-moment of the nonequilibrium distribution function is
\begin{equation}
 \langle\rvec_i^{\mathrm v,q,s},\bm\psi^q\rangle_q
 =\sum_{(m_1,m_2)\in\mathcal V_q}a_{i,(m_1,m_2)}^{q,s}\ell_{(m_1,m_2)}^q(\psi),
 \qquad \ell_{(m_1,m_2)}^q(\psi)=\langle\bm\psi^q,\bm\phi_{(m_1,m_2)}^q\rangle_q.
 \label{eq:moment-rows}
\end{equation}
Here $\mathcal V_q$ is the degree set in Eq.~\eqref{eq:quadrant-degree-set}, and $\bm\psi^q\in\mathbb R^{N_q}$ has components $(\bm\psi^q)_\nu=\psi(u_k^q,v_l^q)$.

\begin{proposition}[Separation of discrete moments]\label{prop:modal-moment-factorization}
For $\psi(u,v)=\psi_u(u)\psi_v(v)$ and the column vectors
\[
 \begin{aligned}
 \mathbf z_u&=\bigl(\psi_u(u_0^q),\ldots,\psi_u(u_{N_u^q-1}^q)\bigr)^{\mathsf T}\in\mathbb R^{N_u^q},\\
 \mathbf z_v&=\bigl(\psi_v(v_0^q),\ldots,\psi_v(v_{N_v^q-1}^q)\bigr)^{\mathsf T}\in\mathbb R^{N_v^q},
 \end{aligned}
\]
the tensor identity $\bm\psi^q=\mathbf z_u\otimes\mathbf z_v$ gives
\begin{equation}
 \ell_{(m_1,m_2)}^q(\psi)
 =\langle\bm\psi^q,\bm\phi_{(m_1,m_2)}^q\rangle_q
 =\langle\mathbf z_u,\mathbf p_{u,m_1}^q\rangle_{w_u^q}
  \langle\mathbf z_v,\mathbf p_{v,m_2}^q\rangle_{w_v^q}.
 \label{eq:modal-moment-factorization}
\end{equation}
\end{proposition}
\begin{proof}
The tensor product basis and quadrature give
\[
 \begin{aligned}
 \langle\bm\psi^q,\bm\phi_{(m_1,m_2)}^q\rangle_q
 &=\sum_{k=0}^{N_u^q-1}\sum_{l=0}^{N_v^q-1}
   w_{u,k}^qw_{v,l}^q
   (\mathbf z_u)_k(\mathbf z_v)_l
   (\mathbf p_{u,m_1}^q)_k(\mathbf p_{v,m_2}^q)_l\\
 &=\left[\sum_{k=0}^{N_u^q-1}w_{u,k}^q(\mathbf z_u)_k(\mathbf p_{u,m_1}^q)_k\right]
   \left[\sum_{l=0}^{N_v^q-1}w_{v,l}^q(\mathbf z_v)_l(\mathbf p_{v,m_2}^q)_l\right],
 \end{aligned}
\]
which equals the right hand side of Eq.~\eqref{eq:modal-moment-factorization}.
\end{proof}

\begin{proposition}[Degree restrictions on polynomial moments]\label{prop:vanishing-polynomial-moments}
For nonnegative integers $r_1,r_2$ and $(m_1,m_2)\in\mathcal V_q$,
\begin{equation}
 \ell_{(m_1,m_2)}^q(u^{r_1}v^{r_2})=0
 \quad\text{if}\quad m_1>r_1\ \text{or}\ m_2>r_2.
 \label{eq:vanishing-polynomial-moments}
\end{equation}
\end{proposition}
\begin{proof}
For the vectors in Proposition~\ref{prop:modal-moment-factorization} with $\psi_u(u)=u^{r_1}$ and $\psi_v(v)=v^{r_2}$, expansion in the orthonormal bases gives~\cite{dlmf18_2,gautschi1999quadrature}
\[
 \begin{aligned}
 \mathbf z_u&=\sum_{j=0}^{r_1}c_j\mathbf p_{u,j}^q,
 &c_j&=\langle\mathbf z_u,\mathbf p_{u,j}^q\rangle_{w_u^q},
 &&m_1>r_1,\\
 \mathbf z_v&=\sum_{j=0}^{r_2}d_j\mathbf p_{v,j}^q,
 &d_j&=\langle\mathbf z_v,\mathbf p_{v,j}^q\rangle_{w_v^q},
 &&m_2>r_2.
 \end{aligned}
\]
Orthonormality yields
\begin{equation}
 \begin{aligned}
 \langle\mathbf z_u,\mathbf p_{u,m_1}^q\rangle_{w_u^q}
 &=\sum_{j=0}^{r_1}c_j\langle\mathbf p_{u,j}^q,\mathbf p_{u,m_1}^q\rangle_{w_u^q}
 =\sum_{j=0}^{r_1}c_j\delta_{jm_1}=0,
 &&m_1>r_1,\\
 \langle\mathbf z_v,\mathbf p_{v,m_2}^q\rangle_{w_v^q}
 &=\sum_{j=0}^{r_2}d_j\langle\mathbf p_{v,j}^q,\mathbf p_{v,m_2}^q\rangle_{w_v^q}
 =\sum_{j=0}^{r_2}d_j\delta_{jm_2}=0,
 &&m_2>r_2.
 \end{aligned}
 \label{eq:one-dimensional-vanishing-moment}
\end{equation}
Equation~\eqref{eq:modal-moment-factorization} then gives Eq.~\eqref{eq:vanishing-polynomial-moments}.
\end{proof}

\begin{proposition}[Conservative moment representation]\label{prop:conservative-moment-representation}
The conservative moments of the nonequilibrium expansion~\eqref{eq:nonequilibrium-expansion} are $\Cmat\avec_i$, where $\Cmat\in\mathbb R^{4\times2N_{\mathrm{vel}}}$ and
\begin{equation}
 \begin{aligned}
 \Cmat\avec_i
 &=\sum_{q=1}^4
 \begin{pmatrix}
 a_{i,(0,0)}^{q,h}\ell_{(0,0)}^q(1)\\[4pt]
 \displaystyle\sum_{m=0}^1 a_{i,(m,0)}^{q,h}\ell_{(m,0)}^q(u)\\[4pt]
 \displaystyle\sum_{n=0}^1 a_{i,(0,n)}^{q,h}\ell_{(0,n)}^q(v)\\[4pt]
 \tfrac12\Bigl[\displaystyle\sum_{m=0}^2 a_{i,(m,0)}^{q,h}\ell_{(m,0)}^q(u^2)
 +\sum_{n=0}^2 a_{i,(0,n)}^{q,h}\ell_{(0,n)}^q(v^2)
 +a_{i,(0,0)}^{q,b}\ell_{(0,0)}^q(1)\Bigr]
 \end{pmatrix}.
 \end{aligned}
 \label{eq:modal-conservative-moments}
\end{equation}
Only the degrees $(0,0),(1,0),(0,1),(2,0),(0,2)$ contribute for $h$, and only $(0,0)$ contributes for $b$.
\end{proposition}
\begin{proof}
The discrete form of the reduced moments~\eqref{eq:reduced-moments} is
\[
 \Cmat\avec_i
 =\sum_{q=1}^4\sum_{\uvec_\nu\in Q_q}w_\nu
 \begin{pmatrix}
 r_{i,\nu}^{\mathrm v,q,h}\\
 u_\nu r_{i,\nu}^{\mathrm v,q,h}\\
 v_\nu r_{i,\nu}^{\mathrm v,q,h}\\
 \tfrac12\bigl[(u_\nu^2+v_\nu^2)r_{i,\nu}^{\mathrm v,q,h}+r_{i,\nu}^{\mathrm v,q,b}\bigr]
 \end{pmatrix},
\]
where $r_{i,\nu}^{\mathrm v,q,s}$ is the $\nu$th entry of $\rvec_i^{\mathrm v,q,s}$, for $s=h,b$. Substitution of Eq.~\eqref{eq:nonequilibrium-expansion} into a moment with weight $\psi$ yields
\[
 \begin{aligned}
 \sum_{\uvec_\nu\in Q_q}w_\nu\psi(u_\nu,v_\nu)r_{i,\nu}^{\mathrm v,q,s}
 &=\sum_{\uvec_\nu\in Q_q}w_\nu\psi(u_\nu,v_\nu)
   \sum_{(m_1,m_2)\in\mathcal V_q}a_{i,(m_1,m_2)}^{q,s}(\bm\phi_{(m_1,m_2)}^q)_\nu\\
 &=\sum_{(m_1,m_2)\in\mathcal V_q}a_{i,(m_1,m_2)}^{q,s}
   \sum_{\uvec_\nu\in Q_q}w_\nu\psi(u_\nu,v_\nu)(\bm\phi_{(m_1,m_2)}^q)_\nu\\
 &=\sum_{(m_1,m_2)\in\mathcal V_q}a_{i,(m_1,m_2)}^{q,s}\ell_{(m_1,m_2)}^q(\psi).
 \end{aligned}
\]
For nonnegative integers $r_1,r_2$, Proposition~\ref{prop:vanishing-polynomial-moments} gives
\[
 \sum_{(m_1,m_2)\in\mathcal V_q}a_{i,(m_1,m_2)}^{q,s}\ell_{(m_1,m_2)}^q(u^{r_1}v^{r_2})
 =\sum_{\substack{(m_1,m_2)\in\mathcal V_q\\m_1\le r_1,\ m_2\le r_2}}
 a_{i,(m_1,m_2)}^{q,s}\ell_{(m_1,m_2)}^q(u^{r_1}v^{r_2}).
\]
Applying this reduction to the four moment rows above yields Eq.~\eqref{eq:modal-conservative-moments}.
\end{proof}
All remaining basis vectors have zero conservative moments, so their coefficients do not affect $\Cmat\avec_i$.

\section{UGKS evolution of the nonequilibrium coefficients}\label{sec:updates}

\subsection{Projection of the interface flux}\label{sec:integral-flux-projection}
The integral solution~\eqref{eq:integral-solution} gives the UGKS flux~\eqref{eq:reference-flux}. We express this flux in the discrete basis for each quadrant $Q_q$ and reduced distribution function $s\in\{h,b\}$.

For an interior Cartesian face $ij$ shared by cells $\Omega_i$ and $\Omega_j$, let $\bm n_{ij}=(n_{ij,x},n_{ij,y})^{\mathsf T}$ be the unit normal directed from cell $i$ toward cell $j$. Writing the nodes of $Q_q$ as $(u_\nu,v_\nu)$, $\nu=0,\ldots,N_q-1$, the velocity matrices are
\[
 \mathbf D_u^q=\operatorname{diag}(u_0,\ldots,u_{N_q-1}),
 \qquad
 \mathbf D_v^q=\operatorname{diag}(v_0,\ldots,v_{N_q-1}).
\]
For $u_{ij,\nu}=n_{ij,x}u_\nu+n_{ij,y}v_\nu$, the normal velocity matrix is
\[
 \mathbf D_{ij}^q
 =\operatorname{diag}(u_{ij,0},\ldots,u_{ij,N_q-1})
 =n_{ij,x}\mathbf D_u^q+n_{ij,y}\mathbf D_v^q.
\]
The nodal vector form of Eq.~\eqref{eq:reference-flux} is
\begin{equation}
 \Flux_{ij}^{\mathrm v,q,s}
 =\mathbf D_{ij}^q\left[
 \begin{aligned}
 &\bar c_1\M_0^{\mathrm v,q,s}
  +\bar c_2(\mathbf D_u^q\M_x^{\mathrm v,q,s}+\mathbf D_v^q\M_y^{\mathrm v,q,s})
  +\bar c_3\M_t^{\mathrm v,q,s}\\
 &+\bar c_4\fvec_0^{\mathrm v,q,s}
  +\bar c_5(\mathbf D_u^q\fvec_x^{\mathrm v,q,s}+\mathbf D_v^q\fvec_y^{\mathrm v,q,s})
 \end{aligned}\right].
 \label{eq:nodal-integral-flux}
\end{equation}
Since $\mathbf P_q^{-1}=\mathbf P_q^{\mathsf T}\mathbf\Omega_q$, left multiplication by $\mathbf P_q^{-1}$ gives
\begin{equation}
 \hat{\Flux}_{ij}^{q,s}
 =\mathbf P_q^{-1}\Flux_{ij}^{\mathrm v,q,s}
 =\mathbf P_q^{-1}\mathbf D_{ij}^q\left[
 \begin{aligned}
 &\bar c_1\M_0^{\mathrm v,q,s}
  +\bar c_2(\mathbf D_u^q\M_x^{\mathrm v,q,s}+\mathbf D_v^q\M_y^{\mathrm v,q,s})
  +\bar c_3\M_t^{\mathrm v,q,s}\\
 &+\bar c_4\fvec_0^{\mathrm v,q,s}
  +\bar c_5(\mathbf D_u^q\fvec_x^{\mathrm v,q,s}+\mathbf D_v^q\fvec_y^{\mathrm v,q,s})
 \end{aligned}\right].
 \label{eq:projected-integral-flux}
\end{equation}
Substituting the discrete basis representations
\[
 \begin{aligned}
 \M_\alpha^{\mathrm v,q,s}&=\mathbf P_q\hat{\M}_\alpha^{q,s},
 \qquad \alpha=0,x,y,t,\\
 \fvec_\beta^{\mathrm v,q,s}&=\mathbf P_q\hat{\fvec}_\beta^{q,s},
 \qquad \beta=0,x,y,
 \end{aligned}
\]
into Eq.~\eqref{eq:projected-integral-flux} yields
\begin{equation}
 \hat{\Flux}_{ij}^{q,s}
 =\mathbf P_q^{-1}\mathbf D_{ij}^q\left[
 \begin{aligned}
 &\bar c_1\mathbf P_q\hat{\M}_0^{q,s}
  +\bar c_2(\mathbf D_u^q\mathbf P_q\hat{\M}_x^{q,s}+\mathbf D_v^q\mathbf P_q\hat{\M}_y^{q,s})
  +\bar c_3\mathbf P_q\hat{\M}_t^{q,s}\\
 &+\bar c_4\mathbf P_q\hat{\fvec}_0^{q,s}
  +\bar c_5(\mathbf D_u^q\mathbf P_q\hat{\fvec}_x^{q,s}+\mathbf D_v^q\mathbf P_q\hat{\fvec}_y^{q,s})
 \end{aligned}\right].
 \label{eq:expanded-polynomial-flux}
\end{equation}
With the definition
\begin{equation}
 \mathbf J_\eta^q
 =\mathbf P_q^{-1}\mathbf D_\eta^q\mathbf P_q
 =\mathbf P_q^{\mathsf T}\mathbf\Omega_q\mathbf D_\eta^q\mathbf P_q,
 \qquad \eta=u,v,ij,
 \label{eq:velocity-matrices}
\end{equation}
the identity $\mathbf P_q\mathbf P_q^{\mathsf T}\mathbf\Omega_q=\mathbf I_{N_q}$ gives
\[
 \mathbf P_q^{\mathsf T}\mathbf\Omega_q\mathbf D_{ij}^q\mathbf D_\xi^q\mathbf P_q
 =\underbrace{(\mathbf P_q^{\mathsf T}\mathbf\Omega_q\mathbf D_{ij}^q\mathbf P_q)}_{\mathbf J_{ij}^q}
  \underbrace{(\mathbf P_q^{\mathsf T}\mathbf\Omega_q\mathbf D_\xi^q\mathbf P_q)}_{\mathbf J_\xi^q}
 =\mathbf J_{ij}^q\mathbf J_\xi^q,
 \qquad \xi=u,v.
\]
Applying these identities to each term in Eq.~\eqref{eq:expanded-polynomial-flux} yields
\begin{equation}
 \begin{aligned}
 \hat{\Flux}_{ij}^{q,s}
 ={}&\bar c_1\mathbf J_{ij}^q\hat{\M}_0^{q,s}
  +\bar c_2\mathbf J_{ij}^q\mathbf J_u^q\hat{\M}_x^{q,s}
  +\bar c_2\mathbf J_{ij}^q\mathbf J_v^q\hat{\M}_y^{q,s}
  +\bar c_3\mathbf J_{ij}^q\hat{\M}_t^{q,s}\\
 &+\bar c_4\mathbf J_{ij}^q\hat{\fvec}_0^{q,s}
  +\bar c_5\mathbf J_{ij}^q\mathbf J_u^q\hat{\fvec}_x^{q,s}
  +\bar c_5\mathbf J_{ij}^q\mathbf J_v^q\hat{\fvec}_y^{q,s}.
 \end{aligned}
 \label{eq:termwise-polynomial-flux}
\end{equation}
Using the common normal velocity operator $\mathbf J_{ij}^q$, this expression becomes
\begin{equation}
 \hat{\Flux}_{ij}^{q,s}
 =\mathbf J_{ij}^q\left[
 \begin{aligned}
 &\bar c_1\hat{\M}_0^{q,s}
  +\bar c_2(\mathbf J_u^q\hat{\M}_x^{q,s}+\mathbf J_v^q\hat{\M}_y^{q,s})
  +\bar c_3\hat{\M}_t^{q,s}\\
 &+\bar c_4\hat{\fvec}_0^{q,s}
  +\bar c_5(\mathbf J_u^q\hat{\fvec}_x^{q,s}+\mathbf J_v^q\hat{\fvec}_y^{q,s})
 \end{aligned}\right].
 \label{eq:complete-polynomial-flux}
\end{equation}

For row and column indices $(m_1,m_2),(n_1,n_2)\in\mathcal V_q$, the matrix entries are
\begin{equation}
 (\mathbf J_\xi^q)_{(m_1,m_2),(n_1,n_2)}
 =\langle\mathbf D_\xi^q\bm\phi_{(n_1,n_2)}^q,\bm\phi_{(m_1,m_2)}^q\rangle_q,
 \qquad \xi=u,v.
 \label{eq:quadrant-velocity-projection}
\end{equation}

The upwind side in Eq.~\eqref{eq:interface-half-space-state} is fixed on each quadrant. The reconstructed coefficients have the decomposition
\begin{equation}
 \begin{aligned}
 \hat{\fvec}_0^{q,s}
 &=\hat{\M}^{q,s,\mathrm{up}}+\avec_0^{q,s},\\
 \hat{\fvec}_\zeta^{q,s}
 &=\hat{\M}_\zeta^{q,s,\mathrm{up}}+\avec_\zeta^{q,s},
 \qquad \zeta=x,y.
 \end{aligned}
 \label{eq:reconstructed-split}
\end{equation}
Here $\avec_0^{q,s}$ is the reconstructed nonequilibrium coefficient vector, and $\avec_x^{q,s},\avec_y^{q,s}$ are its spatial slopes. With $\M^{\mathrm{up}}$ denoting the upwind cell equilibrium and $\M_0=\M(\W_0)$ the interface collision equilibrium, Eq.~\eqref{eq:complete-polynomial-flux} becomes
\begin{equation}
 \hat{\Flux}_{ij}^{q,s}
 =\hat{\Flux}_{ij,\mathrm{neq}}^{q,s}
  +\hat{\Flux}_{ij,\mathrm{eq}}^{q,s},
 \label{eq:full-modal-flux}
\end{equation}
where
\begin{equation}
 \hat{\Flux}_{ij,\mathrm{neq}}^{q,s}
 =\mathbf J_{ij}^q\bigl[
     \bar c_4\avec_0^{q,s}
     +\bar c_5(\mathbf J_u^q\avec_x^{q,s}+\mathbf J_v^q\avec_y^{q,s})\bigr],
 \label{eq:nonequilibrium-modal-flux}
\end{equation}
\begin{equation}
 \hat{\Flux}_{ij,\mathrm{eq}}^{q,s}
 =\mathbf J_{ij}^q\left[
 \begin{aligned}
 &\bar c_1\hat{\M}_0^{q,s}
  +\bar c_2(\mathbf J_u^q\hat{\M}_x^{q,s}+\mathbf J_v^q\hat{\M}_y^{q,s})
  +\bar c_3\hat{\M}_t^{q,s}\\
 &+\bar c_4\hat{\M}^{q,s,\mathrm{up}}
  +\bar c_5(\mathbf J_u^q\hat{\M}_x^{q,s,\mathrm{up}}+\mathbf J_v^q\hat{\M}_y^{q,s,\mathrm{up}})
 \end{aligned}\right].
 \label{eq:equilibrium-modal-flux}
\end{equation}

The $\bar c_4,\bar c_5$ terms in Eq.~\eqref{eq:equilibrium-modal-flux} give the free transport flux of the Maxwellian:
\begin{equation}
 \hat{\Flux}_{ij,\mathrm{M,fr}}^{q,s}
 =\mathbf J_{ij}^q\bigl[
   \bar c_4\hat{\M}^{q,s,\mathrm{up}}
   +\bar c_5(\mathbf J_u^q\hat{\M}_x^{q,s,\mathrm{up}}
             +\mathbf J_v^q\hat{\M}_y^{q,s,\mathrm{up}})\bigr],
 \label{eq:maxwellian-free-transport}
\end{equation}
where the free transport of the equilibrium distribution is expressed analytically, as in the UGKWP wave flux $\mathbf F^{\mathrm{fr,wave}}$~\cite{zhu2019ugkwp2}.

The properties of $\mathbf J_u^q,\mathbf J_v^q$ and the evaluation of $\mathbf J_u^q\avec^{q,s},\mathbf J_v^q\avec^{q,s}$ are detailed in Section~\ref{sec:modal-velocity}.

The equilibrium flux in Eq.~\eqref{eq:equilibrium-modal-flux} is evaluated efficiently using the separable Maxwellian
\begin{equation}
 \begin{aligned}
 M_s&=C_sX(u)Y(v),\qquad C_h=\rho,\quad C_b=K\rho\theta,\\
 X(u)&=\frac{e^{-(u-U)^2/(2\theta)}}{\sqrt{2\pi\theta}},
 \qquad
 Y(v)=\frac{e^{-(v-V)^2/(2\theta)}}{\sqrt{2\pi\theta}}.
 \end{aligned}
 \label{eq:equilibrium-factors}
\end{equation}
Applying Proposition~\ref{prop:modal-moment-factorization} to the Maxwellian gives
\begin{equation}
 \begin{aligned}
 \hat M_{(m_1,m_2)}^{q,s}
 &=\langle M_s,\bm\phi_{(m_1,m_2)}^q\rangle_q
  =C_s\langle XY,\bm\phi_{(m_1,m_2)}^q\rangle_q\\
 &=C_s\langle X,\mathbf p_{u,m_1}^q\rangle_{w_u^q}
       \langle Y,\mathbf p_{v,m_2}^q\rangle_{w_v^q}
  =C_s\hat X_{m_1}^q\hat Y_{m_2}^q.
 \end{aligned}
 \label{eq:separable-coefficients}
\end{equation}
For nonnegative integers $r$, define the one dimensional projections
\begin{equation}
 \hat X_m^{q,(r)}=\langle u^rX,\mathbf p_{u,m}^q\rangle_{w_u^q},
 \qquad
 \hat Y_n^{q,(r)}=\langle v^rY,\mathbf p_{v,n}^q\rangle_{w_v^q}.
 \label{eq:equilibrium-axis-projections}
\end{equation}
Here $\hat X_m^q=\hat X_m^{q,(0)}$ and $\hat Y_n^q=\hat Y_n^{q,(0)}$. Proposition~\ref{prop:modal-moment-factorization} also gives the velocity weighted Maxwellian projection
\begin{equation}
 \begin{aligned}
 \langle u^{r_1}v^{r_2}M_s,\bm\phi_{(m_1,m_2)}^q\rangle_q
 &=C_s\langle (u^{r_1}X)(v^{r_2}Y),\bm\phi_{(m_1,m_2)}^q\rangle_q\\
 &=C_s\langle u^{r_1}X,\mathbf p_{u,m_1}^q\rangle_{w_u^q}
       \langle v^{r_2}Y,\mathbf p_{v,m_2}^q\rangle_{w_v^q}
  =C_s\hat X_{m_1}^{q,(r_1)}\hat Y_{m_2}^{q,(r_2)}.
 \end{aligned}
 \label{eq:weighted-equilibrium-coefficients}
\end{equation}
Since the velocity grid and basis are fixed, $\hat{\M}_\zeta^{q,s}=\partial_\zeta\hat{\M}^{q,s}$ for $\zeta=x,y,t$, and
\begin{equation}
 \partial_\zeta\hat M_{(m_1,m_2)}^{q,s}
 =C_{s,\zeta}\hat X_{m_1}^q\hat Y_{m_2}^q
  +C_s(\partial_\zeta\hat X_{m_1}^q)\hat Y_{m_2}^q
  +C_s\hat X_{m_1}^q(\partial_\zeta\hat Y_{m_2}^q),
 \label{eq:equilibrium-derivative-projections}
\end{equation}
where subscripts $x,y,t$ on the macroscopic variables denote their derivatives, and $C_{h,\zeta}=\rho_\zeta$, $C_{b,\zeta}=K(\theta\rho_\zeta+\rho\theta_\zeta)$. The Gaussian derivatives are
\begin{equation}
 \begin{aligned}
 \partial_\zeta X&=X\left[\frac{u-U}{\theta}U_\zeta+\frac{(u-U)^2-\theta}{2\theta^2}\theta_\zeta\right],\\
 \partial_\zeta Y&=Y\left[\frac{v-V}{\theta}V_\zeta+\frac{(v-V)^2-\theta}{2\theta^2}\theta_\zeta\right],
 \qquad \zeta=x,y,t.
 \end{aligned}
 \label{eq:equilibrium-axis-derivatives}
\end{equation}
Substitution into $\partial_\zeta M_s=C_{s,\zeta}XY+C_s(\partial_\zeta X)Y+C_sX(\partial_\zeta Y)$ gives the explicit expansion
\begin{equation}
 \begin{aligned}
 \partial_\zeta M_s
 =XY\Bigl[&C_{s,\zeta}-\frac{C_s}{\theta}(UU_\zeta+VV_\zeta)
  +\frac{C_s\theta_\zeta}{2\theta^2}(U^2+V^2-2\theta)\\
 &+\frac{C_s}{\theta}\left(U_\zeta-\frac{U\theta_\zeta}{\theta}\right)u
  +\frac{C_s}{\theta}\left(V_\zeta-\frac{V\theta_\zeta}{\theta}\right)v
  +\frac{C_s\theta_\zeta}{2\theta^2}(u^2+v^2)\Bigr].
 \end{aligned}
 \label{eq:equilibrium-derivative-polynomial}
\end{equation}
For the Maxwellian and its derivatives, Eqs.~\eqref{eq:velocity-matrices} and~\eqref{eq:weighted-equilibrium-coefficients} give, for nonnegative integers $r_1,r_2$ and $\zeta=x,y,t$,
\[
 \begin{aligned}
 \bigl((\mathbf J_u^q)^{r_1}(\mathbf J_v^q)^{r_2}\hat{\M}^{q,s}\bigr)_{(m_1,m_2)}
 &=\langle u^{r_1}v^{r_2}M_s,\bm\phi_{(m_1,m_2)}^q\rangle_q
  =C_s\hat X_{m_1}^{q,(r_1)}\hat Y_{m_2}^{q,(r_2)},\\
 \bigl((\mathbf J_u^q)^{r_1}(\mathbf J_v^q)^{r_2}\hat{\M}_\zeta^{q,s}\bigr)_{(m_1,m_2)}
 &=\langle u^{r_1}v^{r_2}\partial_\zeta M_s,\bm\phi_{(m_1,m_2)}^q\rangle_q.
 \end{aligned}
\]
Together with the derivative expansion~\eqref{eq:equilibrium-derivative-polynomial} and $\mathbf J_{ij}^q=n_{ij,x}\mathbf J_u^q+n_{ij,y}\mathbf J_v^q$, these identities give
\[
 (\hat{\Flux}_{ij,\mathrm{eq}}^{q,s})_{(m_1,m_2)}
 =n_{ij,x}\hat F_{u,\mathrm{eq},(m_1,m_2)}^{q,s}
  +n_{ij,y}\hat F_{v,\mathrm{eq},(m_1,m_2)}^{q,s}.
\]
Write $\hat X^{(r)}=\hat X_{m_1}^{q,(r)}$ and $\hat Y^{(r)}=\hat Y_{m_2}^{q,(r)}$, with superscript $(0)$ omitted. Subscripts $0$ and $\mathrm{up}$ identify the interface and reconstructed upwind states.

The $u$-direction equilibrium flux is
\begingroup
\allowdisplaybreaks[2]
\begin{align}
 \hat F_{u,\mathrm{eq},(m_1,m_2)}^{q,s}={}&\bar c_1\bigl[C_s\hat X^{(1)}\hat Y\bigr]_0\notag\\
 &+\bar c_2\left[\begin{aligned}
 &\left(C_{s,x}-\frac{C_s}{\theta}(UU_x+VV_x)+\frac{C_s\theta_x}{2\theta^2}(U^2+V^2-2\theta)\right)\hat X^{(2)}\hat Y\\
 &+\frac{C_s}{\theta}\left(U_x-\frac{U\theta_x}{\theta}\right)\hat X^{(3)}\hat Y+\frac{C_s}{\theta}\left(V_x-\frac{V\theta_x}{\theta}\right)\hat X^{(2)}\hat Y^{(1)}\\
 &+\frac{C_s\theta_x}{2\theta^2}\bigl(\hat X^{(4)}\hat Y+\hat X^{(2)}\hat Y^{(2)}\bigr)\end{aligned}\right]_{0}\notag\\
 &+\bar c_2\left[\begin{aligned}
 &\left(C_{s,y}-\frac{C_s}{\theta}(UU_y+VV_y)+\frac{C_s\theta_y}{2\theta^2}(U^2+V^2-2\theta)\right)\hat X^{(1)}\hat Y^{(1)}\\
 &+\frac{C_s}{\theta}\left(U_y-\frac{U\theta_y}{\theta}\right)\hat X^{(2)}\hat Y^{(1)}+\frac{C_s}{\theta}\left(V_y-\frac{V\theta_y}{\theta}\right)\hat X^{(1)}\hat Y^{(2)}\\
 &+\frac{C_s\theta_y}{2\theta^2}\bigl(\hat X^{(3)}\hat Y^{(1)}+\hat X^{(1)}\hat Y^{(3)}\bigr)\end{aligned}\right]_{0}\notag\\
 &+\bar c_3\left[\begin{aligned}
 &\left(C_{s,t}-\frac{C_s}{\theta}(UU_t+VV_t)+\frac{C_s\theta_t}{2\theta^2}(U^2+V^2-2\theta)\right)\hat X^{(1)}\hat Y\\
 &+\frac{C_s}{\theta}\left(U_t-\frac{U\theta_t}{\theta}\right)\hat X^{(2)}\hat Y+\frac{C_s}{\theta}\left(V_t-\frac{V\theta_t}{\theta}\right)\hat X^{(1)}\hat Y^{(1)}\\
 &+\frac{C_s\theta_t}{2\theta^2}\bigl(\hat X^{(3)}\hat Y+\hat X^{(1)}\hat Y^{(2)}\bigr)\end{aligned}\right]_{0}\notag\\
 &+\bar c_4\bigl[C_s\hat X^{(1)}\hat Y\bigr]_{\mathrm{up}}\notag\\
 &+\bar c_5\left[\begin{aligned}
 &\left(C_{s,x}-\frac{C_s}{\theta}(UU_x+VV_x)+\frac{C_s\theta_x}{2\theta^2}(U^2+V^2-2\theta)\right)\hat X^{(2)}\hat Y\\
 &+\frac{C_s}{\theta}\left(U_x-\frac{U\theta_x}{\theta}\right)\hat X^{(3)}\hat Y+\frac{C_s}{\theta}\left(V_x-\frac{V\theta_x}{\theta}\right)\hat X^{(2)}\hat Y^{(1)}\\
 &+\frac{C_s\theta_x}{2\theta^2}\bigl(\hat X^{(4)}\hat Y+\hat X^{(2)}\hat Y^{(2)}\bigr)\end{aligned}\right]_{\mathrm{up}}\notag\\
 &+\bar c_5\left[\begin{aligned}
 &\left(C_{s,y}-\frac{C_s}{\theta}(UU_y+VV_y)+\frac{C_s\theta_y}{2\theta^2}(U^2+V^2-2\theta)\right)\hat X^{(1)}\hat Y^{(1)}\\
 &+\frac{C_s}{\theta}\left(U_y-\frac{U\theta_y}{\theta}\right)\hat X^{(2)}\hat Y^{(1)}+\frac{C_s}{\theta}\left(V_y-\frac{V\theta_y}{\theta}\right)\hat X^{(1)}\hat Y^{(2)}\\
 &+\frac{C_s\theta_y}{2\theta^2}\bigl(\hat X^{(3)}\hat Y^{(1)}+\hat X^{(1)}\hat Y^{(3)}\bigr)\end{aligned}\right]_{\mathrm{up}}.
 \label{eq:equilibrium-u-flux-products}
\end{align}
\endgroup
The $v$-direction equilibrium flux is
\begingroup
\allowdisplaybreaks[2]
\begin{align}
 \hat F_{v,\mathrm{eq},(m_1,m_2)}^{q,s}={}&\bar c_1\bigl[C_s\hat X\hat Y^{(1)}\bigr]_0\notag\\*
 &+\bar c_2\left[\begin{aligned}
 &\left(C_{s,x}-\frac{C_s}{\theta}(UU_x+VV_x)+\frac{C_s\theta_x}{2\theta^2}(U^2+V^2-2\theta)\right)\hat X^{(1)}\hat Y^{(1)}\\
 &+\frac{C_s}{\theta}\left(U_x-\frac{U\theta_x}{\theta}\right)\hat X^{(2)}\hat Y^{(1)}+\frac{C_s}{\theta}\left(V_x-\frac{V\theta_x}{\theta}\right)\hat X^{(1)}\hat Y^{(2)}\\
 &+\frac{C_s\theta_x}{2\theta^2}\bigl(\hat X^{(3)}\hat Y^{(1)}+\hat X^{(1)}\hat Y^{(3)}\bigr)\end{aligned}\right]_{0}\notag\\
 &+\bar c_2\left[\begin{aligned}
 &\left(C_{s,y}-\frac{C_s}{\theta}(UU_y+VV_y)+\frac{C_s\theta_y}{2\theta^2}(U^2+V^2-2\theta)\right)\hat X\hat Y^{(2)}\\
 &+\frac{C_s}{\theta}\left(U_y-\frac{U\theta_y}{\theta}\right)\hat X^{(1)}\hat Y^{(2)}+\frac{C_s}{\theta}\left(V_y-\frac{V\theta_y}{\theta}\right)\hat X\hat Y^{(3)}\\
 &+\frac{C_s\theta_y}{2\theta^2}\bigl(\hat X^{(2)}\hat Y^{(2)}+\hat X\hat Y^{(4)}\bigr)\end{aligned}\right]_{0}\notag\\
 &+\bar c_3\left[\begin{aligned}
 &\left(C_{s,t}-\frac{C_s}{\theta}(UU_t+VV_t)+\frac{C_s\theta_t}{2\theta^2}(U^2+V^2-2\theta)\right)\hat X\hat Y^{(1)}\\
 &+\frac{C_s}{\theta}\left(U_t-\frac{U\theta_t}{\theta}\right)\hat X^{(1)}\hat Y^{(1)}+\frac{C_s}{\theta}\left(V_t-\frac{V\theta_t}{\theta}\right)\hat X\hat Y^{(2)}\\
 &+\frac{C_s\theta_t}{2\theta^2}\bigl(\hat X^{(2)}\hat Y^{(1)}+\hat X\hat Y^{(3)}\bigr)\end{aligned}\right]_{0}\notag\\
 &+\bar c_4\bigl[C_s\hat X\hat Y^{(1)}\bigr]_{\mathrm{up}}\notag\\
 &+\bar c_5\left[\begin{aligned}
 &\left(C_{s,x}-\frac{C_s}{\theta}(UU_x+VV_x)+\frac{C_s\theta_x}{2\theta^2}(U^2+V^2-2\theta)\right)\hat X^{(1)}\hat Y^{(1)}\\
 &+\frac{C_s}{\theta}\left(U_x-\frac{U\theta_x}{\theta}\right)\hat X^{(2)}\hat Y^{(1)}+\frac{C_s}{\theta}\left(V_x-\frac{V\theta_x}{\theta}\right)\hat X^{(1)}\hat Y^{(2)}\\
 &+\frac{C_s\theta_x}{2\theta^2}\bigl(\hat X^{(3)}\hat Y^{(1)}+\hat X^{(1)}\hat Y^{(3)}\bigr)\end{aligned}\right]_{\mathrm{up}}\notag\\
 &+\bar c_5\left[\begin{aligned}
 &\left(C_{s,y}-\frac{C_s}{\theta}(UU_y+VV_y)+\frac{C_s\theta_y}{2\theta^2}(U^2+V^2-2\theta)\right)\hat X\hat Y^{(2)}\\
 &+\frac{C_s}{\theta}\left(U_y-\frac{U\theta_y}{\theta}\right)\hat X^{(1)}\hat Y^{(2)}+\frac{C_s}{\theta}\left(V_y-\frac{V\theta_y}{\theta}\right)\hat X\hat Y^{(3)}\\
 &+\frac{C_s\theta_y}{2\theta^2}\bigl(\hat X^{(2)}\hat Y^{(2)}+\hat X\hat Y^{(4)}\bigr)\end{aligned}\right]_{\mathrm{up}}.
 \label{eq:equilibrium-v-flux-products}
\end{align}
\endgroup
Combining Eqs.~\eqref{eq:equilibrium-u-flux-products} and~\eqref{eq:equilibrium-v-flux-products} with Eq.~\eqref{eq:nonequilibrium-modal-flux} gives the final microscopic flux:
\begin{equation}
 \begin{aligned}
 \hat F_{ij,(m_1,m_2)}^{q,s}
 ={}&\Bigl(\mathbf J_{ij}^q[\bar c_4\avec_0^{q,s}+\bar c_5(\mathbf J_u^q\avec_x^{q,s}+\mathbf J_v^q\avec_y^{q,s})]\Bigr)_{(m_1,m_2)}\\
 &+n_{ij,x}\hat F_{u,\mathrm{eq},(m_1,m_2)}^{q,s}
  +n_{ij,y}\hat F_{v,\mathrm{eq},(m_1,m_2)}^{q,s}.
 \end{aligned}
 \label{eq:explicit-equilibrium-normal-flux}
\end{equation}
\begin{samepage}
\begin{corollary}[Macroscopic interface flux]
The conservative moments of the microscopic flux~\eqref{eq:explicit-equilibrium-normal-flux} give\label{cor:macroscopic-interface-flux}
\begin{equation}
  \mFlux_{ij}^n=\Cmat\hat{\Flux}_{ij}^n
 =\sum_{q=1}^4
 \begin{pmatrix}
 \hat F_{ij,(0,0)}^{n,q,h}\ell_{(0,0)}^q(1)\\[4pt]
 \displaystyle\sum_{m=0}^1 \hat F_{ij,(m,0)}^{n,q,h}\ell_{(m,0)}^q(u)\\[4pt]
 \displaystyle\sum_{d=0}^1 \hat F_{ij,(0,d)}^{n,q,h}\ell_{(0,d)}^q(v)\\[4pt]
 \tfrac12\Bigl[\displaystyle\sum_{m=0}^2 \hat F_{ij,(m,0)}^{n,q,h}\ell_{(m,0)}^q(u^2)
 +\sum_{d=0}^2 \hat F_{ij,(0,d)}^{n,q,h}\ell_{(0,d)}^q(v^2)
 +\hat F_{ij,(0,0)}^{n,q,b}\ell_{(0,0)}^q(1)\Bigr]
 \end{pmatrix}.
 \label{eq:modal-macroscopic-flux}
\end{equation}
Thus the macroscopic flux requires only degrees $(0,0),(1,0),(0,1),(2,0),(0,2)$ of its $h$ component and degree $(0,0)$ of its $b$ component.
\end{corollary}
\end{samepage}
\begin{proof}
The microscopic flux has the expansion
\[
 F_{ij,\nu}^{\mathrm v,n,q,s}
 =\sum_{(m_1,m_2)\in\mathcal V_q}\hat F_{ij,(m_1,m_2)}^{n,q,s}(\bm\phi_{(m_1,m_2)}^q)_\nu.
\]
Its discrete $\psi$-moment on $Q_q$ is therefore
\[
 \begin{aligned}
 \sum_{\uvec_\nu\in Q_q}w_\nu\psi(u_\nu,v_\nu)F_{ij,\nu}^{\mathrm v,n,q,s}
 &=\sum_{\uvec_\nu\in Q_q}w_\nu\psi(u_\nu,v_\nu)
   \sum_{(m_1,m_2)\in\mathcal V_q}\hat F_{ij,(m_1,m_2)}^{n,q,s}(\bm\phi_{(m_1,m_2)}^q)_\nu\\
 &=\sum_{(m_1,m_2)\in\mathcal V_q}\hat F_{ij,(m_1,m_2)}^{n,q,s}
   \sum_{\uvec_\nu\in Q_q}w_\nu\psi(u_\nu,v_\nu)(\bm\phi_{(m_1,m_2)}^q)_\nu\\
 &=\sum_{(m_1,m_2)\in\mathcal V_q}\hat F_{ij,(m_1,m_2)}^{n,q,s}\ell_{(m_1,m_2)}^q(\psi).
 \end{aligned}
\]
Applying this identity to the mass, two momentum, and energy moments in Eq.~\eqref{eq:macroscopic-face-flux} gives
\[
 \mFlux_{ij}^n=\Cmat\hat{\Flux}_{ij}^n
 =\sum_{q=1}^4\sum_{(m_1,m_2)\in\mathcal V_q}
 \begin{pmatrix}
 \hat F_{ij,(m_1,m_2)}^{n,q,h}\ell_{(m_1,m_2)}^q(1)\\[4pt]
 \hat F_{ij,(m_1,m_2)}^{n,q,h}\ell_{(m_1,m_2)}^q(u)\\[4pt]
 \hat F_{ij,(m_1,m_2)}^{n,q,h}\ell_{(m_1,m_2)}^q(v)\\[4pt]
 \tfrac12\hat F_{ij,(m_1,m_2)}^{n,q,h}\ell_{(m_1,m_2)}^q(u^2+v^2)
 +\tfrac12\hat F_{ij,(m_1,m_2)}^{n,q,b}\ell_{(m_1,m_2)}^q(1)
 \end{pmatrix}.
\]
Proposition~\ref{prop:conservative-moment-representation}, with $\avec_i$ replaced by $\hat{\Flux}_{ij}^n$, reduces the mass sum to $(0,0)$, the $u$-momentum sum to $(0,0),(1,0)$, and the $v$-momentum sum to $(0,0),(0,1)$. Only basis vectors of degrees $(0,0),(1,0),(2,0),(0,1),(0,2)$ for $h$ and $(0,0)$ for $b$ contribute to the energy sum. These reductions give Eq.~\eqref{eq:modal-macroscopic-flux}.
\end{proof}

\subsection{Velocity multiplication and upwind split}\label{sec:modal-velocity}
The flux in Eq.~\eqref{eq:nonequilibrium-modal-flux} involves products $\mathbf J_u^q\avec^{q,s}$ and $\mathbf J_v^q\avec^{q,s}$. These products are evaluated by the three term recurrence derived below. A separate discrete basis on each velocity quadrant preserves this coupling for Cartesian upwinding.

\begin{proposition}[Coupling under velocity multiplication]\label{prop:velocity-coupling}
The matrices defined by Eq.~\eqref{eq:quadrant-velocity-projection} are symmetric and satisfy
\begin{equation}
 \begin{aligned}
 (\mathbf J_u^q)_{(m_1,m_2),(n_1,n_2)}
 &=0\quad\text{if}\quad m_2\ne n_2\ \text{or}\ |m_1-n_1|>1,\\
 (\mathbf J_v^q)_{(m_1,m_2),(n_1,n_2)}
 &=0\quad\text{if}\quad m_1\ne n_1\ \text{or}\ |m_2-n_2|>1.
 \end{aligned}
 \label{eq:tensor-velocity-projection}
\end{equation}
For nonnegative integers $r_1,r_2$ and $(m_1,m_2),(n_1,n_2)\in\mathcal V_q$,
\[
 \begin{aligned}
 (\mathbf J_u^q)^{r_1}(\mathbf J_v^q)^{r_2}
 &=(\mathbf J_v^q)^{r_2}(\mathbf J_u^q)^{r_1},\\
 \bigl[(\mathbf J_u^q)^{r_1}(\mathbf J_v^q)^{r_2}\bigr]_{(m_1,m_2),(n_1,n_2)}
 &=0\quad\text{if}\quad |m_1-n_1|>r_1\ \text{or}\ |m_2-n_2|>r_2.
 \end{aligned}
\]
\end{proposition}
\begin{proof}
By separability~\eqref{eq:tensor-basis},
\begin{equation}
 \begin{aligned}
 \langle\mathbf D_u^q\bm\phi_{(m_1,m_2)}^q,\bm\phi_{(n_1,n_2)}^q\rangle_q
 &=\langle\bm\phi_{(m_1,m_2)}^q,\mathbf D_u^q\bm\phi_{(n_1,n_2)}^q\rangle_q\\[2pt]
 &=\bigl\langle\mathbf p_{u,m_1}^q,u\mathbf p_{u,n_1}^q\bigr\rangle_{w_u^q}
   \,\bigl\langle\mathbf p_{v,m_2}^q,\mathbf p_{v,n_2}^q\bigr\rangle_{w_v^q}\\[2pt]
 &=\delta_{m_2n_2}\bigl\langle\mathbf p_{u,m_1}^q,u\mathbf p_{u,n_1}^q\bigr\rangle_{w_u^q}
 =\delta_{m_2n_2}\bigl\langle u\mathbf p_{u,m_1}^q,\mathbf p_{u,n_1}^q\bigr\rangle_{w_u^q}\\[2pt]
 &=0,\qquad n_2\ne m_2\ \text{or}\ |m_1-n_1|>1.
 \end{aligned}
 \label{eq:recurrence-orthogonality}
\end{equation}
Set $\bm\phi_{(k_1,k_2)}^q=\bm0$ for $(k_1,k_2)\notin\mathcal V_q$. The orthonormal expansion and Eq.~\eqref{eq:recurrence-orthogonality} give
\[
 \begin{aligned}
 \mathbf D_u^q\bm\phi_{(m_1,m_2)}^q
 &=\sum_{(n_1,n_2)\in\mathcal V_q}
   \langle\mathbf D_u^q\bm\phi_{(m_1,m_2)}^q,\bm\phi_{(n_1,n_2)}^q\rangle_q\,\bm\phi_{(n_1,n_2)}^q\\[3pt]
 &=\sum_{\substack{(n_1,n_2)\in\mathcal V_q,\ n_2=m_2\\ |n_1-m_1|\le1}}
   \langle\mathbf D_u^q\bm\phi_{(m_1,m_2)}^q,\bm\phi_{(n_1,n_2)}^q\rangle_q\,\bm\phi_{(n_1,n_2)}^q\\[3pt]
 &=\sum_{j=-1}^{1}
   \bigl\langle\mathbf D_u^q\bm\phi_{(m_1,m_2)}^q,\bm\phi_{(m_1+j,m_2)}^q\bigr\rangle_q\,\bm\phi_{(m_1+j,m_2)}^q.
 \end{aligned}
\]
Denoting the projections for $j=-1,0,1$ by $\beta_{u,m_1}^q$, $\alpha_{u,m_1}^q$, and $\beta_{u,m_1+1}^q$, respectively, gives the three term recurrence; the $v$ expansion is analogous:
\begin{equation}
 \begin{aligned}
 \mathbf D_u^q\bm\phi_{(m_1,m_2)}^q
 &=\beta_{u,m_1}^q\bm\phi_{(m_1-1,m_2)}^q
 +\alpha_{u,m_1}^q\bm\phi_{(m_1,m_2)}^q
 +\beta_{u,m_1+1}^q\bm\phi_{(m_1+1,m_2)}^q,\\
 \mathbf D_v^q\bm\phi_{(m_1,m_2)}^q
 &=\beta_{v,m_2}^q\bm\phi_{(m_1,m_2-1)}^q
 +\alpha_{v,m_2}^q\bm\phi_{(m_1,m_2)}^q
 +\beta_{v,m_2+1}^q\bm\phi_{(m_1,m_2+1)}^q,
 \end{aligned}
 \label{eq:basis-generation}
\end{equation}
where $\beta_{\xi,0}^q=\beta_{\xi,N_\xi^q}^q=0$ for $\xi=u,v$.

For fixed $q$ and $s$, write $\avec=\avec^{q,s}\in\mathbb R^{N_q}$, with components $a_{(n_1,n_2)}$. By symmetry,
\begin{equation}
 \begin{aligned}
 (\mathbf J_\xi^q\avec)_{(m_1,m_2)}
 &=\sum_{(n_1,n_2)\in\mathcal V_q}a_{(n_1,n_2)}
 \langle\mathbf D_\xi^q\bm\phi_{(n_1,n_2)}^q,\bm\phi_{(m_1,m_2)}^q\rangle_q\\
 &=\sum_{(n_1,n_2)\in\mathcal V_q}a_{(n_1,n_2)}
 \langle\bm\phi_{(n_1,n_2)}^q,\mathbf D_\xi^q\bm\phi_{(m_1,m_2)}^q\rangle_q.
  \end{aligned}
 \label{eq:velocity-coefficient-action}
\end{equation}
Substituting Eq.~\eqref{eq:basis-generation} into Eq.~\eqref{eq:velocity-coefficient-action} and using orthonormality gives
\begin{equation}
 \begin{aligned}
 (\mathbf J_u^q\avec)_{(m_1,m_2)}
 &=\beta_{u,m_1}^q a_{(m_1-1,m_2)}
   +\alpha_{u,m_1}^q a_{(m_1,m_2)}
   +\beta_{u,m_1+1}^q a_{(m_1+1,m_2)},\\
 (\mathbf J_v^q\avec)_{(m_1,m_2)}
 &=\beta_{v,m_2}^q a_{(m_1,m_2-1)}
   +\alpha_{v,m_2}^q a_{(m_1,m_2)}
   +\beta_{v,m_2+1}^q a_{(m_1,m_2+1)},
 \end{aligned}
 \label{eq:modal-products}
\end{equation}
where coefficients outside $\mathcal V_q$ are zero. The similarity transform~\eqref{eq:velocity-matrices} gives
\[
 (\mathbf J_u^q)^{r_1}(\mathbf J_v^q)^{r_2}
 =\mathbf P_q^{-1}(\mathbf D_u^q)^{r_1}(\mathbf D_v^q)^{r_2}\mathbf P_q
 =(\mathbf J_v^q)^{r_2}(\mathbf J_u^q)^{r_1}.
\]
Separability also gives
\[
 \begin{aligned}[b]
 \bigl[(\mathbf J_u^q)^{r_1}(\mathbf J_v^q)^{r_2}\bigr]_{(m_1,m_2),(n_1,n_2)}
 &=\langle\bm\phi_{(m_1,m_2)}^q,(\mathbf D_u^q)^{r_1}(\mathbf D_v^q)^{r_2}\bm\phi_{(n_1,n_2)}^q\rangle_q\\
 &=\langle\mathbf p_{u,m_1}^q,u^{r_1}\mathbf p_{u,n_1}^q\rangle_{w_u^q}
   \langle\mathbf p_{v,m_2}^q,v^{r_2}\mathbf p_{v,n_2}^q\rangle_{w_v^q}\\
 &=\langle u^{r_1}\mathbf p_{u,m_1}^q,\mathbf p_{u,n_1}^q\rangle_{w_u^q}
   \langle v^{r_2}\mathbf p_{v,m_2}^q,\mathbf p_{v,n_2}^q\rangle_{w_v^q}\\
 &=0\quad\text{if}\quad |m_1-n_1|>r_1\ \text{or}\ |m_2-n_2|>r_2.
 \end{aligned}
 \qedhere
\]
\end{proof}

\begin{corollary}[Quadrant partition and upwinding]\label{cor:quadrant-upwind}
On each velocity quadrant $Q_q$, the projections of multiplication by $\xi^\pm=(\xi\pm|\xi|)/2$, $\xi=u,v$, are either $\mathbf J_\xi^q$ or $\bm0$. They retain the adjacent degree coupling in Eq.~\eqref{eq:tensor-velocity-projection}.
\end{corollary}
\begin{proof}
For a face normal to the positive $x$ direction, the left and right reconstructed distribution functions $\fvec_0^L$ and $\fvec_0^R$ in Eq.~\eqref{eq:interface-half-space-state} give
\[
 u\fvec_0
 =\frac{u+|u|}{2}\fvec_0^L+\frac{u-|u|}{2}\fvec_0^R
 =\frac{u}{2}(\fvec_0^L+\fvec_0^R)
  +\frac{|u|}{2}(\fvec_0^L-\fvec_0^R).
\]
Without the quadrant partition,
\[
 |u|\notin\operatorname{span}\{1,u\}\qquad\text{across }u=0,
\]
so the last term can couple nonadjacent degrees.

On each quadrant, both velocity signs are constant, so
\begin{equation}
 (\xi^+,\xi^-)=
 \begin{cases}
 (\xi,0),&\xi>0\text{ on }Q_q,\\[3pt]
 (0,\xi),&\xi<0\text{ on }Q_q,
 \end{cases}
 \qquad \xi=u,v.
 \label{eq:quadrant-upwind-split}
\end{equation}
Projection and Proposition~\ref{prop:velocity-coupling} give the result. We therefore construct the discrete basis separately on the four quadrants to preserve adjacent degree coupling for Cartesian upwinding in both coordinate directions.
\end{proof}

\subsection{Macroscopic and nonequilibrium updates}\label{sec:modal-cell-update}
The macroscopic variables are updated by Eq.~\eqref{eq:macro-update}. The flux divergence is
\[
 \begin{aligned}
 \bm{\mathcal D}_i^{\mathrm v,n}
 &=\frac{1}{|\Omega_i|}\sum_{j\in\mathcal J_i}A_{ij}\Flux_{ij}^{\mathrm v,n},\\
 \hat{\bm{\mathcal D}}_i^n
 &=\mathbf P^{-1}\bm{\mathcal D}_i^{\mathrm v,n}
 =\frac{1}{|\Omega_i|}\sum_{j\in\mathcal J_i}A_{ij}\hat{\Flux}_{ij}^n.
 \end{aligned}
\]
For the update $\fvec_i^{\mathrm v,\ast}=\M_i^{\mathrm v,n+1}+\rvec_i^{\mathrm v,\ast}$, Eq.~\eqref{eq:distribution-update} gives
\begin{equation}
 \left(1+\frac{\Delta t}{\tau_i^{n+1}}\right)\rvec_i^{\mathrm v,\ast}
 =\rvec_i^{\mathrm v,n}-\Delta t\bm{\mathcal D}_i^{\mathrm v,n}
  +\M_i^{\mathrm v,n}-\M_i^{\mathrm v,n+1}.
 \label{eq:residual-distribution-update}
\end{equation}
The superscript $\ast$ denotes the state before truncation and conservative correction. The projection $\avec_i=\mathbf P^{\mathsf T}\mathbf\Omega\rvec_i^{\mathrm v}$ gives
\begin{equation}
 \avec_i^\ast=\lambda_i\bigl[
 \avec_i^n-\Delta t\hat{\bm{\mathcal D}}_i^n
 +\hat{\M}_i^n-\hat{\M}_i^{n+1}\bigr],
 \qquad \lambda_i=\frac{\tau_i^{n+1}}{\tau_i^{n+1}+\Delta t}.
 \label{eq:explicit-candidate}
\end{equation}
The difference $\hat{\M}_i^n-\hat{\M}_i^{n+1}$ expresses the updated distribution function relative to its new equilibrium. Section~\ref{sec:properties} applies this update on the candidate coefficients, then truncates and corrects them to obtain $\avec_i^{n+1}$.

\subsection{Diffuse wall boundary condition}\label{sec:boundaries}
At a diffuse wall, outgoing velocities carry the gas side distribution function, while incoming velocities carry a Maxwellian $\M_w$ with the wall temperature $T_w$ and tangential velocity $\bm U_w$. The reconstructed outgoing coefficients are
\[
 \hat{\fvec}^{\mathrm{out}}
 =\hat{\M}^{\mathrm{out}}+\avec^{\mathrm{out}}.
\]
With the normal directed out of the gas cell, the two contributions are weighted by $u_{ij}^+=\max(u_{ij},0)$ and $u_{ij}^-=\min(u_{ij},0)$. As in Eq.~\eqref{eq:velocity-matrices}, their quadrant matrices are
\[
 \mathbf J_{ij}^{\pm,q}
 =\mathbf P_q^{-1}\mathbf D_{ij}^{\pm,q}\mathbf P_q,
 \qquad \mathbf J_{ij}^q=\mathbf J_{ij}^{+,q}+\mathbf J_{ij}^{-,q},
\]
where $\mathbf D_{ij}^{\pm,q}$ is diagonal with entries $u_{ij,\nu}^{\pm}$ on $Q_q$. Assembling the quadrant blocks for both $h$ and $b$ gives
\[
 \mathbf J_{ij}^{\pm}
 =\mathbf P^{-1}\mathbf D_{ij}^{\pm}\mathbf P
 =\mathbf P^{\mathsf T}\mathbf\Omega\mathbf D_{ij}^{\pm}\mathbf P,
 \qquad \mathbf J_{ij}=\mathbf J_{ij}^{+}+\mathbf J_{ij}^{-}.
\]
For a Cartesian wall, $u_{ij}=\pm u$ or $\pm v$, so Corollary~\ref{cor:quadrant-upwind} gives the same adjacent degree coupling.

Zero normal mass flux determines the wall density:
\begin{equation}
 \sum_{q=1}^{4}\sum_{(m_1,m_2)\in\mathcal V_q}
 \left[
  \hat f_{(m_1,m_2)}^{q,h,\mathrm{out}}\,\ell_{(m_1,m_2)}^q(u_{ij}^{+})
  +\hat M_{w,(m_1,m_2)}^{q,h}\,\ell_{(m_1,m_2)}^q(u_{ij}^{-})
 \right]=0.
 \label{eq:wall}
\end{equation}
Since $\M_w$ is proportional to its density, this is a scalar linear equation. The resulting wall flux is
\begin{equation}
 \hat{\Flux}_{ij}
 =\mathbf J_{ij}^+\hat{\fvec}^{\mathrm{out}}
  +\mathbf J_{ij}^-\hat{\M}_w.
 \label{eq:diffuse-wall-modal-flux}
\end{equation}
The decomposition~\eqref{eq:full-modal-flux} has the wall contributions
\[
 \begin{aligned}
 \hat{\Flux}_{ij,\mathrm{neq}}&=\mathbf J_{ij}^+\avec^{\mathrm{out}},\\
 \hat{\Flux}_{ij,\mathrm{eq}}&=\mathbf J_{ij}^+\hat{\M}^{\mathrm{out}}
                         +\mathbf J_{ij}^-\hat{\M}_w.
 \end{aligned}
\]
The wall flux enters the updates in Section~\ref{sec:modal-cell-update} in the same manner as an interior face flux.

\section{Adaptive evolution of nonequilibrium DOF}\label{sec:properties}
The blocks for updating the nonequilibrium DOF are determined jointly by the nonequilibrium transport candidates and the equilibrium source blocks. Block truncation and conservative correction then give the nonequilibrium DOF at the beginning of the next time step.

\subsection{Nonequilibrium DOF blocks and candidates}\label{sec:candidates}\label{sec:transport-support}
In each velocity quadrant $Q_q$, the nonequilibrium DOF are grouped into $5\times5$ degree blocks in $\mathcal V_q$, allowing smaller blocks at the degree limits:
\begin{equation}
 B_{\alpha\beta}^{q}
 =\left\{(m_1,m_2)\in\mathcal V_q:
 \begin{array}{l}
 5\alpha\le m_1<\min\{5(\alpha+1),N_u^q\},\\
 5\beta \le m_2<\min\{5(\beta+1),N_v^q\}
 \end{array}\right\}.
 \label{eq:modal-block}
\end{equation}
The width five includes all nonequilibrium degrees that can contribute to the macroscopic flux: the moment weights and the velocity products in the nonequilibrium flux each have degree at most two.
The set $\mathfrak B^q$ consists of all nonempty blocks $B_{\alpha\beta}^q$ in the partition of $\mathcal V_q$, with $\alpha,\beta$ nonnegative integers. For each component $s\in\{h,b\}$, the blocks retained in cell $i$ at $t^n$ form the set $\mathcal B_i^{n,q,s}\subseteq\mathfrak B^q$. The associated degree pairs form
\begin{equation}
 \mathcal A_i^{n,q,s}
 =\bigcup_{B\in\mathcal B_i^{n,q,s}}B,
 \qquad
 a_{i,(m_1,m_2)}^{n,q,s}=0
 \quad\text{for }(m_1,m_2)\notin\mathcal A_i^{n,q,s}.
 \label{eq:retained-block-indices}
\end{equation}
The lowest block $B_{00}^q$ is retained for both $h$ and $b$ because it contains all DOF that contribute to the conservative moments in Section~\ref{sec:modal-moments}.

The discrete basis coefficients are reconstructed in physical space, with unretained coefficients set to zero. In two dimensions, the reconstruction stencil of either cell $\zeta\in\{i,j\}$ adjacent to an interior Cartesian face $ij$ consists of $\zeta$ and its four coordinate neighbors:
\begin{equation}
 \mathcal S_{ij,\zeta}
 =\{\zeta,\zeta_x^-,\zeta_x^+,\zeta_y^-,\zeta_y^+\}.
 \label{eq:face-reconstruction-stencil}
\end{equation}
This reconstruction is supported on the blocks
\begin{equation}
 \mathcal I_{ij,\zeta}^{n,q,s}
 =\bigcup_{\ell\in\mathcal S_{ij,\zeta}}
  \mathcal B_{\ell}^{n,q,s}.
 \label{eq:face-input-support}
\end{equation}

On a Cartesian face, the nonequilibrium flux~\eqref{eq:nonequilibrium-modal-flux} contains velocity products of degree at most two. The required flux coefficients are evaluated directly on the candidate blocks using the three term recurrence~\eqref{eq:modal-products}. By Proposition~\ref{prop:velocity-coupling} and Corollary~\ref{cor:quadrant-upwind}, these velocity products couple degree pairs $(k_1,k_2)$ and $(m_1,m_2)$ only if
\begin{equation}
 |m_d-k_d|\le2,
 \qquad d=1,2.
 \label{eq:modal-bandwidth-two}
\end{equation}
A single layer of neighboring $5\times5$ blocks covers this coupling.

For $\mathcal I\subseteq\mathfrak B^q$, define
\begin{equation}
 \mathcal N_1(\mathcal I)
 =\left\{B_{\gamma\delta}^{q}\in\mathfrak B^q:
 \exists B_{\alpha\beta}^{q}\in\mathcal I,
 \max\{|\gamma-\alpha|,|\delta-\beta|\}\le1\right\}.
 \label{eq:one-block-expansion}
\end{equation}
Both sides of an interface use the same nonequilibrium candidates to evaluate a shared nonequilibrium flux. The face nonequilibrium candidate set is
\begin{equation}
 \mathcal C_{ij,\mathrm{neq}}^{n,q,s}=\mathcal C_{ji,\mathrm{neq}}^{n,q,s}
 =\mathcal N_1\!\left(
   \mathcal I_{ij,i}^{n,q,s}\cup\mathcal I_{ij,j}^{n,q,s}
   \right).
 \label{eq:common-face-support}
\end{equation}
The cell nonequilibrium candidates comprise the retained blocks and the candidates of all adjacent faces:
\begin{equation}
 \mathcal C_{i,\mathrm{neq}}^{n,q,s}
 =\mathcal B_i^{n,q,s}
  \cup
  \bigcup_{j\in\mathcal J_i}\mathcal C_{ij,\mathrm{neq}}^{n,q,s}.
 \label{eq:local-candidate}
\end{equation}

\subsection{Equilibrium source blocks and blockwise evolution}\label{sec:adaptive-evolution}\label{sec:equilibrium-flux-evaluation}
\emph{(i) Macroscopic update.}
The macroscopic flux~\eqref{eq:modal-macroscopic-flux} can be written as
\[
 \mFlux_{ij}^n
 =\Cmat\bigl(\hat{\Flux}_{ij,\mathrm{eq}}^n
                  +\hat{\Flux}_{ij,\mathrm{neq}}^n\bigr),
\]
which requires only the flux coefficients in $B_{00}^q$. The macroscopic update~\eqref{eq:macro-update} is then
\[
 \W_i^{n+1}
 =\W_i^n-\frac{\Delta t}{|\Omega_i|}
       \sum_{j\in\mathcal J_i}A_{ij}\mFlux_{ij}^n,
 \qquad \M_i^{n+1}=\M(\W_i^{n+1}).
\]

\emph{(ii) Equilibrium source and candidate blocks.}
For fixed $q,s$ and a $5\times5$ block $B$, define the coefficient vectors in $\mathbb R^{25}$
\[
 \begin{aligned}
 \mathbf A_{i,B}^{n}
 &=(\avec_i^{n,q,s})_B
 =\bigl(a_{i,(m_1,m_2)}^{n,q,s}\bigr)_{(m_1,m_2)\in B},\\
 \widehat{\mathbf M}_{i,B}^{n}
 &=(\hat{\M}_i^{n,q,s})_B
 =\bigl(\hat M_{i,(m_1,m_2)}^{n,q,s}\bigr)_{(m_1,m_2)\in B}.
 \end{aligned}
\]
Blocks at the degree limits have correspondingly shorter vectors.
The nonequilibrium face flux is zero outside $\mathcal C_{ij,\mathrm{neq}}^{n,q,s}$:
\[
 (\hat{\Flux}_{ij,\mathrm{neq}}^n)_B=\mathbf0,
 \qquad B\notin\mathcal C_{ij,\mathrm{neq}}^{n,q,s}.
\]
The nonequilibrium flux divergence is
\begin{equation}
 (\hat{\bm{\mathcal D}}_{i,\mathrm{neq}}^n)_B
 =\frac{1}{|\Omega_i|}\sum_{j\in\mathcal J_i}
     A_{ij}(\hat{\Flux}_{ij,\mathrm{neq}}^n)_B.
 \label{eq:candidate-flux-divergence}
\end{equation}
The equilibrium flux divergence is
\begin{equation}
 \hat{\bm{\mathcal D}}_{i,\mathrm{eq}}^n
 =\frac{1}{|\Omega_i|}\sum_{j\in\mathcal J_i}
      A_{ij}\hat{\Flux}_{ij,\mathrm{eq}}^n.
 \label{eq:source-equilibrium-divergence}
\end{equation}
With $\lambda_i$ from Eq.~\eqref{eq:explicit-candidate}, define the equilibrium contribution to the update by
\begin{equation}
 \mathbf G_{i,B}^{n}
 =\lambda_i\left[
       \widehat{\mathbf M}_{i,B}^{n}
      -\widehat{\mathbf M}_{i,B}^{n+1}
      -\Delta t\,(\hat{\bm{\mathcal D}}_{i,\mathrm{eq}}^n)_B
              \right].
 \label{eq:explicit-update-source}
\end{equation}
This contribution can be nonzero even on blocks with zero nonequilibrium flux. The Maxwellian coefficients follow from Eq.~\eqref{eq:separable-coefficients}.

Let $\mathcal B_{G,i}^{n-1,q,s}$ denote the equilibrium source blocks retained at the preceding time step. We use their union with the current nonequilibrium transport candidates to select the blocks evaluated in the update:
\begin{equation}
 \mathcal P_{G,i}^{n,q}
 =\begin{cases}
    \mathfrak B^q,& n=0,\\[3pt]
    \displaystyle\bigcup_{s\in\{h,b\}}
      \left(\mathcal B_{G,i}^{n-1,q,s}
            \cup\mathcal C_{i,\mathrm{neq}}^{n,q,s}\right),& n\ge1.
   \end{cases}
 \label{eq:source-candidates}
\end{equation}

\emph{(iii) Nonequilibrium update.}
For $B\in\mathcal P_{G,i}^{n,q}$, Eq.~\eqref{eq:explicit-candidate} becomes
\begin{equation}
 \mathbf A_{i,B}^{\ast}
 =\begin{cases}
   \lambda_i\bigl[\mathbf A_{i,B}^{n}
       -\Delta t\,(\hat{\bm{\mathcal D}}_{i,\mathrm{neq}}^n)_B\bigr]
       +\mathbf G_{i,B}^{n},
       &B\in\mathcal C_{i,\mathrm{neq}}^{n,q,s},\\[4pt]
   \mathbf G_{i,B}^{n},
       &B\in\mathcal P_{G,i}^{n,q}\setminus\mathcal C_{i,\mathrm{neq}}^{n,q,s}.
  \end{cases}
 \label{eq:candidate-block-update}
\end{equation}
\subsection{Block truncation}\label{sec:adaptive-truncation}\label{sec:adaptive-analysis}
A block discarded by the nonequilibrium truncation may still be needed to represent the equilibrium source, which may contribute to the nonequilibrium evolution through Eq.~\eqref{eq:candidate-block-update}. Therefore, the nonequilibrium blocks and equilibrium source blocks are truncated separately on $\mathcal P_{G,i}^{n,q}$, yielding $\mathcal B_i^{n+1,q,s}$ and $\mathcal B_{G,i}^{n,q,s}$, respectively. The latter is used only to construct $\mathcal P_{G,i}^{n+1,q}$ through Eq.~\eqref{eq:source-candidates}.

\emph{(i) Block energies and truncation budget.}
For a coefficient vector $\mathbf z$, with $\mathbf z^h$ and $\mathbf z^b$ collecting the four quadrants of each reduced distribution, define
\[
 \|\mathbf z\|_{\theta}^2
 =\|\mathbf z^h\|_2^2+\theta^{-2}\|\mathbf z^b\|_2^2,
 \qquad w_h(\theta)=1,\quad w_b(\theta)=\theta^{-2}.
\]
Here $\|\mathbf z\|_2^2=\sum_j|z_j|^2$ denotes the squared Euclidean norm.
The energy of the updated coefficients on block $B$ is
\begin{equation}
 E_{i,B}^{n,q,s}
 =w_s(\theta_i^{n+1})\|\mathbf A_{i,B}^{\ast}\|_2^2,
 \qquad B\in\mathcal P_{G,i}^{n,q}.
 \label{eq:block-energy}
\end{equation}
The components $h$ and $b$ use the same relative tolerance, $\varepsilon^{\mathrm{ref}}=\varepsilon_h^{\mathrm{ref}}=\varepsilon_b^{\mathrm{ref}}$. The squared truncation budget is
\begin{equation}
 \delta_i^2
   =(\varepsilon^{\mathrm{ref}})^2
      \|\hat{\M}_i^{n+1}\|_{\theta_i^{n+1}}^2.
 \label{eq:state-truncation-budget}
\end{equation}

\emph{(ii) Nonequilibrium block truncation.}
Retain $B_{00}^q$ for both components. Discard the other updated blocks in increasing order of energy across all quadrants and both components, stopping before the cumulative discarded energy exceeds $\delta_i^2$. The remaining blocks form $\mathcal B_i^{n+1,q,s}$ and satisfy
\begin{equation}
 \sum_{q,s}\sum_{B\in\mathcal P_{G,i}^{n,q}\setminus\mathcal B_i^{n+1,q,s}}
     E_{i,B}^{n,q,s}\le\delta_i^2.
 \label{eq:block-truncation}
\end{equation}
The truncated vector $\avec_i^{\mathrm{tr}}$ has block coefficients
\begin{equation}
 \mathbf A_{i,B}^{\mathrm{tr},q,s}
 =\begin{cases}
   \mathbf A_{i,B}^{\ast},&B\in\mathcal B_i^{n+1,q,s},\\
   \mathbf0,&\text{otherwise}.
  \end{cases}
 \label{eq:stored-block-update}
\end{equation}
Retaining $B_{00}^q$ preserves the conservative moments by Proposition~\ref{prop:conservative-moment-representation}:
\[
 \Cmat\avec_i^{\mathrm{tr}}=\Cmat\avec_i^{\ast}.
\]

\emph{(iii) Equilibrium source block truncation.}
Apply the rule in (ii) independently to the candidate blocks $B\in\mathcal P_{G,i}^{n,q}$, using the source energies $w_s(\theta_i^{n+1})\|\mathbf G_{i,B}^n\|_2^2$ and the budget $\delta_i^2$. The retained source block sets $\mathcal B_{G,i}^{n,q,s}$ satisfy
\[
 \sum_{q,s}\sum_{B\in\mathcal P_{G,i}^{n,q}\setminus\mathcal B_{G,i}^{n,q,s}}
    w_s(\theta_i^{n+1})\|\mathbf G_{i,B}^n\|_2^2
     \le\delta_i^2.
\]

\subsection{Conservative correction}\label{sec:conservative-correction}
Truncation preserves the existing conservative moments. We next require the nonequilibrium DOF to contribute zero to these moments. This condition, $\Cmat\avec_i=\bm0$, is enforced by correcting only the constant coefficients $a_{i,(0,0)}^{q,s}$, $q=1,\ldots,4$, $s\in\{h,b\}$. Their moment columns are
\[
 \bm c_h^q=
 \begin{pmatrix}
  \ell_{(0,0)}^q(1)\\
  \ell_{(0,0)}^q(u)\\
  \ell_{(0,0)}^q(v)\\
  \tfrac12\ell_{(0,0)}^q(u^2+v^2)
 \end{pmatrix},
 \qquad
 \bm c_b^q=
 \begin{pmatrix}0\\0\\0\\\tfrac12\ell_{(0,0)}^q(1)\end{pmatrix}.
\]
With moment residual $\bm r_i^{\mathrm c}=\Cmat\avec_i^{\mathrm{tr}}$ and temperature scale $\vartheta_i=\theta_i^{n+1}$, the increments solve the constrained least squares problem
\begin{equation}
 \begin{aligned}
 \underset{\{\delta a_h^q,\delta a_b^q\}_{q=1}^{4}}{\operatorname{minimize}}
 &\quad\frac12\sum_{q=1}^{4}\left[
     (\delta a_h^q)^2+(\vartheta_i)^{-2}(\delta a_b^q)^2
     \right],\\
 \text{subject to}
 &\quad\sum_{q=1}^{4}
    \left(\bm c_h^q\delta a_h^q+\bm c_b^q\delta a_b^q\right)
    =-\bm r_i^{\mathrm c}.
 \end{aligned}
 \label{eq:conservative-correction}
\end{equation}
The updated coefficients on each retained block are
\[
 \mathbf A_{i,B}^{n+1,q,s}
 =\mathbf A_{i,B}^{\mathrm{tr},q,s}
  +\delta\mathbf A_{i,B}^{q,s},
 \qquad B\in\mathcal B_i^{n+1,q,s}.
\]
Here $\delta\mathbf A_{i,B}^{q,s}$ is zero except for $\delta a_s^q$, the coefficient increment of the constant basis vector in $B_{00}^q$.
Section~\ref{sec:implicit-update} gives the corresponding correction for the implicit iteration.

Algorithm~\ref{alg:adaptive-ugks} summarizes the explicit evolution.

\begin{algorithm}[H]
\caption{Nonequilibrium capturing UGKS}\label{alg:adaptive-ugks}
\begin{algorithmic}[1]
\Require Physical and velocity grids, block partition, boundary conditions, $\W_i^0$, $\fvec_i^{\mathrm{v},0}$, $\Delta t$, $\varepsilon^{\mathrm{ref}}$, and $t_{\mathrm{end}}$.
\State Construct the fixed discrete basis. Set $n=0$ and $t^0=0$.
\State Project $\fvec_i^{\mathrm{v},0}-\M^{\mathrm v}(\W_i^0)$ to obtain $\avec_i^0$.
\State Apply nonequilibrium block truncation on $\mathfrak B^q$ as in Section~\ref{sec:adaptive-truncation}(ii), using $\avec_i^0$, $\theta_i^0$, and $\hat{\M}_i^0$; store the retained coefficients in $\avec_i^0$ and the retained blocks in $\mathcal B_i^{0,q,s}$.
\State Correct $\avec_i^0$ by Eq.~\eqref{eq:conservative-correction}, with $\bm r_i^{\mathrm c}=\Cmat\avec_i^0$ and $\vartheta_i=\theta_i^0$.
\While{$t^n<t_{\mathrm{end}}$}
  \State Set $\Delta t\gets\min\{\Delta t,t_{\mathrm{end}}-t^n\}$.
  \State Reconstruct the face states using Eq.~\eqref{eq:reconstructed-split}, with unretained nonequilibrium coefficients set to zero.
  \State Construct the face and cell nonequilibrium candidates $\mathcal C_{ij,\mathrm{neq}}^{n,q,s}$ and $\mathcal C_{i,\mathrm{neq}}^{n,q,s}$ as in Section~\ref{sec:candidates}.
  \State Evaluate the nonequilibrium face fluxes on $\mathcal C_{ij,\mathrm{neq}}^{n,q,s}$ using Eq.~\eqref{eq:nonequilibrium-modal-flux} for interior faces and the boundary conditions elsewhere, including Eqs.~\eqref{eq:wall}--\eqref{eq:diffuse-wall-modal-flux} for diffuse walls.
  \State Evaluate the equilibrium flux coefficients in $B_{00}^q$ and compute $\mFlux_{ij}^n$ from the total flux coefficients by Eq.~\eqref{eq:modal-macroscopic-flux}. Update $\W_i^{n+1}$ by Eq.~\eqref{eq:macro-update}.
  \State Evaluate $\M_i^{n+1}$ and $\tau_i^{n+1}$ from $\W_i^{n+1}$, and $\lambda_i$ from Eq.~\eqref{eq:explicit-candidate}.
  \State Construct $\mathcal P_{G,i}^{n,q}$ by Eq.~\eqref{eq:source-candidates}, with $\mathcal P_{G,i}^{0,q}=\mathfrak B^q$.
  \State For each $B\in\mathcal P_{G,i}^{n,q}$, compute the flux divergences by Eqs.~\eqref{eq:candidate-flux-divergence}--\eqref{eq:source-equilibrium-divergence}, the source coefficients $\mathbf G_{i,B}^n$ by Eq.~\eqref{eq:explicit-update-source}, and the updated coefficients $\mathbf A_{i,B}^{\ast}$ by Eq.~\eqref{eq:candidate-block-update}.
  \State Truncate the nonequilibrium blocks in $\mathcal P_{G,i}^{n,q}$ as in Section~\ref{sec:adaptive-truncation}(ii), using the energies~\eqref{eq:block-energy} and budget~\eqref{eq:state-truncation-budget}; obtain $\avec_i^{\mathrm{tr}}$ and $\mathcal B_i^{n+1,q,s}$.
  \State Independently truncate the equilibrium source blocks in $\mathcal P_{G,i}^{n,q}$ as in Section~\ref{sec:adaptive-truncation}(iii), with the same budget $\delta_i^2$; retain $\mathcal B_{G,i}^{n,q,s}$ for constructing $\mathcal P_{G,i}^{n+1,q}$.
  \State Solve Eq.~\eqref{eq:conservative-correction} and set $\mathbf A_{i,B}^{n+1,q,s}=\mathbf A_{i,B}^{\mathrm{tr},q,s}+\delta\mathbf A_{i,B}^{q,s}$ for $B\in\mathcal B_i^{n+1,q,s}$.
  \State $t^{n+1}\gets t^n+\Delta t$; $n\gets n+1$.
\EndWhile
\Ensure Final time values $\W_i^N$, $\mathbf A_{i,B}^{N,q,s}$, and $\mathcal B_i^{N,q,s}$, where $t^N=t_{\mathrm{end}}$.
\end{algorithmic}
\end{algorithm}
\FloatBarrier

\section{Implicit iteration for steady flow}\label{sec:implicit-extension}
\begingroup\emergencystretch=1em
For steady flow, the nonequilibrium capturing implicit UGKS (NCIUGKS) couples macroscopic prediction with microscopic correction~\cite{zhu2016implicit}. The block and coefficient notation of Section~\ref{sec:properties} is used with outer iteration index $k$ in place of time index $n$. Each outer iteration follows
\[
 (\W^k,\avec^k)
 \longrightarrow (\W^{\mathrm p},\M^{\mathrm p},\tau^{\mathrm p})
 \longrightarrow \dvec
 \longrightarrow (\W^{k+1},\avec^{k+1}).
\]
Here $\dvec$ contains the basis coefficients of the distribution function correction. The predicted $\M^{\mathrm p}$ and $\tau^{\mathrm p}$ are fixed during the microscopic solve; the final update includes block truncation and conservative correction.
\par\endgroup

At outer iteration $k$, the macroscopic and microscopic residuals are
\begin{equation}
 \begin{aligned}
 \Rop_{W,i}^k
 &=-\frac{1}{|\Omega_i|}
       \sum_{j\in\mathcal J_i}A_{ij}\mFlux_{ij}^k,\\
 \Rop_{f,i}^k&=\frac{\M(\W_i^k)-\fvec_i^k}{\tau_i^k}
                                  -\bm{\mathcal D}_i^k.
 \end{aligned}
 \label{eq:steady-transport-residual}
\end{equation}
The time step used to construct the UGKS flux is set by the explicit CFL number,
\begin{equation}
 \Delta t
 =\mathrm{CFL}_{\mathrm{exp}}
  \left(\frac{u_{\max}}{\Delta x_{\min}}
       +\frac{v_{\max}}{\Delta y_{\min}}\right)^{-1}.
 \label{eq:explicit-cfl-step}
\end{equation}
Here $u_{\max}$ and $v_{\max}$ are the upper bounds of $|u|$ and $|v|$ on the velocity domain, $\Delta x_{\min}$ and $\Delta y_{\min}$ are the minimum characteristic cell lengths, and $\mathrm{CFL}_{\mathrm{exp}}$ is the explicit CFL number. The implicit iteration uses the numerical time step
\[
 \Delta t_i^{\mathrm{num}}=\mathrm{CFL}_{\mathrm{imp}}\Delta t,
\]
where $\mathrm{CFL}_{\mathrm{imp}}$ is the implicit CFL number.

\subsection{Macroscopic prediction}\label{sec:macro-prediction}
We predict the macroscopic variables for the microscopic correction as
\[
 \W_i^{\mathrm p}=\W_i^k+\Delta\W_i^{\mathrm p}.
\]
The corresponding delta form implicit system is~\cite{zhu2016implicit}
\begin{equation}
 \frac{\Delta\W_i^{\mathrm p}}{\Delta t_i^{\mathrm{num}}}
 +\frac{1}{|\Omega_i|}\sum_{j\in\mathcal J_i}
                  A_{ij}\,\delta\mFlux_{ij}^{\mathrm p}
 =\Rop_{W,i}^k.
 \label{eq:macro-predictor-balance}
\end{equation}
The face flux increment is approximated using the Euler flux:
\begin{equation}
 \begin{aligned}
 \delta\mFlux_{E,i}^{\mathrm p}
 &=\mFlux_E(\W_i^k+\Delta\W_i^{\mathrm p})
   -\mFlux_E(\W_i^k),\\
 \delta\mFlux_{ij}^{\mathrm p}
 &=\frac12\left[
   \delta\mFlux_{E,i}^{\mathrm p}
  +\delta\mFlux_{E,j}^{\mathrm p}\right]
  +\frac12\Lambda_{ij}
      (\Delta\W_i^{\mathrm p}-\Delta\W_j^{\mathrm p}),
 \end{aligned}
 \label{eq:macro-face-increment}
\end{equation}
where $\mFlux_E$ is the Euler flux directed outward from cell $i$, and $\Lambda_{ij}$ is the spectral radius:
\begin{equation}
 \Lambda_{ij}
 =|\overline{\bm U}_{ij}\cdot\bm n_{ij}|
       +\sqrt{\gamma\bar p_{ij}/\bar\rho_{ij}}
       +\frac{2\overline{\nu}_{{\mathrm{kin}},ij}}{d_{ij}},
 \label{eq:spectral-speed}
\end{equation}
where $\overline{\nu}_{{\mathrm{kin}},ij}=\overline{(\mu/\rho)}_{ij}$, bars denote face averages, and $d_{ij}$ is the distance between adjacent cell centers.

The outward face normals of a closed cell satisfy
\[
 \sum_{j\in\mathcal J_i}A_{ij}\bm n_{ij}=\bm0.
\]
Equation~\eqref{eq:macro-predictor-balance} then reduces to
\begin{equation}
 D_i^{\mathrm p}\Delta\W_i^{\mathrm p}
 =|\Omega_i|\Rop_{W,i}^k
 -\frac12\sum_{j\in\mathcal J_i}A_{ij}
  \left[\delta\mFlux_{E,j}^{\mathrm p}
        -\Lambda_{ij}\Delta\W_j^{\mathrm p}\right],
 \label{eq:macro-prediction}
\end{equation}
with
\begin{equation}
 D_i^{\mathrm p}
 =\frac{|\Omega_i|}{\Delta t_i^{\mathrm{num}}}
  +\frac12\sum_{j\in\mathcal J_i}A_{ij}\Lambda_{ij}.
 \label{eq:predictor-diagonal}
\end{equation}
Equation~\eqref{eq:macro-prediction} is solved by lower upper symmetric Gauss--Seidel (LU-SGS) iterations. The resulting prediction determines the equilibrium and relaxation time:
\begin{equation}
 \M_i^{\mathrm p}=\M(\W_i^{\mathrm p}),
 \qquad
 \tau_i^{\mathrm p}=\tau(\W_i^{\mathrm p}).
 \label{eq:predicted-equilibrium}
\end{equation}
These predicted quantities enter the microscopic equation; the updated macroscopic state $\W_i^{\dagger}$ follows from the conservative moments of $\dvec_i$ in Section~\ref{sec:implicit-update}.

\subsection{Microscopic correction equation}\label{sec:implicit}
Let $\Delta\fvec_i$ be the correction to the two reduced distribution functions. A first order upwind operator approximates the transport increment:
\begin{equation}
 (\mathcal T^{\mathrm{up}}\Delta\fvec)_i
 =\frac{1}{|\Omega_i|}\sum_{j\in\mathcal J_i}A_{ij}
       \left(u_{ij}^{+}\Delta\fvec_i
                    +u_{ij}^{-}\Delta\fvec_j\right),
 \label{eq:upwind-increment-operator}
\end{equation}
with the signed normal velocities $u_{ij}^{\pm}$ defined in Section~\ref{sec:boundaries}. With $\M_i^{\mathrm p}$ and $\tau_i^{\mathrm p}$ fixed, the implicit microscopic equation is
\begin{equation}
 \frac{\Delta\fvec_i}{\Delta t_i^{\mathrm{num}}}
 +\bm{\mathcal D}_i^k
 +(\mathcal T^{\mathrm{up}}\Delta\fvec)_i
 =\frac{\M_i^{\mathrm p}-\fvec_i^k-\Delta\fvec_i}
        {\tau_i^{\mathrm p}}.
 \label{eq:frozen-implicit-balance}
\end{equation}
Rearranging gives
\begin{equation}
 \kappa_i\Delta\fvec_i
 +(\mathcal T^{\mathrm{up}}\Delta\fvec)_i
 =\bm S_i^k,
 \qquad
 \kappa_i=\frac{1}{\Delta t_i^{\mathrm{num}}}
          +\frac{1}{\tau_i^{\mathrm p}},
 \label{eq:implicit}
\end{equation}
where
\begin{equation}
 \bm S_i^k
 =\frac{\M_i^{\mathrm p}-\fvec_i^k}{\tau_i^{\mathrm p}}
  -\bm{\mathcal D}_i^k.
 \label{eq:implicit-source}
\end{equation}

To solve for the correction in the discrete basis, define
\begin{equation}
 \dvec_i
 =\mathbf P^{\mathsf T}\mathbf\Omega\Delta\fvec_i^{\mathrm v}
 =\begin{pmatrix}\dvec_i^h\\\dvec_i^b\end{pmatrix}
 \in\mathbb R^{2N_{\mathrm{vel}}}.
 \label{eq:increment-notation}
\end{equation}
The outer index $k$ is suppressed on the correction and the local vectors, matrices, and ordering of the retained basis indices used in the microscopic solve. Using Section~\ref{sec:boundaries} and $\Delta\fvec_i^{\mathrm v}=\mathbf P\dvec_i$, the projected velocity product is
\begin{equation}
 \mathbf P^{\mathsf T}\mathbf\Omega
  \bigl(\mathbf D_{ij}^{\pm}\Delta\fvec_i^{\mathrm v}\bigr)
 =\mathbf J_{ij}^{\pm}\dvec_i.
 \label{eq:signed-modal-operators}
\end{equation}
Projecting Eq.~\eqref{eq:implicit} gives
\begin{equation}
 \kappa_i\dvec_i
 +\frac{1}{|\Omega_i|}\sum_{j\in\mathcal J_i}A_{ij}
        \left(\mathbf J_{ij}^{+}\dvec_i
                       +\mathbf J_{ij}^{-}\dvec_j\right)
 =\hat{\bm S}_i^k.
 \label{eq:local-implicit}
\end{equation}
Separating the current distribution function as $\fvec_i^k=\M_i^k+\rvec_i^k$ decomposes the implicit source into nonequilibrium terms and an equilibrium source contribution,
\begin{equation}
 \hat{\bm S}_i^k
 =-\frac{\avec_i^k}{\tau_i^{\mathrm p}}
   -\hat{\bm{\mathcal D}}_{i,\mathrm{neq}}^k
   +\frac{\hat{\M}_i^{\mathrm p}-\hat{\M}_i^k}{\tau_i^{\mathrm p}}
   -\hat{\bm{\mathcal D}}_{i,\mathrm{eq}}^k.
 \label{eq:modal-source-decomposition}
\end{equation}
The divergences are assembled by Eqs.~\eqref{eq:candidate-flux-divergence}--\eqref{eq:source-equilibrium-divergence} at iteration $k$. The Maxwellian coefficients at $\W_i^{\mathrm p}$ and $\W_i^k$ follow from Eq.~\eqref{eq:separable-coefficients}.

At a diffuse wall with fixed velocity and temperature, the outgoing correction $\dvec^{\mathrm{out}}$ determines the incoming correction and wall flux increment through Eq.~\eqref{eq:wall}:
\begin{equation}
 \begin{gathered}
 \dvec^{\mathrm{in}}
 =-\frac{\displaystyle\sum_{q=1}^{4}\sum_{(m_1,m_2)\in\mathcal V_q}
          d_{(m_1,m_2)}^{q,h,\mathrm{out}}\,\ell_{(m_1,m_2)}^q(u_{ij}^{+})}
         {\displaystyle\sum_{q=1}^{4}\sum_{(m_1,m_2)\in\mathcal V_q}
          \hat M_{w,(m_1,m_2)}^{q,h}\,\ell_{(m_1,m_2)}^q(u_{ij}^{-})}\hat{\M}_w,\\
 \delta\hat{\Flux}_{ij}
 =\mathbf J_{ij}^{+}\dvec^{\mathrm{out}}
  +\mathbf J_{ij}^{-}\dvec^{\mathrm{in}}.
 \end{gathered}
 \label{eq:implicit-wall-increment}
\end{equation}
Here $\dvec^{\mathrm{in}}$ is the correction to the wall Maxwellian on incoming velocities.

\subsection{Implicit source block truncation and microscopic correction}\label{sec:implicit-restriction}
At each outer iteration, assemble the implicit source coefficients, truncate the source blocks, then solve for the microscopic correction on the union of the blocks retained for $h$ and $b$.

\emph{(i) Source block assembly.}
Let $\mathcal K_i^{k-1,q}$ denote the blocks used for the microscopic correction at the preceding outer iteration. Together with the current nonequilibrium transport candidates, they give the source block candidates
\[
 \mathcal P_{S,i}^{k,q}
 =\begin{cases}
    \mathfrak B^q,&k=0,\\[3pt]
    \displaystyle\mathcal K_i^{k-1,q}
       \cup\bigcup_{s\in\{h,b\}}\mathcal C_{i,\mathrm{neq}}^{k,q,s},
       &k\ge1.
   \end{cases}
\]
Assemble $(\hat{\bm S}_i^{k,q,s})_B$ by Eq.~\eqref{eq:modal-source-decomposition} for every $B\in\mathcal P_{S,i}^{k,q}$ and $s\in\{h,b\}$.

\emph{(ii) Source block truncation.}
The energy of each source block is
\[
 E_{S,i,B}^{k,q,s}
 =w_s(\theta_i^{\mathrm p})
       \|(\hat{\bm S}_i^{k,q,s})_B\|_2^2,
 \qquad B\in\mathcal P_{S,i}^{k,q},
\]
with the weights defined in Section~\ref{sec:adaptive-truncation}(i). The squared source block truncation budget is
\begin{equation}
 \delta_{S,i}^2
 =\left(\kappa_i\varepsilon_S^{\mathrm{ref}}
      \|\hat{\M}_i^{\mathrm p}\|_{\theta_i^{\mathrm p}}\right)^2.
 \label{eq:source-tolerance}
\end{equation}
Here $\varepsilon_S^{\mathrm{ref}}=\varepsilon_{S,h}^{\mathrm{ref}}=\varepsilon_{S,b}^{\mathrm{ref}}$.
Apply the block truncation rule in Section~\ref{sec:adaptive-truncation}(ii) to $E_{S,i,B}^{k,q,s}$. For $s\in\{h,b\}$, the retained source block sets satisfy
\[
 B_{00}^q\in\mathcal K_i^{k,q,s}\subseteq\mathcal P_{S,i}^{k,q},
 \qquad
 \sum_{q,s}\sum_{B\in\mathcal P_{S,i}^{k,q}\setminus\mathcal K_i^{k,q,s}}
    E_{S,i,B}^{k,q,s}\le\delta_{S,i}^2.
\]
The truncated source coefficients are
\[
 (\hat{\bm S}_i^{\mathrm{tr},q,s})_B
 =\begin{cases}
    (\hat{\bm S}_i^{k,q,s})_B,&B\in\mathcal K_i^{k,q,s},\\
    \bm0,&\text{otherwise}.
   \end{cases}
\]

\emph{(iii) Microscopic correction.}
For the reduced distributions $h$ and $b$, use
\[
 \mathcal K_i^{k,q}=\mathcal K_i^{k,q,h}\cup\mathcal K_i^{k,q,b},
 \qquad
 (\dvec_i^{q,s})_B=\bm0\quad\text{for }B\notin\mathcal K_i^{k,q},
 \quad s\in\{h,b\}.
\]
For each quadrant $q$, list the retained degree pairs in the ordering of Section~\ref{sec:modal-projection}:
\begin{equation}
 \{(m_{1,i,a},m_{2,i,a}):a=1,\ldots,n_i^q\}
 =\bigcup_{B\in\mathcal K_i^{k,q}}B,
 \label{eq:retained-dof-ordering}
\end{equation}
Here $m_{1,i,a}$ and $m_{2,i,a}$ are the $u$ and $v$ degrees of the $a$th retained basis vector, and $n_i^q$ is the number of retained degree pairs. The restricted vectors have components
\[
 (\dvec_{i,\mathcal K}^{q,s})_a=d_{i,(m_{1,i,a},m_{2,i,a})}^{q,s},
 \qquad
 (\hat{\bm S}_{i,\mathcal K}^{q,s})_a
 =\hat S_{i,(m_{1,i,a},m_{2,i,a})}^{\mathrm{tr},q,s},
 \qquad a=1,\ldots,n_i^q.
\]
For a Cartesian interior cell, restricting Eq.~\eqref{eq:local-implicit} to these DOF gives
\begin{equation}
 \mathbf H_{ii}^q\dvec_{i,\mathcal K}^{q,s}
 +\sum_{j\in\mathcal J_i}\mathbf H_{ij}^q\dvec_{j,\mathcal K}^{q,s}
 =\hat{\bm S}_{i,\mathcal K}^{q,s},
 \label{eq:restricted-system}
\end{equation}
where
\begin{equation}
 \begin{aligned}
 (\mathbf H_{ii}^q)_{ab}
 &=\kappa_i\delta_{ab}
  +\frac{1}{|\Omega_i|}\sum_{j\in\mathcal J_i}A_{ij}
       (\mathbf J_{ij}^{+,q})_{(m_{1,i,a},m_{2,i,a}),(m_{1,i,b},m_{2,i,b})},
 &\quad 1\le a,b\le n_i^q,\\
 (\mathbf H_{ij}^q)_{ab}
 &=\frac{A_{ij}}{|\Omega_i|}
       (\mathbf J_{ij}^{-,q})_{(m_{1,i,a},m_{2,i,a}),(m_{1,j,b},m_{2,j,b})},
 &\quad 1\le a\le n_i^q,\quad 1\le b\le n_j^q.
 \end{aligned}
 \label{eq:implicit-cell-operators}
\end{equation}
Here $\delta_{ab}$ is the Kronecker delta, and $\mathbf J_{ij}^{\pm,q}\in\mathbb R^{N_q\times N_q}$ is the quadrant operator defined in Section~\ref{sec:boundaries}, identical for $h$ and $b$. Its entries follow from Eq.~\eqref{eq:modal-products} and the signed split.

We solve Eq.~\eqref{eq:restricted-system} by Gauss--Seidel (GS) sweeps over the physical cells. Let $\mathcal L_i^q$ and $\mathcal U_i^q$ denote the neighbors preceding and following cell $i$ in the sweep order. Sweep $\ell$ gives
\begin{equation}
 \begin{aligned}
 \mathbf H_{ii}^q\dvec_{i,\mathcal K}^{q,s,(\ell+1)}
 ={}&\hat{\bm S}_{i,\mathcal K}^{q,s}\\
 &-\sum_{j\in\mathcal L_i^q}\mathbf H_{ij}^q
          \dvec_{j,\mathcal K}^{q,s,(\ell+1)}
  -\sum_{j\in\mathcal U_i^q}\mathbf H_{ij}^q
          \dvec_{j,\mathcal K}^{q,s,(\ell)}.
 \end{aligned}
 \label{eq:micro-spatial-gs}
\end{equation}
For a Cartesian sweep aligned with the velocities in $Q_q$, the interior upwind couplings satisfy
\[
 \mathbf H_{ij}^q=\mathbf0,
 \quad j\in\mathcal U_i^q.
\]

\begin{samepage}
Within each cell, LU-SGS~\cite{zhu2016implicit} solves the system in Eq.~\eqref{eq:micro-spatial-gs}. Decompose
\[
 \mathbf H_{ii}^q=\mathbf D_i^q+\mathbf L_i^q+\mathbf U_i^q
\]
into its diagonal, strictly lower triangular, and strictly upper triangular parts. For fixed sweep $\ell$, let $\bm b_i^{q,s,(\ell)}$ denote the right hand side of Eq.~\eqref{eq:micro-spatial-gs}. Local iteration $r$ consists of the forward and backward sweeps
\begin{equation}
 (\mathbf D_i^q+\mathbf L_i^q)\dvec_{i,\mathcal K}^{q,s,(r+1/2)}
 =\bm b_i^{q,s,(\ell)}
  -\mathbf U_i^q\dvec_{i,\mathcal K}^{q,s,(r)}.
 \label{eq:micro-local-sgs-forward}
\end{equation}
\begin{equation}
 (\mathbf D_i^q+\mathbf U_i^q)\dvec_{i,\mathcal K}^{q,s,(r+1)}
 =\mathbf D_i^q\dvec_{i,\mathcal K}^{q,s,(r+1/2)}
  +\mathbf U_i^q\dvec_{i,\mathcal K}^{q,s,(r)}.
 \label{eq:micro-local-sgs-backward}
\end{equation}
\end{samepage}
Each spatial GS sweep uses two local LU-SGS iterations per cell, and two spatial sweeps are performed in each quadrant.

\subsection{Update of nonequilibrium DOF}\label{sec:implicit-update}
\begin{samepage}
\emph{(i) Nonequilibrium update.}
The computed correction $\dvec_i$ updates the macroscopic variables through
\begin{equation}
 \W_i^{\dagger}=\W_i^k+\Cmat\dvec_i.
 \label{eq:implicit-macro}
\end{equation}
Define $\M_i^{\dagger}=\M(\W_i^{\dagger})$. The nonequilibrium DOF relative to this equilibrium are
\begin{equation}
 \avec_i^{\ast}
 =\avec_i^k+\dvec_i+\hat{\M}_i^k-\hat{\M}_i^{\dagger}.
 \label{eq:implicit-recenter}
\end{equation}
\end{samepage}
For the adaptive update, evaluate Eq.~\eqref{eq:implicit-recenter} on the union of the stored nonequilibrium blocks and the correction blocks:
\[
 \mathbf A_{i,B}^{\ast}
 =\mathbf A_{i,B}^{k}+(\dvec_i^{q,s})_B
  +\widehat{\mathbf M}_{i,B}^{k}-\widehat{\mathbf M}_{i,B}^{\dagger},
 \qquad B\in\mathcal B_i^{k,q,s}\cup\mathcal K_i^{k,q}.
\]
Here $\mathcal B_i^{k,q,s}$ is the stored nonequilibrium block set at the start of iteration $k$.

\emph{(ii) Nonequilibrium block truncation.}
The weighted block energy is
\[
 E_{i,B}^{k,q,s}
 =w_s(\theta_i^{\dagger})\|\mathbf A_{i,B}^{\ast}\|_2^2.
\]
The squared truncation budget is
\[
 \delta_i^2=(\varepsilon^{\mathrm{ref}})^2
      \|\hat{\M}_i^{\dagger}\|_{\theta_i^{\dagger}}^2.
\]
Apply the block truncation rule in Section~\ref{sec:adaptive-truncation}(ii) to $E_{i,B}^{k,q,s}$. The retained block sets satisfy
\[
 B_{00}^q\in\mathcal B_i^{k+1,q,s}
 \subseteq\mathcal B_i^{k,q,s}\cup\mathcal K_i^{k,q},
 \qquad s\in\{h,b\}.
\]
The discarded energy satisfies
\[
 \sum_{q,s}
 \sum_{B\in(\mathcal B_i^{k,q,s}\cup\mathcal K_i^{k,q})
                  \setminus\mathcal B_i^{k+1,q,s}}
     E_{i,B}^{k,q,s}\le\delta_i^2.
\]
The truncated coefficients are
\[
 \mathbf A_{i,B}^{\mathrm{tr},q,s}
 =\begin{cases}
    \mathbf A_{i,B}^{\ast},&B\in\mathcal B_i^{k+1,q,s},\\
    \bm0,&\text{otherwise}.
   \end{cases}
\]
By Eq.~\eqref{eq:local-candidate}, the retained nonequilibrium blocks enter the next source candidates through
\[
 \mathcal B_i^{k+1,q,s}
 \subseteq\mathcal C_{i,\mathrm{neq}}^{k+1,q,s}
 \subseteq\mathcal P_{S,i}^{k+1,q}.
\]

\emph{(iii) Conservative correction.}
Compute the moment residual $\bm r_i^{\mathrm c}=\Cmat\avec_i^{\mathrm{tr}}$ and solve Eq.~\eqref{eq:conservative-correction} with $\vartheta_i=\theta_i^{\dagger}$ for the constant-coefficient increments $\delta a_h^q$ and $\delta a_b^q$. Form the correction vector $\delta\avec_i$ by
\[
 \delta a_{i,(m_1,m_2)}^{q,s}
 =\begin{cases}
    \delta a_s^q,&(m_1,m_2)=(0,0),\\
    0,&\text{otherwise},
   \end{cases}
 \qquad
 \Cmat\delta\avec_i=-\bm r_i^{\mathrm c}.
\]
These corrections lie in the retained blocks $B_{00}^q$, so the block sets $\mathcal B_i^{k+1,q,s}$ are unchanged. The next iterate is
\[
 \W_i^{k+1}=\W_i^{\dagger},\qquad
 \avec_i^{k+1}=\avec_i^{\mathrm{tr}}+\delta\avec_i,\qquad
 \Cmat\avec_i^{k+1}=\bm0.
\]

Algorithm~\ref{alg:implicit-ugks} summarizes the iteration and includes optional multi step macro prediction~\cite{liu2026waveparticleII} and conservative Anderson acceleration.

\begin{algorithm}[H]
\caption{Nonequilibrium capturing implicit UGKS (NCIUGKS)}\label{alg:implicit-ugks}
\begin{algorithmic}[1]
\Require Grids and boundary conditions; initial $\W_i^0$, $\avec_i^0$, and block sets $\mathcal B_i^{0,q,s}$ containing $B_{00}^q$; $\Delta t$, the implicit CFL schedule, $\varepsilon_S^{\mathrm{ref}}$, $\varepsilon^{\mathrm{ref}}$, and convergence parameters.
\State Enforce $\Cmat\avec_i^0=\bm0$ by Eq.~\eqref{eq:conservative-correction}, with $\vartheta_i=\theta_i^0$.
\For{$k=0,1,\ldots$}
  \State Construct $\mathcal C_{ij,\mathrm{neq}}^{k,q,s}$ and $\mathcal C_{i,\mathrm{neq}}^{k,q,s}$ from $\{\mathcal B_\ell^{k,q,s}\}$ as in Section~\ref{sec:candidates}.
  \State Evaluate the microscopic UGKS fluxes $\hat{\Flux}_{ij}^k$ with $\Delta t$ using Eq.~\eqref{eq:full-modal-flux} and the boundary conditions.
  \State Compute the macroscopic fluxes $\mFlux_{ij}^k=\Cmat\hat{\Flux}_{ij}^k$, the divergence $\bm{\mathcal D}_i^k$, and the residuals~\eqref{eq:steady-transport-residual}.
  \If{$k>0$ and the prescribed convergence criterion is satisfied}
    \State \textbf{break}
  \EndIf
  \State Determine $\Delta t_i^{\mathrm{num}}$ and obtain $\W_i^{\mathrm p}$ by LU-SGS iterations of Eq.~\eqref{eq:macro-prediction}.
  \State Evaluate $\M_i^{\mathrm p}$, $\tau_i^{\mathrm p}$, and $\kappa_i$ from Eqs.~\eqref{eq:predicted-equilibrium} and~\eqref{eq:implicit}.
  \State Construct $\mathcal P_{S,i}^{k,q}$ as in Section~\ref{sec:implicit-restriction}(i), with $\mathcal P_{S,i}^{0,q}=\mathfrak B^q$. Assemble $(\hat{\bm S}_i^{k,q,s})_B$ by Eq.~\eqref{eq:modal-source-decomposition} for $B\in\mathcal P_{S,i}^{k,q}$ and $s\in\{h,b\}$.
  \State Truncate the source blocks in $\mathcal P_{S,i}^{k,q}$ with budget~\eqref{eq:source-tolerance}, as in Section~\ref{sec:implicit-restriction}(ii), to obtain $\mathcal K_i^{k,q,s}$ and $\hat{\bm S}_i^{\mathrm{tr}}$.
  \State Set $\mathcal K_i^{k,q}=\mathcal K_i^{k,q,h}\cup\mathcal K_i^{k,q,b}$ for the microscopic correction and the next source candidate construction.
  \State Solve Eq.~\eqref{eq:restricted-system} on $\mathcal K_i^{k,q}$ for $\dvec_i^h$ and $\dvec_i^b$ by GS with local LU-SGS, using $\hat{\bm S}_i^{\mathrm{tr}}$, imposing the boundary corrections, and setting correction coefficients outside $\mathcal K_i^{k,q}$ to zero.
  \State Update $\W_i^{\dagger}$ by Eq.~\eqref{eq:implicit-macro} and evaluate $\M_i^{\dagger}=\M(\W_i^{\dagger})$.
  \State Compute $\mathbf A_{i,B}^{\ast}$ from Eq.~\eqref{eq:implicit-recenter} for $B\in\mathcal B_i^{k,q,s}\cup\mathcal K_i^{k,q}$.
  \State Truncate these nonequilibrium blocks as in Section~\ref{sec:implicit-update}(ii), using $E_{i,B}^{k,q,s}$ and $\delta_i^2$ at $\W_i^{\dagger}$ and retaining $B_{00}^q$; obtain $\avec_i^{\mathrm{tr}}$ and $\mathcal B_i^{k+1,q,s}$.
  \State Set $\bm r_i^{\mathrm c}=\Cmat\avec_i^{\mathrm{tr}}$ and solve Eq.~\eqref{eq:conservative-correction} with $\vartheta_i=\theta_i^{\dagger}$. Set $\avec_i^{k+1}=\avec_i^{\mathrm{tr}}+\delta\avec_i$ and $\W_i^{k+1}=\W_i^{\dagger}$.
  \State \textit{Optional:} Apply multi step macro prediction~\cite{liu2026waveparticleII} to $(\W^{k+1},\avec^{k+1})$.
  \State \textit{Optional:} Apply conservative Anderson acceleration (\ref{app:implicit}) and update the block sets.
\EndFor
\Ensure Converged $\W_i^k$, $\avec_i^k$, and retained block sets $\mathcal B_i^{k,q,s}$.
\end{algorithmic}
\end{algorithm}
\FloatBarrier

\section{Numerical results}\label{sec:results}
The numerical examples cover low speed and hypersonic flows across continuum and rarefied regimes. The explicit formulation is examined through shock tube and shock structure calculations, followed by steady cavity flow and flow past a square cylinder using the implicit formulation. Together, these examples illustrate how the scheme captures nonequilibrium with fewer kinetic DOF while maintaining accuracy.

Flow velocity components are denoted by $(U_x,U_y)=(U,V)$. The two dimensional calculations use a single CPU core and a $100\times100$ midpoint quadrature in velocity space. Cost comparisons use full nodal UGKS and, for steady flows, its implicit formulation (IUGKS), on the same physical and velocity grids. Retained DOF are expressed as a percentage of the full nodal count for both reduced distribution functions. For steady flows, $R_{\Delta\rho}$ is the root mean square of successive density differences over fluid cells.

The Prandtl number is $\Pr=1$ throughout, and the specific heat ratio is $\gamma=5/3$ unless otherwise stated. The viscosity follows the temperature power law in Section~\ref{sec:distribution-notation}, with $\omega=0.72$ for the argon shock and $0.81$ for the other tests. For the argon shock, rarefied cavity, and square cylinder flow, the reference viscosity follows the variable hard sphere (VHS) relation~\cite{jiang2019parallel}
\[
 \mu_{\rm ref}=
 \frac{15\rho_{\rm ref}\lambda_{\rm ref}\sqrt{2\pi RT_{\rm ref}}}
 {2(5-2\omega_{\rm ref})(7-2\omega_{\rm ref})}.
\]
Here $\rho_{\rm ref}$ and $\lambda_{\rm ref}$ are the density and mean free path at the reference state, and $\omega_{\rm ref}=0.81$ for the square cylinder and $0.5$ for the argon shock and rarefied cavity. The square cylinder uses the free stream reference state, with $RT_{\rm ref}=1$. The argon shock and cavity use the upstream state and the initial equilibrium state at $T_w$, respectively, with $RT_{\rm ref}=0.5$.

For the cavity at $\mathrm{Re}=1000$, $\mu_{\rm ref}=\rho_0U_{\rm lid}L/\mathrm{Re}$, where $\rho_0$, $U_{\rm lid}$, and $L$ are the initial density, lid speed, and cavity side length.

All steady calculations use an explicit CFL number of $\mathrm{CFL}_{\mathrm{exp}}=0.5$. The implicit CFL is $\min((k+1)^{1.2},100)$ on coarse meshes and $\min((k+6)^{1.2},100)$ on the final fine mesh, with $k=0,1,\ldots$ on each mesh.

\subsection{Sod shock tube}\label{sec:sod}
The Sod shock tube illustrates DOF adaptation as the flow changes from hydrodynamic waves to strongly rarefied flow. The Knudsen numbers are $\mathrm{Kn}=10^{-4},10^{-3},10^{-2},0.1,1$, and $10$. The specific heat ratio is $\gamma=7/5$, and the initial variables are
\begin{equation}
 (\rho,U_x,p)(x,0)=
 \begin{cases}
 (1,0,1),&x<1/2,\\
 (0.125,0,0.1),&x>1/2,
 \end{cases}
 \label{eq:sod-initial}
\end{equation}
with zero transverse velocity. The reference viscosity is set by
\[
 \mu_{\rm ref}=
 \frac{15\sqrt{\pi}}{2(5-2\omega_{\rm ref})(7-2\omega_{\rm ref})}\,\mathrm{Kn},
 \qquad \omega_{\rm ref}=0.81,
\]
with $RT_{\rm ref}=1$ at the initial left state.

A uniform $100\times2$ mesh covers $[0,1]\times[0,1]$, with transmissive boundaries in $x$ and periodic boundaries in $y$. The velocity space domain is $[-7,7]^2$, and the truncation tolerances are $\varepsilon_h^{\mathrm{ref}}=\varepsilon_b^{\mathrm{ref}}=3\times10^{-5}$. Both methods use a CFL number of $0.8$. Results are compared at $t=0.12$, with temperature plotted as $T=2\theta$.

Figures~\ref{fig:sod-1em4}--\ref{fig:sod-10} illustrate the change in wave structure over the full range of $\mathrm{Kn}$. At small $\mathrm{Kn}$, the rarefaction wave, contact discontinuity, and shock remain distinct. As $\mathrm{Kn}$ increases, these features broaden into smooth kinetic profiles. Density, velocity, and temperature agree well with UGKS in all six cases.

\begin{figure}[p]
 \centering
 \begin{minipage}[t]{0.325\textwidth}
  \centering
  \includegraphics[width=\linewidth]{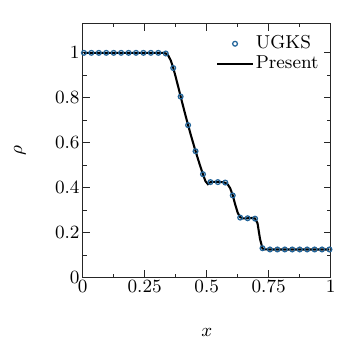}
  \panelcaption{(a) Density}
 \end{minipage}\hfill
 \begin{minipage}[t]{0.325\textwidth}
  \centering
  \includegraphics[width=\linewidth]{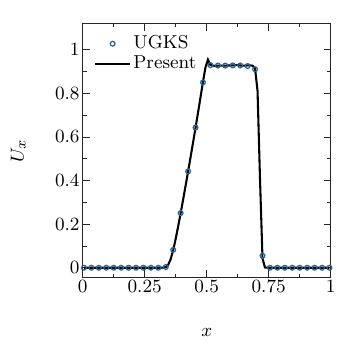}
  \panelcaption{(b) Velocity $U_x$}
 \end{minipage}\hfill
 \begin{minipage}[t]{0.325\textwidth}
  \centering
  \includegraphics[width=\linewidth]{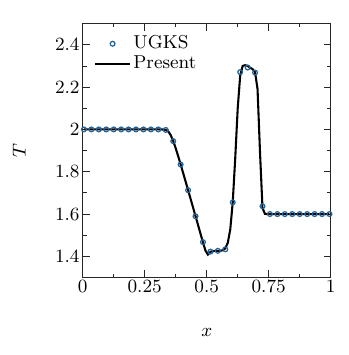}
  \panelcaption{(c) Temperature}
 \end{minipage}
 \caption{Sod shock tube at $\mathrm{Kn}=10^{-4}$.}
 \label{fig:sod-1em4}
\end{figure}

\begin{figure}[p]
 \centering
 \begin{minipage}[t]{0.325\textwidth}
  \centering
  \includegraphics[width=\linewidth]{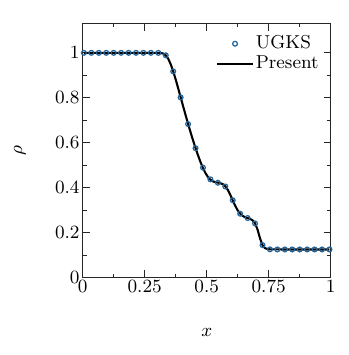}
  \panelcaption{(a) Density}
 \end{minipage}\hfill
 \begin{minipage}[t]{0.325\textwidth}
  \centering
  \includegraphics[width=\linewidth]{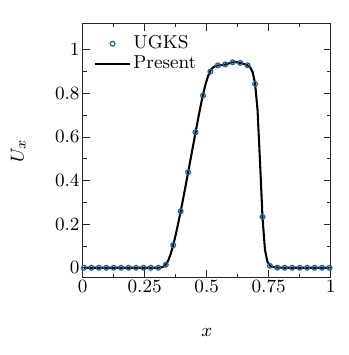}
  \panelcaption{(b) Velocity $U_x$}
 \end{minipage}\hfill
 \begin{minipage}[t]{0.325\textwidth}
  \centering
  \includegraphics[width=\linewidth]{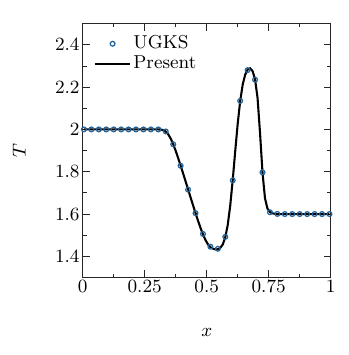}
  \panelcaption{(c) Temperature}
 \end{minipage}
 \caption{Sod shock tube at $\mathrm{Kn}=10^{-3}$.}
 \label{fig:sod-1em3}
\end{figure}

\begin{figure}[p]
 \centering
 \begin{minipage}[t]{0.325\textwidth}
  \centering
  \includegraphics[width=\linewidth]{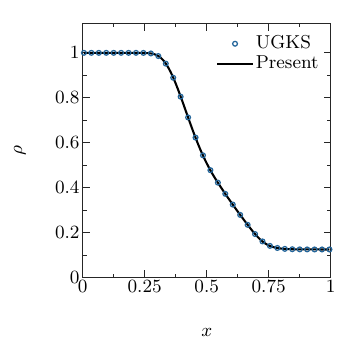}
  \panelcaption{(a) Density}
 \end{minipage}\hfill
 \begin{minipage}[t]{0.325\textwidth}
  \centering
  \includegraphics[width=\linewidth]{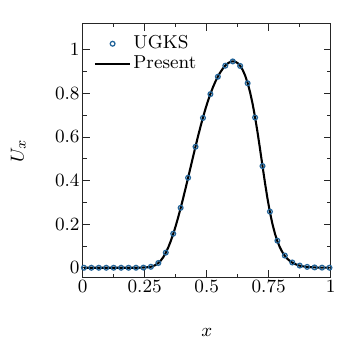}
  \panelcaption{(b) Velocity $U_x$}
 \end{minipage}\hfill
 \begin{minipage}[t]{0.325\textwidth}
  \centering
  \includegraphics[width=\linewidth]{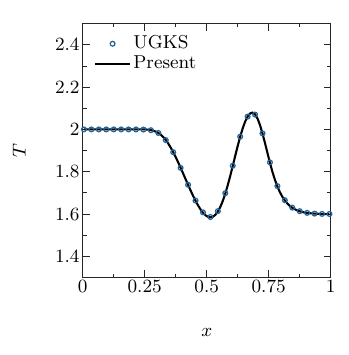}
  \panelcaption{(c) Temperature}
 \end{minipage}
 \caption{Sod shock tube at $\mathrm{Kn}=10^{-2}$.}
 \label{fig:sod-1em2}
\end{figure}

\begin{figure}[p]
 \centering
 \begin{minipage}[t]{0.325\textwidth}
  \centering
  \includegraphics[width=\linewidth]{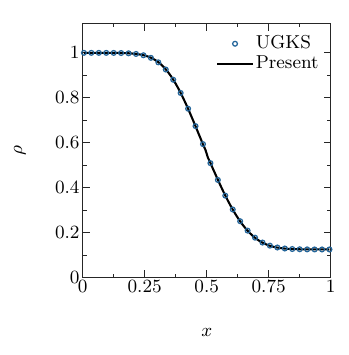}
  \panelcaption{(a) Density}
 \end{minipage}\hfill
 \begin{minipage}[t]{0.325\textwidth}
  \centering
  \includegraphics[width=\linewidth]{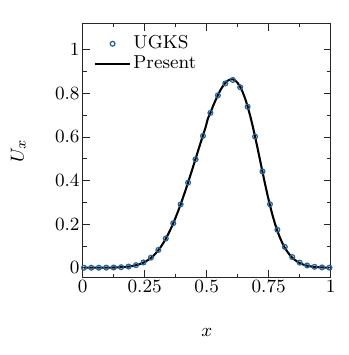}
  \panelcaption{(b) Velocity $U_x$}
 \end{minipage}\hfill
 \begin{minipage}[t]{0.325\textwidth}
  \centering
  \includegraphics[width=\linewidth]{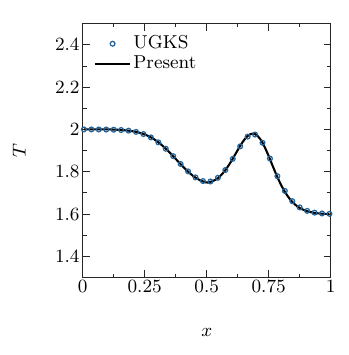}
  \panelcaption{(c) Temperature}
 \end{minipage}
 \caption{Sod shock tube at $\mathrm{Kn}=0.1$.}
 \label{fig:sod-1em1}
\end{figure}

\begin{figure}[p]
 \centering
 \begin{minipage}[t]{0.325\textwidth}
  \centering
  \includegraphics[width=\linewidth]{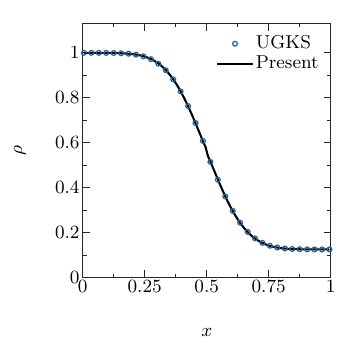}
  \panelcaption{(a) Density}
 \end{minipage}\hfill
 \begin{minipage}[t]{0.325\textwidth}
  \centering
  \includegraphics[width=\linewidth]{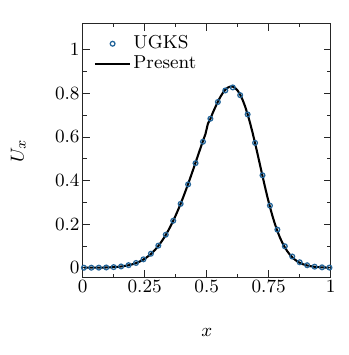}
  \panelcaption{(b) Velocity $U_x$}
 \end{minipage}\hfill
 \begin{minipage}[t]{0.325\textwidth}
  \centering
  \includegraphics[width=\linewidth]{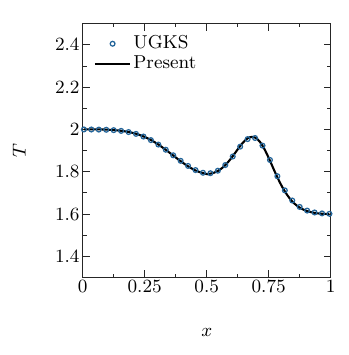}
  \panelcaption{(c) Temperature}
 \end{minipage}
 \caption{Sod shock tube at $\mathrm{Kn}=1$.}
 \label{fig:sod-1}
\end{figure}

\begin{figure}[p]
 \centering
 \begin{minipage}[t]{0.325\textwidth}
  \centering
  \includegraphics[width=\linewidth]{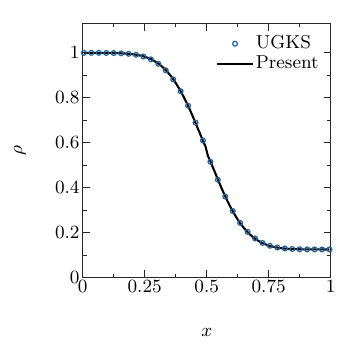}
  \panelcaption{(a) Density}
 \end{minipage}\hfill
 \begin{minipage}[t]{0.325\textwidth}
  \centering
  \includegraphics[width=\linewidth]{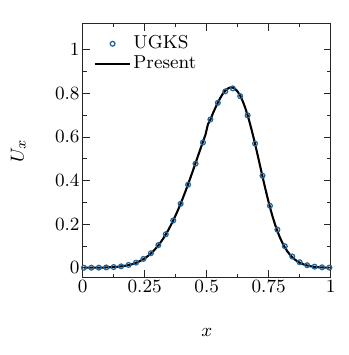}
  \panelcaption{(b) Velocity $U_x$}
 \end{minipage}\hfill
 \begin{minipage}[t]{0.325\textwidth}
  \centering
  \includegraphics[width=\linewidth]{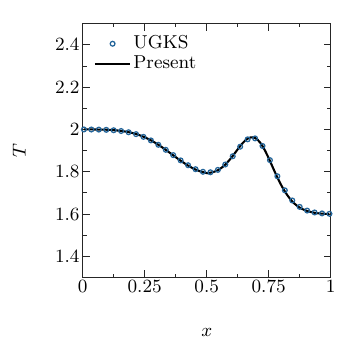}
  \panelcaption{(c) Temperature}
 \end{minipage}
 \caption{Sod shock tube at $\mathrm{Kn}=10$.}
 \label{fig:sod-10}
\end{figure}

Table~\ref{tab:sod-cost} summarizes the retained DOF and computational cost: this agreement is obtained with only $1.07\%$--$2.97\%$ of the full kinetic DOF. The retained fraction increases with $\mathrm{Kn}$ but remains below $3\%$ throughout the range. Using fewer DOF gives speedups of $26.63$--$58.50$ over UGKS. The largest speedup occurs at $\mathrm{Kn}=10^{-4}$, where the small nonequilibrium contribution requires the fewest DOF.

\begin{table}[!htbp]
 \centering
 \caption{Retained DOF and computational cost for the Sod shock tube at $t=0.12$. Time integration costs are medians of three calculations.}
 \label{tab:sod-cost}
 \begin{tabular}{ccrrr}
  \toprule
  $\mathrm{Kn}$ & Retained DOF (\%) & UGKS (s) & Present (s) & Speedup \\
  \midrule
  $10^{-4}$ & 1.07 & 56.16 & 0.96 & 58.50 \\
  $10^{-3}$ & 1.74 & 60.79 & 1.17 & 51.96 \\
  $10^{-2}$ & 2.54 & 55.63 & 1.66 & 33.51 \\
  $0.1$ & 2.90 & 55.90 & 2.03 & 27.54 \\
  $1$ & 2.97 & 56.19 & 2.11 & 26.63 \\
  $10$ & 2.97 & 57.39 & 2.12 & 27.07 \\
  \bottomrule
 \end{tabular}
\end{table}
\FloatBarrier

\subsection{Argon shock structure}\label{sec:shock}
The argon shock at $\mathrm{Ma}=8$ provides a more demanding test of nonequilibrium capturing. The upstream state, indexed by $1$, is $(\rho_1,U_{x,1},\theta_1)=(1,8\sqrt{5/6},0.5)$. The downstream state follows from the Rankine--Hugoniot relations. The $y$ velocity is zero. Lengths are normalized by the upstream hard sphere reference mean free path. The initial discontinuity lies at $x=0$. Upstream and downstream equilibria are prescribed at the left and right boundaries.

The domain $[-25,25]\times[0,50]$ uses a uniform $100\times2$ mesh with periodic boundaries in $y$. The velocity space domain is $[-15,15]^2$, the truncation tolerances are $\varepsilon_h^{\mathrm{ref}}=\varepsilon_b^{\mathrm{ref}}=10^{-4}$, and the CFL number is $0.95$. Results are compared at $t=250$.

The profiles are normalized as $Q_n=(Q-Q_1)/(Q_2-Q_1)$ for $Q\in\{\rho,U_x,T\}$, with subscripts $1$ and $2$ denoting the upstream and downstream states. The comparison in Fig.~\ref{fig:argon-shock} shows close agreement with UGKS through the shock transition for all three profiles.

\begin{figure}[!htbp]
 \centering
 \begin{minipage}[t]{0.325\textwidth}
  \centering
  \includegraphics[width=\linewidth]{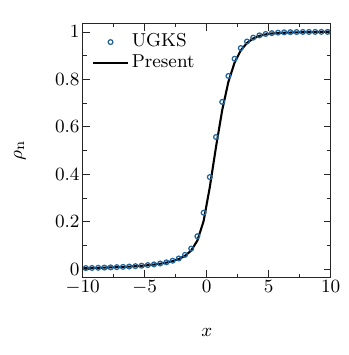}
  \panelcaption{(a) Density}
 \end{minipage}\hfill
 \begin{minipage}[t]{0.325\textwidth}
  \centering
  \includegraphics[width=\linewidth]{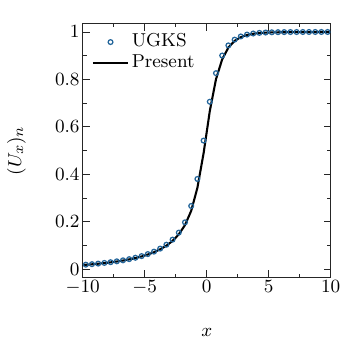}
  \panelcaption{(b) Velocity}
 \end{minipage}\hfill
 \begin{minipage}[t]{0.325\textwidth}
  \centering
  \includegraphics[width=\linewidth]{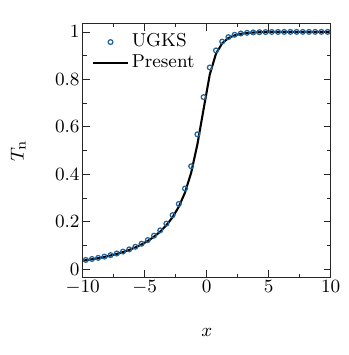}
  \panelcaption{(c) Temperature}
 \end{minipage}
 \caption{Normalized density, velocity, and temperature profiles across the $\mathrm{Ma}=8$ argon shock. Solid lines denote the present method; open circles denote UGKS.}
 \label{fig:argon-shock}
\end{figure}

The shock requires more DOF than the Sod tests, yet the retained fraction is only $7.47\%$ of the full kinetic DOF. Reaching $t=250$ takes $139$~s, giving a speedup of $12.64$ over full nodal UGKS.

An additional one dimensional calculation examines how successive partial sums of the discrete basis expansion approximate the bimodal distribution function within the shock. The setup and results are given in \ref{app:shock-distribution}.

\FloatBarrier

\subsection{Cavity flow driven by a moving lid}\label{sec:cavity}
Lid driven cavity flow is computed with the implicit formulation at $\mathrm{Re}=1000$ in the continuum regime and at $\mathrm{Kn}=0.1$ and $1$ in the rarefied regime. These cases illustrate how local DOF adaptation and steady convergence change with the level of nonequilibrium. The gas occupies a square of side length $L=1$ and is initially at rest with density $\rho_0=1$ and temperature $T_w$. All walls are held at $T_w$ and obey diffuse reflection, with the upper wall moving to the right at $U_{\rm lid}=0.148305$. Velocities are normalized by $\sqrt{2RT_w}$, giving $\theta_w=1/2$. For the rarefied cases, $\mathrm{Kn}=\lambda_{\rm ref}/L$.

The rarefied cases use uniform $64\times64$ fine meshes. At $\mathrm{Re}=1000$, a $96\times96$ mesh is refined near the walls, with a minimum cell length of $0.004L$. Each coarse mesh has half as many cells in each direction. The velocity domain is $[-6,6]^2$. The state and implicit source block truncation tolerances are $\varepsilon_s^{\mathrm{ref}}=\varepsilon_{S,s}^{\mathrm{ref}}=10^{-7}$ at $\mathrm{Re}=1000$ and $10^{-5}$ at $\mathrm{Kn}=0.1$ and $1$, for $s=h,b$.

The accelerated calculations begin with a coarse mesh warm start, as described in \ref{app:warm-start}. Once $R_{\Delta\rho}$ falls below $10^{-5}$ on the coarse mesh, the solution is transferred to the fine mesh. Multi step macro prediction and conservative Anderson acceleration are used on both meshes. The fine mesh stopping tolerance for $R_{\Delta\rho}$ is $10^{-7}$ at $\mathrm{Re}=1000$ and $10^{-6}$ at $\mathrm{Kn}=0.1$ and $1$.

Figure~\ref{fig:cavity-centerlines} compares the centerline velocity components with the benchmark of Ghia et al.~\cite{ghia1982high} at $\mathrm{Re}=1000$ and with UGKS in the rarefied cases. Both components agree well with their respective references.

\begin{figure}[!htbp]
 \centering
 \begin{minipage}[t]{0.32\textwidth}
  \centering
  \includegraphics[width=\linewidth]{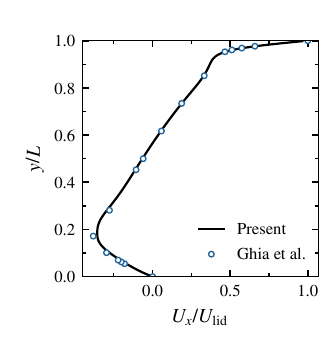}
  \panelcaption{(a) $\mathrm{Re}=1000$}
 \end{minipage}\hfill
 \begin{minipage}[t]{0.32\textwidth}
  \centering
  \includegraphics[width=\linewidth]{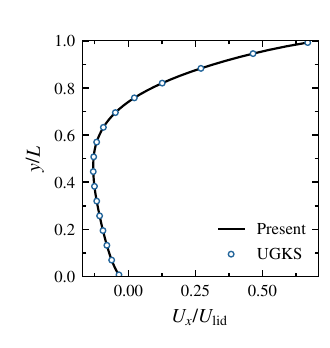}
  \panelcaption{(b) $\mathrm{Kn}=0.1$}
 \end{minipage}\hfill
 \begin{minipage}[t]{0.32\textwidth}
  \centering
  \includegraphics[width=\linewidth]{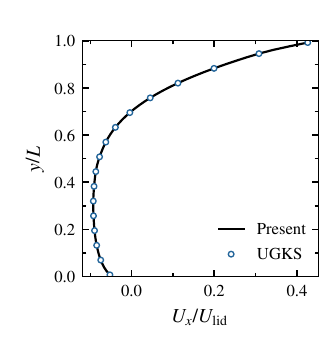}
  \panelcaption{(c) $\mathrm{Kn}=1$}
 \end{minipage}
\par\vspace{0.4em}
 \begin{minipage}[t]{0.32\textwidth}
  \centering
  \includegraphics[width=\linewidth]{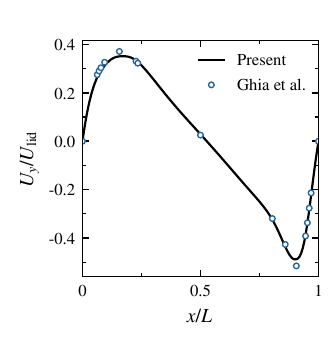}
  \panelcaption{(d) $\mathrm{Re}=1000$}
 \end{minipage}\hfill
 \begin{minipage}[t]{0.32\textwidth}
  \centering
  \includegraphics[width=\linewidth]{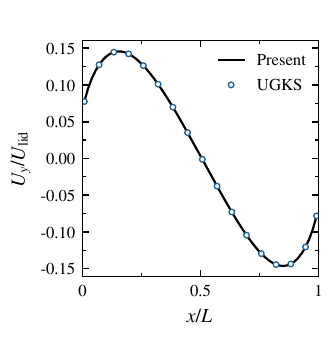}
  \panelcaption{(e) $\mathrm{Kn}=0.1$}
 \end{minipage}\hfill
 \begin{minipage}[t]{0.32\textwidth}
  \centering
  \includegraphics[width=\linewidth]{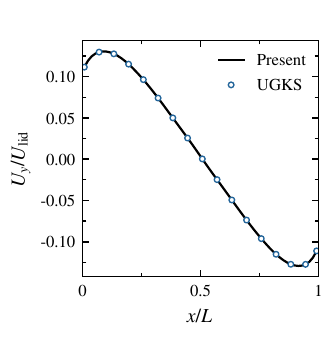}
  \panelcaption{(f) $\mathrm{Kn}=1$}
 \end{minipage}
 \caption{Cavity centerline velocities: (a)--(c) $U_x$ along $x/L=0.5$ and (d)--(f) $U_y$ along $y/L=0.5$. Black solid lines denote the present method, and blue open circles denote Ghia et al.~\cite{ghia1982high} at $\mathrm{Re}=1000$ and UGKS in the rarefied cases.}
 \label{fig:cavity-centerlines}
\end{figure}
\FloatBarrier

At $\mathrm{Re}=1000$, the streamlines in Fig.~\ref{fig:cavity-streamlines} reveal a dominant central vortex and two weaker vortices in the lower corners.

\begin{figure}[!htbp]
 \centering
 \includegraphics[width=0.70\textwidth]{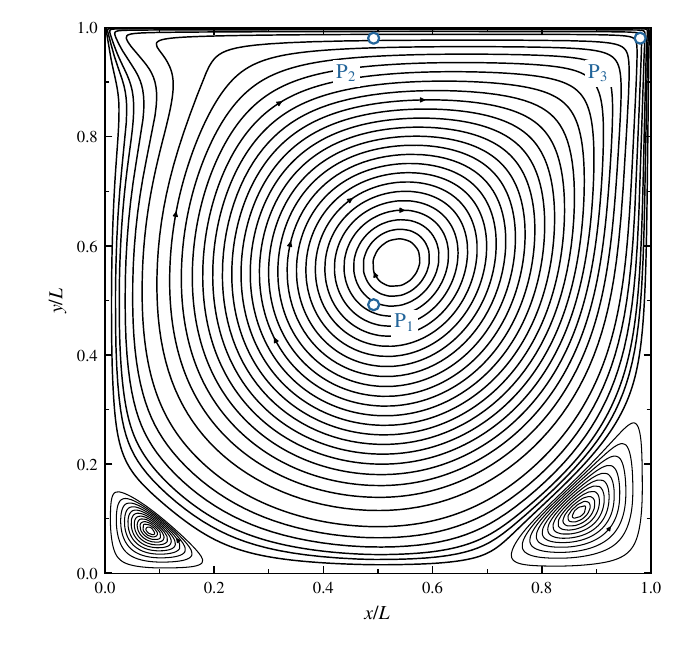}
 \caption{Streamlines of the steady mass flux in the cavity at $\mathrm{Re}=1000$. All three cavity cases use the cell centers nearest the sampling targets $(x/L,y/L)$: $\mathrm{P}_1=(0.50,0.50)$, $\mathrm{P}_2=(0.50,0.98)$, and $\mathrm{P}_3=(0.98,0.98)$.}
 \label{fig:cavity-streamlines}
\end{figure}
\FloatBarrier

\clearpage
The DOF map in Fig.~\ref{fig:cavity-re1000}(c) shows that few kinetic DOF are needed in the continuum cavity interior, with additional DOF concentrated near the two upper corners. The entire solution retains only $2.13\%$ of the full kinetic DOF.

\begin{figure}[!htbp]
 \centering
 \begin{minipage}[t]{0.32\textwidth}
  \centering
  \includegraphics[width=\linewidth]{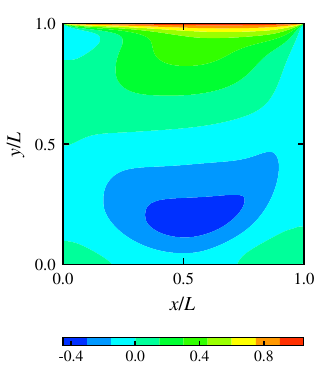}
  \panelcaption{(a) $U_x/U_{\rm lid}$}
 \end{minipage}\hfill
 \begin{minipage}[t]{0.32\textwidth}
  \centering
  \includegraphics[width=\linewidth]{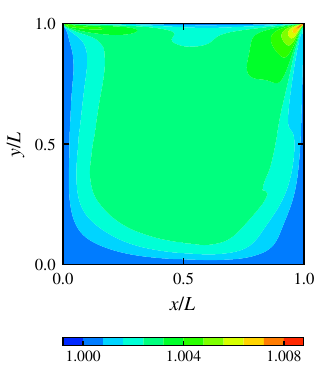}
  \panelcaption{(b) $T/T_w$}
 \end{minipage}\hfill
 \begin{minipage}[t]{0.32\textwidth}
  \centering
  \includegraphics[width=\linewidth]{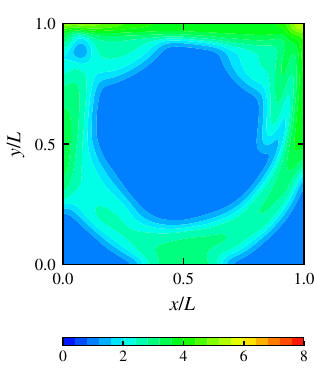}
  \panelcaption{(c) Retained DOF (\%)}
 \end{minipage}
 \par\medskip
 \begin{minipage}[t]{0.34\textwidth}
  \centering
  \includegraphics[width=\linewidth]{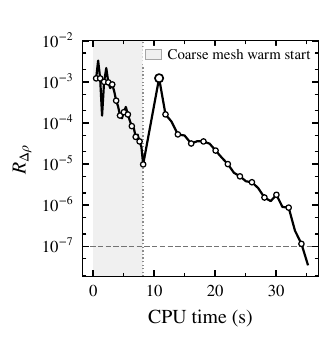}
  \panelcaption{(d) $R_{\Delta\rho}$}
 \end{minipage}\hspace{0.03\textwidth}
 \begin{minipage}[t]{0.34\textwidth}
  \centering
  \includegraphics[width=\linewidth]{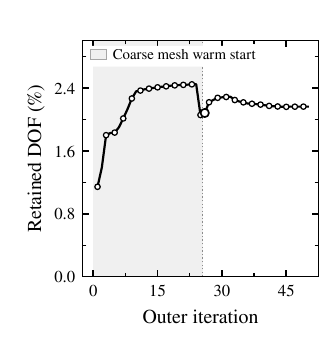}
  \panelcaption{(e) DOF adaptation}
 \end{minipage}
 \par\smallskip
 \caption{Velocity, temperature, retained DOF, and iteration histories for the $\mathrm{Re}=1000$ cavity. Shading marks the coarse mesh warm start; dotted vertical lines mark the transfer to the fine mesh. Larger circles identify the first fine mesh iteration, and the dashed horizontal line indicates the stopping tolerance.}
 \label{fig:cavity-re1000}
\end{figure}
\FloatBarrier

The distributions plotted in Fig.~\ref{fig:cavity-distributions} remain close to their local Maxwellians at all three sampling locations. The maximum nodal values of $|h-M_h|$ increase from $1.06\times10^{-6}$ at $\mathrm{P}_1$ to $6.70\times10^{-5}$ at $\mathrm{P}_2$ and $1.95\times10^{-4}$ at $\mathrm{P}_3$. This increase toward the moving lid and upper right corner is consistent with the greater local kinetic resolution in Fig.~\ref{fig:cavity-re1000}(c).

\begin{figure}[!htbp]
 \centering
 \begin{minipage}[t]{0.325\textwidth}
  \centering
  \includegraphics[width=\linewidth]{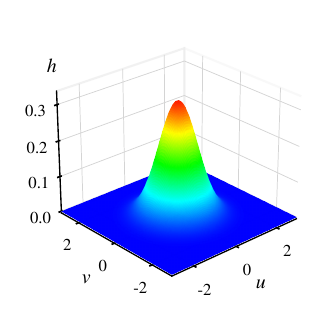}
  \panelcaption{(a) $\mathrm{P}_1$}
 \end{minipage}\hfill
 \begin{minipage}[t]{0.325\textwidth}
  \centering
  \includegraphics[width=\linewidth]{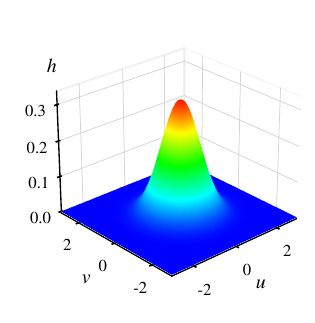}
  \panelcaption{(b) $\mathrm{P}_2$}
 \end{minipage}\hfill
 \begin{minipage}[t]{0.325\textwidth}
  \centering
  \includegraphics[width=\linewidth]{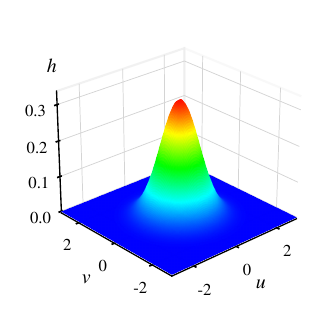}
  \panelcaption{(c) $\mathrm{P}_3$}
 \end{minipage}
 \par\medskip
 \begin{minipage}[t]{0.325\textwidth}
  \centering
  \includegraphics[width=\linewidth]{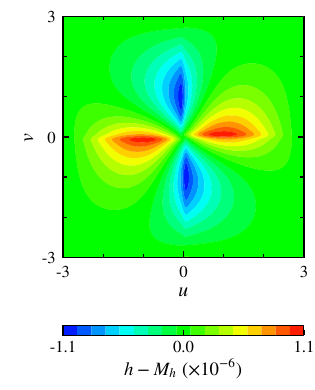}
  \panelcaption{(d) $\mathrm{P}_1$}
 \end{minipage}\hfill
 \begin{minipage}[t]{0.325\textwidth}
  \centering
  \includegraphics[width=\linewidth]{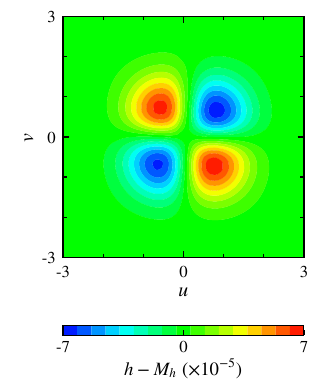}
  \panelcaption{(e) $\mathrm{P}_2$}
 \end{minipage}\hfill
 \begin{minipage}[t]{0.325\textwidth}
  \centering
  \includegraphics[width=\linewidth]{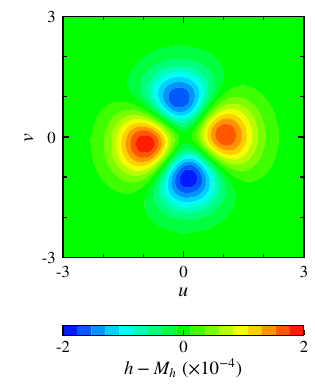}
  \panelcaption{(f) $\mathrm{P}_3$}
 \end{minipage}
 \caption{$h(u,v)$ (top) and $h-M_h$ (bottom) at the three sampling locations in the $\mathrm{Re}=1000$ cavity.}
 \label{fig:cavity-distributions}
\end{figure}
\FloatBarrier

The two rarefied cases use the same sampling targets marked in Fig.~\ref{fig:cavity-streamlines}. The DOF maps in Figs.~\ref{fig:cavity-kn01}(c) and~\ref{fig:cavity-kn1}(c) reflect the wider extent of nonequilibrium: more DOF are retained throughout the cavity as rarefaction increases. The retained fraction rises to $10.29\%$ at $\mathrm{Kn}=0.1$, yet remains at only $27.20\%$ even at $\mathrm{Kn}=1$.

\begin{figure}[!htbp]
 \centering
 \begin{minipage}[t]{0.32\textwidth}
  \centering
  \includegraphics[width=\linewidth]{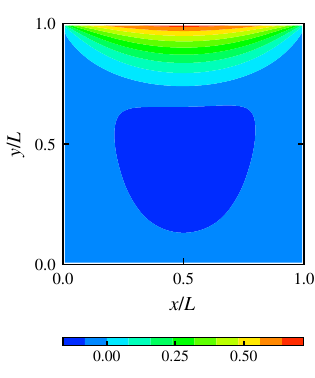}
  \panelcaption{(a) $U_x/U_{\rm lid}$}
 \end{minipage}\hfill
 \begin{minipage}[t]{0.32\textwidth}
  \centering
  \includegraphics[width=\linewidth]{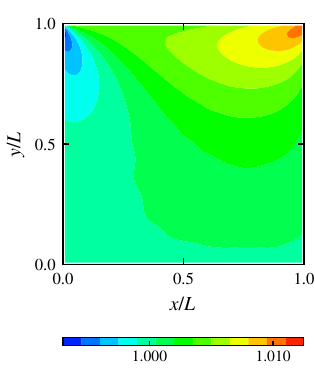}
  \panelcaption{(b) $T/T_w$}
 \end{minipage}\hfill
 \begin{minipage}[t]{0.32\textwidth}
  \centering
  \includegraphics[width=\linewidth]{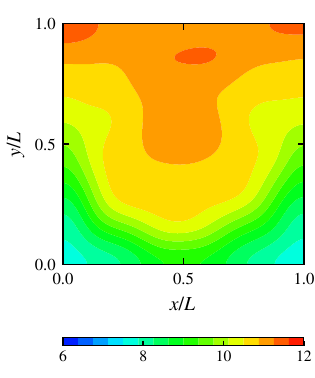}
  \panelcaption{(c) Retained DOF (\%)}
 \end{minipage}
 \par\medskip
 \begin{minipage}[t]{0.34\textwidth}
  \centering
  \includegraphics[width=\linewidth]{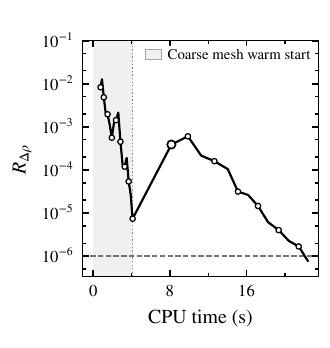}
  \panelcaption{(d) $R_{\Delta\rho}$}
 \end{minipage}\hspace{0.03\textwidth}
 \begin{minipage}[t]{0.34\textwidth}
  \centering
  \includegraphics[width=\linewidth]{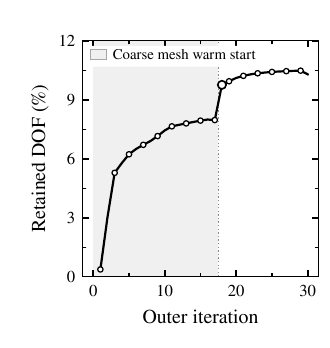}
  \panelcaption{(e) DOF adaptation}
 \end{minipage}
 \par\smallskip
 \caption{Velocity, temperature, retained DOF, and iteration histories for cavity flow at $\mathrm{Kn}=0.1$.}
 \label{fig:cavity-kn01}
\end{figure}
\FloatBarrier

At $\mathrm{Kn}=0.1$, the distributions depart visibly from local equilibrium, as seen in Fig.~\ref{fig:cavity-kn01-distributions}. The maximum nodal values of $|h-M_h|$ are $1.48\times10^{-3}$, $9.16\times10^{-3}$, and $1.68\times10^{-2}$ at $\mathrm{P}_1$, $\mathrm{P}_2$, and $\mathrm{P}_3$, respectively, with stronger departures near the moving lid and upper right corner.

\begin{figure}[!htbp]
 \centering
 \begin{minipage}[t]{0.325\textwidth}
  \centering
  \includegraphics[width=\linewidth]{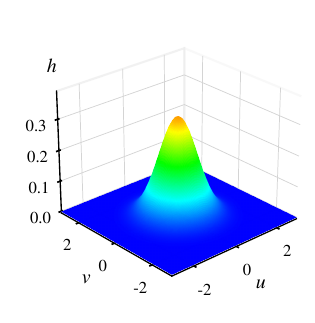}
  \panelcaption{(a) $\mathrm{P}_1$}
 \end{minipage}\hfill
 \begin{minipage}[t]{0.325\textwidth}
  \centering
  \includegraphics[width=\linewidth]{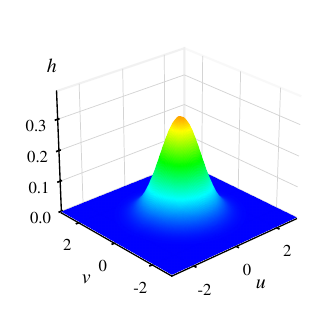}
  \panelcaption{(b) $\mathrm{P}_2$}
 \end{minipage}\hfill
 \begin{minipage}[t]{0.325\textwidth}
  \centering
  \includegraphics[width=\linewidth]{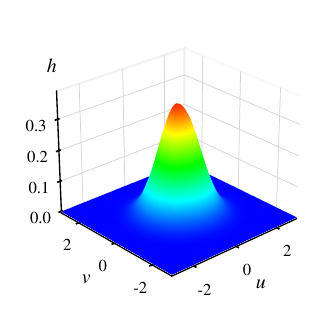}
  \panelcaption{(c) $\mathrm{P}_3$}
 \end{minipage}
 \par\medskip
 \begin{minipage}[t]{0.325\textwidth}
  \centering
  \includegraphics[width=\linewidth]{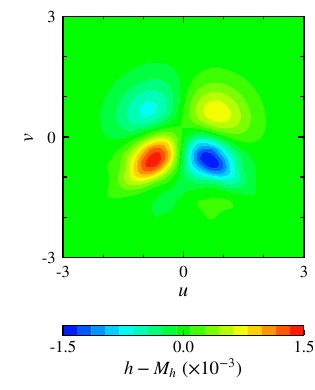}
  \panelcaption{(d) $\mathrm{P}_1$}
 \end{minipage}\hfill
 \begin{minipage}[t]{0.325\textwidth}
  \centering
  \includegraphics[width=\linewidth]{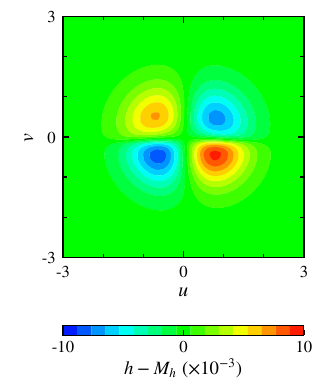}
  \panelcaption{(e) $\mathrm{P}_2$}
 \end{minipage}\hfill
 \begin{minipage}[t]{0.325\textwidth}
  \centering
  \includegraphics[width=\linewidth]{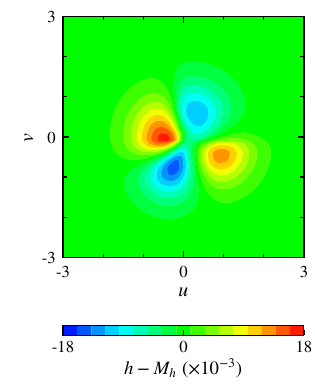}
  \panelcaption{(f) $\mathrm{P}_3$}
 \end{minipage}
 \caption{$h(u,v)$ (top) and $h-M_h$ (bottom) at the three sampling locations in the cavity at $\mathrm{Kn}=0.1$.}
 \label{fig:cavity-kn01-distributions}
\end{figure}
\FloatBarrier

\clearpage
The histories in Figs.~\ref{fig:cavity-kn01}(d,e) and~\ref{fig:cavity-kn1}(d,e) track convergence and DOF adaptation in the rarefied cases. After transfer to the fine mesh, the retained fraction increases and then levels off as $R_{\Delta\rho}$ approaches its stopping tolerance.

\begin{figure}[!htbp]
 \centering
 \begin{minipage}[t]{0.32\textwidth}
  \centering
  \includegraphics[width=\linewidth]{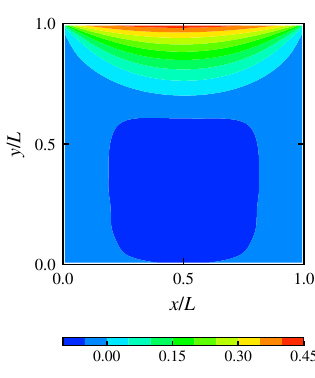}
  \panelcaption{(a) $U_x/U_{\rm lid}$}
 \end{minipage}\hfill
 \begin{minipage}[t]{0.32\textwidth}
  \centering
  \includegraphics[width=\linewidth]{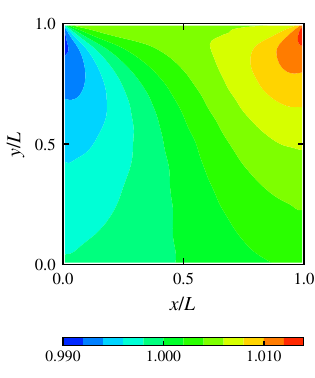}
  \panelcaption{(b) $T/T_w$}
 \end{minipage}\hfill
 \begin{minipage}[t]{0.32\textwidth}
  \centering
  \includegraphics[width=\linewidth]{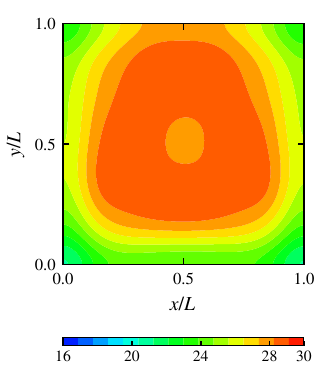}
  \panelcaption{(c) Retained DOF (\%)}
 \end{minipage}
 \par\medskip
 \begin{minipage}[t]{0.34\textwidth}
  \centering
  \includegraphics[width=\linewidth]{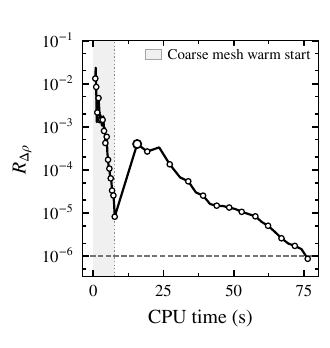}
  \panelcaption{(d) $R_{\Delta\rho}$}
 \end{minipage}\hspace{0.03\textwidth}
 \begin{minipage}[t]{0.34\textwidth}
  \centering
  \includegraphics[width=\linewidth]{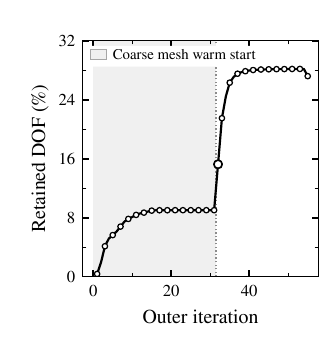}
  \panelcaption{(e) DOF adaptation}
 \end{minipage}
 \par\smallskip
 \caption{Velocity, temperature, retained DOF, and iteration histories for cavity flow at $\mathrm{Kn}=1$.}
 \label{fig:cavity-kn1}
\end{figure}
\FloatBarrier

Comparing the distributions in Figs.~\ref{fig:cavity-kn01-distributions} and~\ref{fig:cavity-kn1-distributions} reveals stronger departures from equilibrium at all three locations as $\mathrm{Kn}$ increases from $0.1$ to $1$. The maximum nodal values of $|h-M_h|$ reach $1.36\times10^{-2}$, $2.77\times10^{-2}$, and $2.31\times10^{-2}$ at $\mathrm{P}_1$, $\mathrm{P}_2$, and $\mathrm{P}_3$, respectively; the increase at $\mathrm{P}_1$ indicates stronger nonequilibrium near the cavity center.

\begin{figure}[!htbp]
 \centering
 \begin{minipage}[t]{0.325\textwidth}
  \centering
  \includegraphics[width=\linewidth]{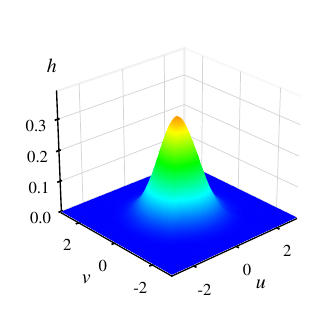}
  \panelcaption{(a) $\mathrm{P}_1$}
 \end{minipage}\hfill
 \begin{minipage}[t]{0.325\textwidth}
  \centering
  \includegraphics[width=\linewidth]{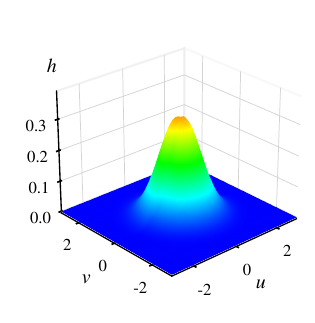}
  \panelcaption{(b) $\mathrm{P}_2$}
 \end{minipage}\hfill
 \begin{minipage}[t]{0.325\textwidth}
  \centering
  \includegraphics[width=\linewidth]{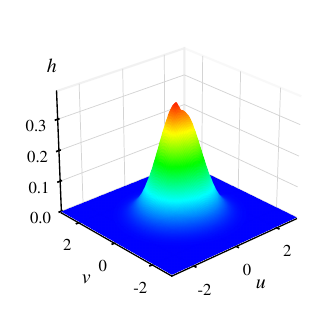}
  \panelcaption{(c) $\mathrm{P}_3$}
 \end{minipage}
 \par\medskip
 \begin{minipage}[t]{0.325\textwidth}
  \centering
  \includegraphics[width=\linewidth]{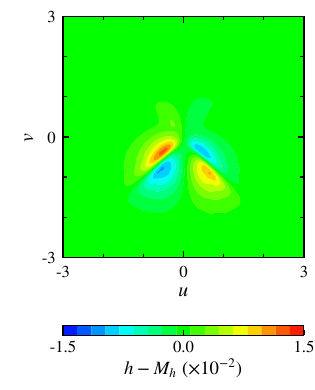}
  \panelcaption{(d) $\mathrm{P}_1$}
 \end{minipage}\hfill
 \begin{minipage}[t]{0.325\textwidth}
  \centering
  \includegraphics[width=\linewidth]{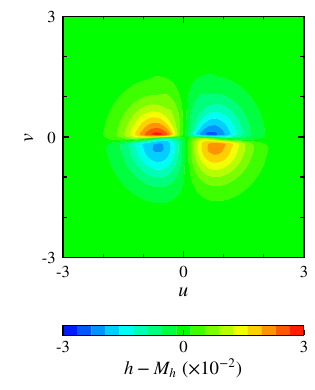}
  \panelcaption{(e) $\mathrm{P}_2$}
 \end{minipage}\hfill
 \begin{minipage}[t]{0.325\textwidth}
  \centering
  \includegraphics[width=\linewidth]{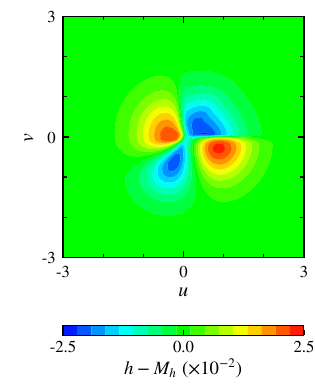}
  \panelcaption{(f) $\mathrm{P}_3$}
 \end{minipage}
 \caption{$h(u,v)$ (top) and $h-M_h$ (bottom) at the three sampling locations in the cavity at $\mathrm{Kn}=1$.}
 \label{fig:cavity-kn1-distributions}
\end{figure}
\FloatBarrier

Table~\ref{tab:cavity-cost} compares full nodal IUGKS, NCIUGKS, and accelerated NCIUGKS. Accelerated NCIUGKS refers to the method that uses multi step macro prediction based on separating the NS part from the high order residual in the macroscopic equations, conservative Anderson acceleration, and a coarse mesh warm start. Stepwise speedup compares total times in adjacent method columns; overall speedup is relative to full nodal IUGKS. Iteration costs refer to the fine mesh. In the ``Iterations'' rows of Table~\ref{tab:cavity-cost}, entries $a+b$ in the ``Accelerated NCIUGKS'' column give the coarse and fine mesh iteration counts, respectively.

At $\mathrm{Re}=1000$, DOF adaptation reduces the cost per iteration by a factor of $33.87$ and the total solution time from $22.91$~h to $40.58$~min. The accelerated method provides a further speedup of $69.21$, bringing the solution time to $35.18$~s and the overall speedup to $2344.27$.

The savings from DOF adaptation become smaller as nonequilibrium grows: at $\mathrm{Kn}=1$, the cost per iteration is reduced by a factor of $3.84$. Steady acceleration provides a further speedup of about four in both rarefied cases, giving overall speedups of $47.32$ at $\mathrm{Kn}=0.1$ and $13.71$ at $\mathrm{Kn}=1$. The largest benefit therefore occurs near the continuum limit, where fewer DOF combine with rapid steady convergence.

\begin{table}[!htbp]
 \centering
 \caption{Computational cost and speedup for the cavity tests.}
 \label{tab:cavity-cost}
 \small
 \renewcommand{\arraystretch}{1.15}
 \setlength{\tabcolsep}{3pt}
 \begin{tabular*}{\textwidth}{@{\extracolsep{\fill}}llccc@{}}
 \toprule
 Case & & Full nodal IUGKS & NCIUGKS & Accelerated NCIUGKS\\
 \midrule
$\mathrm{Re}=1000$ & Total time & 22.91 h & 40.58 min & \textbf{35.18 s}\\[-2pt]
 & & {\footnotesize($82\,467.00$ s)} & {\footnotesize($2\,434.74$ s)} & \\[2pt]
 & Stepwise speedup & --- & $33.87\times$ & $69.21\times$\\
 & Overall speedup & $1.00\times$ & $33.87\times$ & $\boldsymbol{2\,344.27\times}$\\
 \cmidrule(l){2-5}
 & Time/iteration (s) & 21.00 & 0.62 & 1.07\\
 & Iterations & 3927 & 3927 & $25+25$\\
 \midrule
$\mathrm{Kn}=0.1$ & Total time & 17.65 min & 1.61 min & \textbf{22.38 s}\\[-2pt]
 & & {\footnotesize($1\,059.10$ s)} & {\footnotesize($96.76$ s)} & \\[2pt]
 & Stepwise speedup & --- & $10.95\times$ & $4.32\times$\\
 & Overall speedup & $1.00\times$ & $10.95\times$ & $\boldsymbol{47.32\times}$\\
 \cmidrule(l){2-5}
 & Time/iteration (s) & 8.90 & 0.82 & 1.22\\
 & Iterations & 119 & 118 & $17+13$\\
 \midrule
$\mathrm{Kn}=1$ & Total time & 17.44 min & 4.90 min & \textbf{76.34 s}\\[-2pt]
 & & {\footnotesize($1\,046.50$ s)} & {\footnotesize($293.88$ s)} & \\[2pt]
 & Stepwise speedup & --- & $3.56\times$ & $3.85\times$\\
 & Overall speedup & $1.00\times$ & $3.56\times$ & $\boldsymbol{13.71\times}$\\
 \cmidrule(l){2-5}
 & Time/iteration (s) & 9.10 & 2.37 & 2.74\\
 & Iterations & 115 & 124 & $31+24$\\
 \bottomrule
 \end{tabular*}
\end{table}

\FloatBarrier

\subsection{Flow past a square cylinder}\label{sec:square}
Hypersonic flow past a square cylinder brings a detached bow shock and an expanding wake into the same flow field. The flow is computed at $\mathrm{Ma}=5$ and $\mathrm{Kn}=0.1$ with the implicit formulation. For a cylinder of side length $D=1$, symmetry reduces the domain to $[-6,8]\times[0,6]$, with the solid occupying $[-0.5,0.5]\times[0,0.5]$. The free stream has $\rho_\infty=\theta_\infty=1$, $U_\infty=5\sqrt{5/3}$, and $V_\infty=0$. The Knudsen number is based on the VHS mean free path, $\mathrm{Kn}=\lambda_\infty/D$. The stationary cylinder is isothermal at $\theta_w=1$ and obeys diffuse reflection. The lower boundary outside the cylinder is a symmetry boundary; incoming molecules at the remaining boundaries follow the free stream distribution function.

A nonuniform $154\times84$ Cartesian mesh contains $12{,}096$ fluid cells, with a minimum spacing of $0.01D$. The velocity domain is $[-15,15]^2$. The state and implicit source block truncation tolerances are $\varepsilon_s^{\mathrm{ref}}=\varepsilon_{S,s}^{\mathrm{ref}}=10^{-4}$ for $s=h,b$. Coarse mesh warm start on $34\times17$ and $77\times42$ meshes is combined with multi step macro prediction and conservative Anderson acceleration.

Figure~\ref{fig:square-fields} displays the bow shock and wake expansion.

\begin{figure}[!htbp]
 \centering
 \begin{minipage}[t]{0.325\textwidth}
  \centering
  \includegraphics[width=\linewidth]{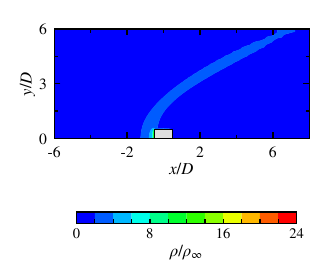}
  \panelcaption{(a) Density}
 \end{minipage}\hfill
 \begin{minipage}[t]{0.325\textwidth}
  \centering
  \includegraphics[width=\linewidth]{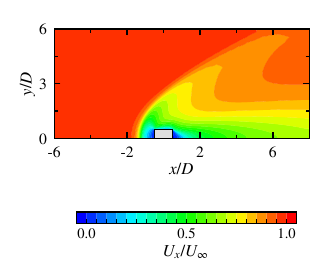}
  \panelcaption{(b) Velocity $U_x$}
 \end{minipage}\hfill
 \begin{minipage}[t]{0.325\textwidth}
  \centering
  \includegraphics[width=\linewidth]{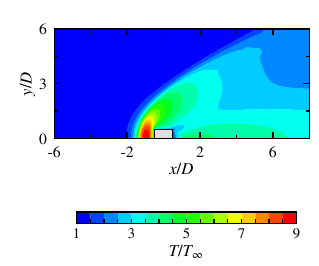}
  \panelcaption{(c) Temperature}
 \end{minipage}
 \caption{Normalized density, streamwise velocity, and temperature around the square cylinder at $\mathrm{Ma}=5$ and $\mathrm{Kn}=0.1$.}
 \label{fig:square-fields}
\end{figure}
\FloatBarrier

Along the stagnation line, the density, velocity, and temperature profiles in Fig.~\ref{fig:square-profiles} closely follow UGKS through the shock and toward the cylinder surface.

\begin{figure}[!htbp]
 \centering
 \begin{minipage}[t]{0.325\textwidth}
  \centering
  \includegraphics[width=\linewidth]{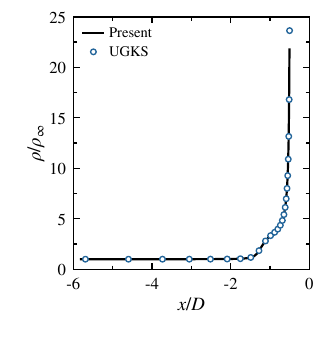}
  \panelcaption{(a) Density}
 \end{minipage}\hfill
 \begin{minipage}[t]{0.325\textwidth}
  \centering
  \includegraphics[width=\linewidth]{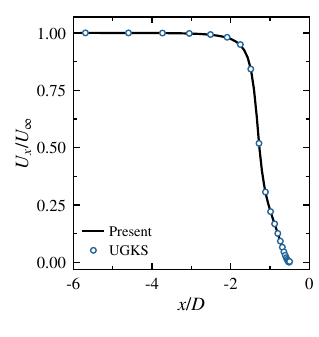}
  \panelcaption{(b) Velocity $U_x$}
 \end{minipage}\hfill
 \begin{minipage}[t]{0.325\textwidth}
  \centering
  \includegraphics[width=\linewidth]{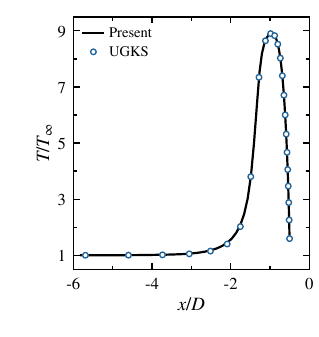}
  \panelcaption{(c) Temperature}
 \end{minipage}
 \caption{Stagnation line profiles for the square cylinder at $\mathrm{Ma}=5$ and $\mathrm{Kn}=0.1$. Solid lines denote the present method; circles denote UGKS.}
 \label{fig:square-profiles}
\end{figure}
\FloatBarrier

In Fig.~\ref{fig:square-adaptation}(a), the local DOF count follows the spatial variation of nonequilibrium. More DOF are needed across the shock, near the cylinder, and in the rarefaction wave, while few are needed in other regions, such as the free stream. Even in regions of strong nonequilibrium, the maximum local retained fraction is only about $12\%$. Figure~\ref{fig:square-adaptation}(b) records a retained fraction of $7.03\%$--$7.50\%$ during convergence on the fine mesh.

\begin{figure}[!htbp]
 \centering
 \begin{minipage}[t]{0.485\textwidth}
  \centering
  \includegraphics[width=\linewidth]{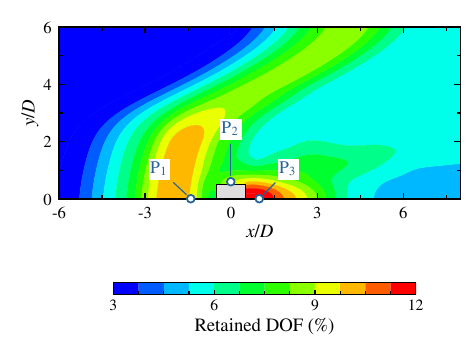}
  \panelcaption{(a) Spatial distribution}
 \end{minipage}\hfill
 \begin{minipage}[t]{0.485\textwidth}
  \centering
  \includegraphics[width=\linewidth]{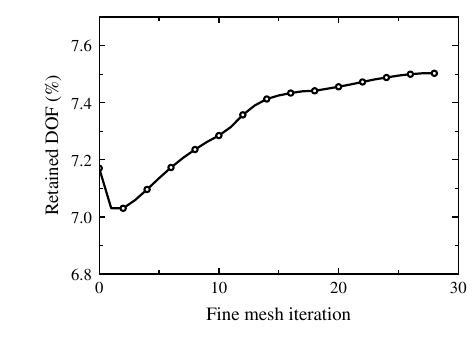}
  \panelcaption{(b) Iteration history}
 \end{minipage}
 \caption{Retained DOF for the square cylinder at $\mathrm{Ma}=5$ and $\mathrm{Kn}=0.1$. The history in (b) covers the fine mesh iteration. The sampling locations $(x/D,y/D)$ in (a) are $\mathrm{P}_1=(-1.4072,0.0119)$, $\mathrm{P}_2=(-0.0125,0.6056)$, and $\mathrm{P}_3=(0.9814,0.0119)$.}
 \label{fig:square-adaptation}
\end{figure}
\FloatBarrier
\clearpage

The three locations marked in Fig.~\ref{fig:square-adaptation}(a) sample the bow shock, the flow above the cylinder, and the near wake. Their distributions in Fig.~\ref{fig:square-distributions} exhibit pronounced departures from local equilibrium, with maximum nodal values of $|h-M_h|$ of $5.20\times10^{-2}$, $4.77\times10^{-2}$, and $1.12\times10^{-2}$, respectively. The distribution at $\mathrm{P}_1$ is dominated by a narrow peak, while that at $\mathrm{P}_3$ has two distinct peaks.

\begin{figure}[!htbp]
 \centering
 \begin{minipage}[t]{0.325\textwidth}
  \centering
  \includegraphics[width=\linewidth]{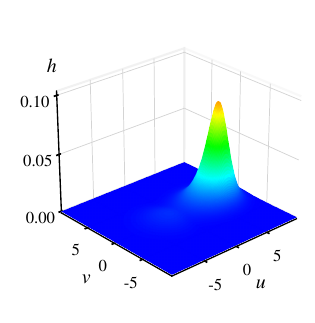}
  \panelcaption{(a) $\mathrm{P}_1$}
 \end{minipage}\hfill
 \begin{minipage}[t]{0.325\textwidth}
  \centering
  \includegraphics[width=\linewidth]{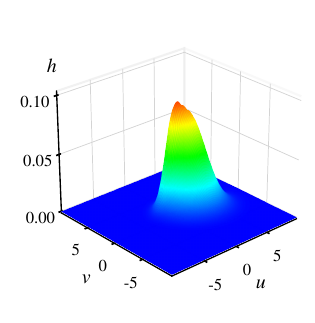}
  \panelcaption{(b) $\mathrm{P}_2$}
 \end{minipage}\hfill
 \begin{minipage}[t]{0.325\textwidth}
  \centering
  \includegraphics[width=\linewidth]{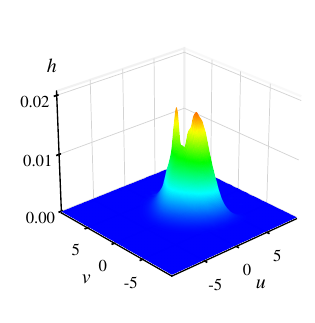}
  \panelcaption{(c) $\mathrm{P}_3$}
 \end{minipage}
 \par\medskip
 \begin{minipage}[t]{0.325\textwidth}
  \centering
  \includegraphics[width=\linewidth]{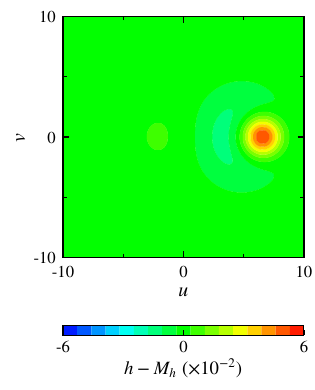}
  \panelcaption{(d) $\mathrm{P}_1$}
 \end{minipage}\hfill
 \begin{minipage}[t]{0.325\textwidth}
  \centering
  \includegraphics[width=\linewidth]{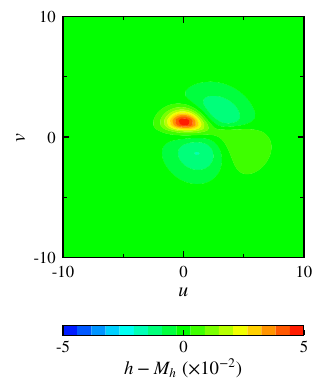}
  \panelcaption{(e) $\mathrm{P}_2$}
 \end{minipage}\hfill
 \begin{minipage}[t]{0.325\textwidth}
  \centering
  \includegraphics[width=\linewidth]{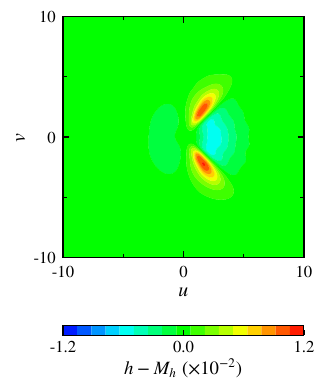}
  \panelcaption{(f) $\mathrm{P}_3$}
 \end{minipage}
 \caption{$h(u,v)$ (top) and $h-M_h$ (bottom) at the three locations marked in Fig.~\ref{fig:square-adaptation}(a) for the square cylinder.}
 \label{fig:square-distributions}
\end{figure}
\FloatBarrier

\begin{samepage}
The convergence curves in Fig.~\ref{fig:square-convergence} show that $R_{\Delta\rho}<10^{-4}$ is reached after $28$ iterations on the fine mesh.\par
\end{samepage}

\begin{figure}[!htbp]
 \centering
 \begin{minipage}[t]{0.485\textwidth}
  \centering
  \includegraphics[width=\linewidth]{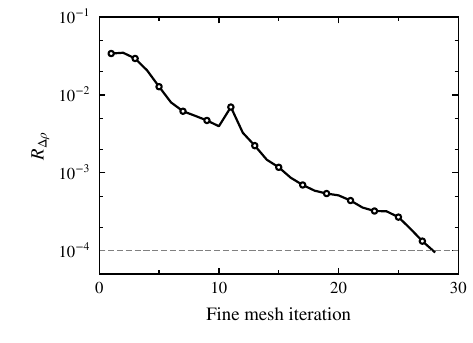}
  \panelcaption{(a) Iteration history}
 \end{minipage}\hfill
 \begin{minipage}[t]{0.485\textwidth}
  \centering
  \includegraphics[width=\linewidth]{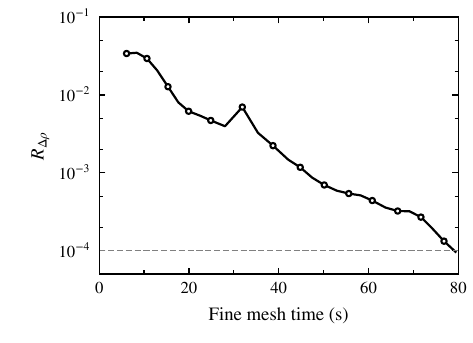}
  \panelcaption{(b) Time history}
 \end{minipage}
 \caption{$R_{\Delta\rho}$ during the fine mesh iteration for the square cylinder at $\mathrm{Ma}=5$, $\mathrm{Kn}=0.1$. Time in (b) is measured from the start of the fine mesh stage. The dashed line indicates $10^{-4}$.}
 \label{fig:square-convergence}
\end{figure}

\FloatBarrier
Table~\ref{tab:cost} uses the same definitions of stepwise and overall speedup as Table~\ref{tab:cavity-cost}. Accelerated NCIUGKS refers to the method that uses multi step macro prediction based on separating the NS part from the high order residual in the macroscopic equations, conservative Anderson acceleration, and a coarse mesh warm start. Iteration costs refer to the fine mesh. In the ``Iterations'' row and ``Accelerated NCIUGKS'' column of Table~\ref{tab:cost}, the entry $100+40+28$ gives the iteration counts on the $34\times17$, $77\times42$, and $154\times84$ meshes, respectively. Each coarse stage uses the prescribed iteration count and reduces $R_{\Delta\rho}$ by more than three orders of magnitude before transfer. The accelerated total includes initialization and all three mesh stages.

Each iteration costs $23.00$~s for full nodal IUGKS and $1.70$~s with adaptive kinetic DOF. The same $257$ iterations then take $436.9$~s ($7.28$~min) instead of $5\,911.0$~s ($1.64$~h), giving a $13.53$-fold speedup. The accelerated method provides a further $3.74$-fold speedup, reducing the total solution time to $116.8$~s and bringing the overall speedup to $50.61$.

\begin{table}[!htbp]
 \centering
 \caption{Computational cost and speedup for flow past the square cylinder at $\mathrm{Ma}=5$ and $\mathrm{Kn}=0.1$.}
 \label{tab:cost}
 \small
 \renewcommand{\arraystretch}{1.18}
 \setlength{\tabcolsep}{4pt}
 \begin{tabular*}{\textwidth}{@{\extracolsep{\fill}}lccc@{}}
 \toprule
 & Full nodal IUGKS & NCIUGKS & Accelerated NCIUGKS\\
 \midrule
 Total time
 & 1.64 h
 & 7.28 min
 & \textbf{116.8 s}\\[-2pt]
 & {\footnotesize($5\,911.0$ s)} & {\footnotesize($436.9$ s)} & \\[2pt]
 Stepwise speedup & --- & $13.53\times$ & $3.74\times$\\
 Overall speedup & $1.00\times$ & $13.53\times$ & $\boldsymbol{50.61\times}$\\
 \addlinespace[2pt]
 \midrule
 Time/iteration (s) & 23.00 & 1.70 & 2.80\\
 Iterations & 257 & 257 & $100+40+28$\\
 \bottomrule
 \end{tabular*}
\end{table}
\FloatBarrier
\clearpage

\section{Conclusions}\label{sec:conclusions}


A nonequilibrium capturing formulation for UGKS has been developed to mitigate the high computational cost associated with velocity space discretization in multiscale gas flows. In this approach, the local Maxwellian is evaluated analytically, while the nonequilibrium component is represented by an adaptive expansion in a discrete orthogonal basis. By evolving the retained DOF directly, the formulation reduces computational overhead while preserving the essential UGKS coupling of transport and collision processes. Furthermore, the number of DOFs can vary dynamically from cell to cell without requiring modifications to the underlying velocity grid. Both explicit and implicit formulations are provided to support unsteady and steady flow simulations.

When the full set of basis vectors is retained, the transport operators in the discrete basis recover their nodal counterparts for the same reconstructed distribution function. Local truncation of the basis preserves the conservative moments, and the implementation of a shared interface flux ensures the discrete balance of mass, momentum, and energy.

Numerical tests, including shock tube flow, shock structure, lid-driven cavity flow, and hypersonic flow, demonstrate close agreement with reference solutions. The method adaptively allocates kinetic DOFs. Minimal DOFs are required near equilibrium, whereas a larger number are retained as nonequilibrium effects become stronger or more widespread. The fraction of retained kinetic DOFs ranges from $1.07\%$ to $7.47\%$ in explicit calculations and from $2.13\%$ to $27.20\%$ in steady state calculations.

The explicit calculations achieve speedups ranging from $12.64$ to $58.50$ over the original UGKS. In implicit calculations, the reduction of nonequilibrium DOF alone decreases the computational cost per iteration by factors of $3.84$ to $33.87$. When combined with a coarse mesh warm start, multi step macro prediction, and conservative Anderson acceleration, the overall speedups in steady tests range from $13.7$ to approximately $2344$ relative to the original implicit UGKS. The most significant speedup is observed near the continuum limit, where a representation with fewer nonequilibrium DOF is effectively combined with rapid steady state convergence.

Adapting kinetic DOFs to local nonequilibrium thus provides an efficient and conservative framework for multiscale gas flow simulations within the UGKS architecture.

\section*{Acknowledgement}

The current research is supported by National Key R\&D Program of China (Grant Nos. 2022YFA1004500), National Science Foundation of China (92371107), and Hong Kong research grant council (16208324).

\clearpage
\appendix
\section{Conservative Anderson acceleration}\label{app:implicit}
The map $\mathcal G$ consists of one NCIUGKS update in Algorithm~\ref{alg:implicit-ugks}, including truncation, conservative correction~\eqref{eq:conservative-correction}, and optional multi step macro prediction. For $\bm z^j=(\W^j,\avec^j)$, the state before conservative Anderson acceleration~\cite{anderson1965} and its increment are
\begin{equation}
 \widetilde{\bm z}^{\,j}=\mathcal G(\bm z^j),\qquad
 \bm e^j=\widetilde{\bm z}^{\,j}-\bm z^j.
 \label{eq:anderson-map}
\end{equation}
Let $\mathcal H$ contain the recent iteration indices, including the current index $k$. The coefficients $\alpha_j$ are based on the weighted least squares problem
\begin{equation}
 \min_{\sum_{j\in\mathcal H}\alpha_j=1}
 \left\|\sum_{j\in\mathcal H}\alpha_j\bm e^j\right\|_*^2.
 \label{eq:anderson-regularized-fit}
\end{equation}
Here $\|\cdot\|_*^2$ weights the increments of $(\rho,\rho U,\rho V,\rho E)$ by $(1,U_{\mathrm{ref}}^{-2},U_{\mathrm{ref}}^{-2},E_{\mathrm{ref}}^{-2})$, where $U_{\mathrm{ref}}$ is the reference flow speed and $E_{\mathrm{ref}}=U_{\mathrm{ref}}^2/2+3\theta_{\mathrm{ref}}/2$, with $\theta_{\mathrm{ref}}=RT_{\mathrm{ref}}$. The $h$ and $b$ coefficients have relative weights $1$ and $\theta_{\mathrm{ref}}^{-2}$.
With a relaxation factor $\beta$, the accelerated state is
\begin{equation}
 \begin{aligned}
 \W^{\mathrm{AA}}
 &=\sum_{j\in\mathcal H}\alpha_j
   \bigl[(1-\beta)\W^j+\beta\widetilde{\W}^{\,j}\bigr],\\
 \avec^{\mathrm{AA}}
 &=\sum_{j\in\mathcal H}\alpha_j
   \bigl[(1-\beta)\avec^j+\beta\widetilde{\avec}^{\,j}\bigr].
 \end{aligned}
 \label{eq:anderson-state-combination}
\end{equation}
Here the superscript $\mathrm{AA}$ denotes the accelerated state.

At each iteration, unretained nonequilibrium DOF are set to zero. The coefficients $\avec_i^j$ are those at the end of the preceding iteration, after conservative correction, and satisfy $\Cmat\avec_i^j=\bm0$. The current conservative correction~\eqref{eq:conservative-correction} enforces $\Cmat\widetilde{\avec}_i^{\,j}=\bm0$. Hence,
\begin{equation}
 \Cmat\avec_i^{\mathrm{AA}}
 =\sum_{j\in\mathcal H}\alpha_j
   \bigl[(1-\beta)\Cmat\avec_i^j
         +\beta\Cmat\widetilde{\avec}_i^{\,j}\bigr]=\bm0.
 \label{eq:anderson-zero-moments}
\end{equation}
We call this conservative Anderson acceleration because the nonequilibrium part retains zero conservative moments. The Maxwellian is recomputed from $\W_i^{\mathrm{AA}}$.

The iteration history contains the current iterate and up to six previous iterates, with $\beta=1$.

\FloatBarrier

\section{Coarse mesh warm start}\label{app:warm-start}
The solution is interpolated from coarse cells $J$ ($\mathrm c$) to fine cells $i$ ($\mathrm f$):
\begin{equation}
 \W_i^{\mathrm f}=\sum_J\omega_{iJ}\W_J^{\mathrm c},
 \qquad
 \avec_i^{\mathrm f}=\sum_J\omega_{iJ}\avec_J^{\mathrm c}.
 \label{eq:warm-start-interpolation}
\end{equation}
The weights $\omega_{iJ}$ follow the interpolation in multigrid UGKS~\cite{zhu2017multigrid}, with linear reconstruction on the stretched mesh.
Since $\Cmat\avec_J^{\mathrm c}=\bm0$, interpolation preserves the zero conservative moments:
\begin{equation}
 \Cmat\avec_i^{\mathrm f}
 =\sum_J\omega_{iJ}\Cmat\avec_J^{\mathrm c}=\bm0.
 \label{eq:warm-start-zero-moments}
\end{equation}
The fine mesh Maxwellian is then evaluated from $\W_i^{\mathrm f}$.

\FloatBarrier

\section{Discrete basis expansion of the distribution function}\label{app:orthogonal-basis}
\setcounter{figure}{0}
\renewcommand{\theHfigure}{C.\arabic{figure}}
The tensor basis~\eqref{eq:tensor-basis} is formed from one dimensional orthonormal basis vectors. This appendix constructs these vectors by weighted Gram--Schmidt orthogonalization on each velocity quadrant.

\subsection{Discrete velocity space}\label{app:basis-discrete-space}
Let $N_u$ and $N_v$ be the numbers of nodes on each coordinate half range. The midpoint grid on $[-L_u,L_u]\times[-L_v,L_v]$ is
\begin{equation}
 \begin{aligned}
 u_k&=-L_u+(k+\tfrac12)\Delta u,
 &\Delta u&=\frac{L_u}{N_u},
 &0\le k<2N_u,\\
 v_l&=-L_v+(l+\tfrac12)\Delta v,
 &\Delta v&=\frac{L_v}{N_v},
 &0\le l<2N_v.
 \end{aligned}
 \label{eq:app-full-midpoint-grid}
\end{equation}
Its nodes and weights are $\uvec_\nu=(u_k,v_l)^{\mathsf T}$ and $w_\nu=\Delta u\Delta v$, with $N_{\mathrm{vel}}=4N_uN_v$. The quadrant index $q\in\{1,2,3,4\}$ follows the counterclockwise order:
\[
 \begin{array}{c|cccc}
 q & 1 & 2 & 3 & 4\\ \hline
  \operatorname{sign}(u) & + & - & - & +\\
 \operatorname{sign}(v) & + & + & - & -
 \end{array}
\]
Within $Q_q$, the coordinate nodes $u_k^q$ and $v_l^q$ are ordered increasingly, with $0\le k<N_u^q=N_u$ and $0\le l<N_v^q=N_v$. Their weights are $w_{u,k}^q=\Delta u$ and $w_{v,l}^q=\Delta v$. For either coordinate $\xi\in\{u,v\}$, the weighted inner product and norm are
\begin{equation}
 \langle\mathbf z,\mathbf y\rangle_{w_\xi^q}
 =\sum_{k=0}^{N_\xi^q-1}w_{\xi,k}^q z_k y_k,
 \qquad
 \|\mathbf z\|_{w_\xi^q}^2
 =\langle\mathbf z,\mathbf z\rangle_{w_\xi^q}.
 \label{eq:app-vector-inner-product}
\end{equation}

Figure~\ref{fig:velocity-grid-indices} shows the four quadrants and the nodes in $Q_1$.
\begin{figure}[H]
 \centering
 \begin{minipage}[t]{0.48\textwidth}
  \centering
  \includegraphics[width=0.80\linewidth]{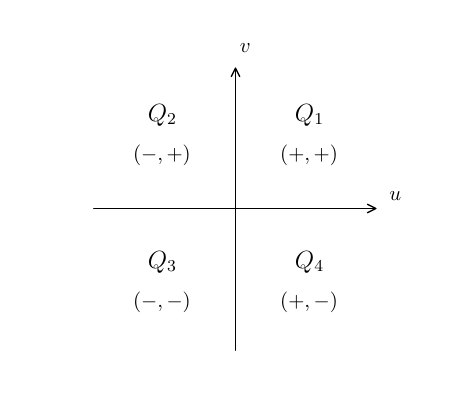}
  \panelcaption{(a) Velocity quadrants}
 \end{minipage}\hfill
 \begin{minipage}[t]{0.48\textwidth}
  \centering
  \includegraphics[width=0.80\linewidth]{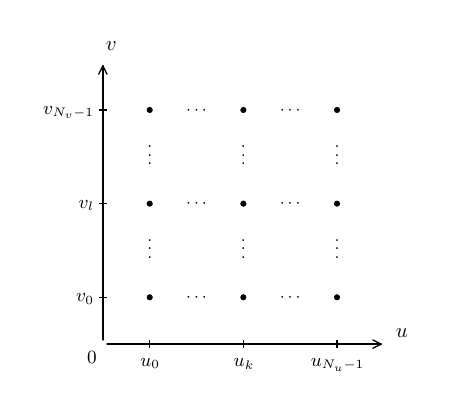}
  \panelcaption{(b) Nodes in $Q_1$}
 \end{minipage}
 \caption{Velocity space discretization: (a) quadrants $Q_1$--$Q_4$ and the signs of $(u,v)$; (b) nodes $u_k,v_l$ in $Q_1$, with the quadrant superscript omitted. Ellipses indicate omitted nodes.}
 \label{fig:velocity-grid-indices}
\end{figure}

\clearpage

\subsection{Basis construction}\label{app:basis-one-dimensional}\label{app:basis-two-dimensional}
Fix a quadrant $q$ and consider the $u$ direction. With $N=N_u^q$ and $\mathbf p_n=\mathbf p_{u,n}^q$, the velocity nodes are mapped to $[-1,1]$ by
\begin{equation}
 \widetilde u_k^q=\frac{u_k^q-c_u^q}{d_u^q},
 \qquad
 c_u^q=\frac{u_0^q+u_{N-1}^q}{2},
 \qquad
 d_u^q=\frac{u_{N-1}^q-u_0^q}{2}.
 \label{eq:app-normalized-velocity}
\end{equation}
The monomial vectors are defined by
\begin{equation}
 \mathbf e_n=\bigl((\widetilde u_0^q)^n,\ldots,(\widetilde u_{N-1}^q)^n\bigr)^{\mathsf T},
 \qquad 0\le n<N.
 \label{eq:app-monomial-definition}
\end{equation}
The first four are
\begin{equation}
 \mathbf e_0=\begin{pmatrix}1\\1\\\vdots\\1\end{pmatrix},\quad
 \mathbf e_1=\begin{pmatrix}\widetilde u_0^q\\\widetilde u_1^q\\\vdots\\\widetilde u_{N-1}^q\end{pmatrix},\quad
 \mathbf e_2=\begin{pmatrix}(\widetilde u_0^q)^2\\(\widetilde u_1^q)^2\\\vdots\\(\widetilde u_{N-1}^q)^2\end{pmatrix},\quad
 \mathbf e_3=\begin{pmatrix}(\widetilde u_0^q)^3\\(\widetilde u_1^q)^3\\\vdots\\(\widetilde u_{N-1}^q)^3\end{pmatrix}.
 \label{eq:app-monomial-vectors}
\end{equation}
The monomial vectors of degrees below $N$ are linearly independent, but they are generally not orthogonal. Gram--Schmidt subtracts projections onto the preceding basis directions and normalizes the orthogonal part at each step.

The constant vector is normalized directly:
\begin{equation}
 \mathbf p_0=\frac{\mathbf e_0}{\|\mathbf e_0\|_{w_u^q}}.
 \label{eq:first-basis-vectors}
\end{equation}
For degree one, remove the projection onto $\mathbf p_0$:
\begin{equation}
 \mathbf q_1=\mathbf e_1-\langle\mathbf e_1,\mathbf p_0\rangle_{w_u^q}\mathbf p_0,
 \qquad
 \mathbf p_1=\frac{\mathbf q_1}{\|\mathbf q_1\|_{w_u^q}}.
 \label{eq:app-linear-basis}
\end{equation}
For degrees two and three, remove the projections onto all preceding directions:
\begin{equation}
 \mathbf q_2=\mathbf e_2
 -\langle\mathbf e_2,\mathbf p_0\rangle_{w_u^q}\mathbf p_0
 -\langle\mathbf e_2,\mathbf p_1\rangle_{w_u^q}\mathbf p_1,
 \qquad
 \mathbf p_2=\frac{\mathbf q_2}{\|\mathbf q_2\|_{w_u^q}},
 \label{eq:app-quadratic-basis}
\end{equation}
\begin{equation}
 \mathbf q_3=\mathbf e_3
 -\langle\mathbf e_3,\mathbf p_0\rangle_{w_u^q}\mathbf p_0
 -\langle\mathbf e_3,\mathbf p_1\rangle_{w_u^q}\mathbf p_1
 -\langle\mathbf e_3,\mathbf p_2\rangle_{w_u^q}\mathbf p_2,
 \qquad
 \mathbf p_3=\frac{\mathbf q_3}{\|\mathbf q_3\|_{w_u^q}}.
 \label{eq:app-cubic-basis}
\end{equation}
In general, for $0\le n<N$,
\begin{equation}
 \mathbf q_n=\mathbf e_n-\sum_{j=0}^{n-1}
 \langle\mathbf e_n,\mathbf p_j\rangle_{w_u^q}\mathbf p_j,
 \qquad
 \mathbf p_n=\frac{\mathbf q_n}{\|\mathbf q_n\|_{w_u^q}}.
 \label{eq:app-gram-schmidt}
\end{equation}
The resulting vectors satisfy $\langle\mathbf p_m,\mathbf p_n\rangle_{w_u^q}=\delta_{mn}$. The $v$ basis follows the same construction with its own normalized nodes and weights.

\clearpage
\FloatBarrier

\section{Distribution function in a one dimensional shock}\label{app:shock-distribution}
\setcounter{figure}{0}
\renewcommand{\theHfigure}{D.\arabic{figure}}
The one dimensional calculation uses the streamwise grid and boundary states of Section~\ref{sec:shock}, $100$ midpoint velocity nodes on $[-15,15]$, and $\varepsilon^{\mathrm{ref}}=10^{-5}$ for both reduced components. At $t=250$, we examine the cell average distribution function in the cell centered at $x_i=0.75$.

In this one dimensional reduction, $h(u)$ and $b(u)$ are the integrals of $f$ and $(v^2+\xi^2)f$, respectively, over both transverse velocities $(v,\xi)\in\mathbb R^2$. Their cell equilibria are
\[
 M_h(u)=\frac{\rho}{\sqrt{2\pi\theta}}
        \exp\!\left[-\frac{(u-U_x)^2}{2\theta}\right],
 \qquad M_b(u)=2\theta M_h(u).
\]
The curves in Fig.~\ref{fig:shock-distribution} are reconstructed from successive partial sums of the same nonequilibrium expansion, with $M_h(u)$ held fixed. The pair $(\mathbf p_{u,m}^{+},\mathbf p_{u,n}^{-})$ denotes the highest-degree basis vectors retained on the positive and negative velocity halves, respectively; their degree pair is $(m,n)$. The final retained expansion has $(m,n)=(34,9)$.

\begin{figure}[H]
 \centering
 \includegraphics[width=0.635\textwidth]{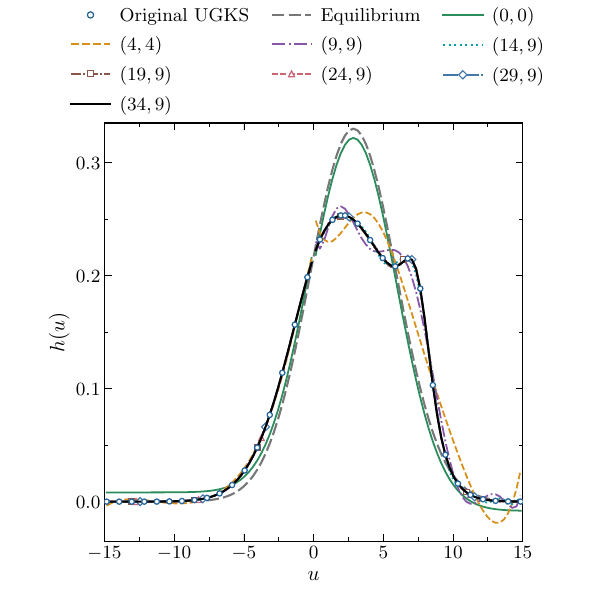}
 \caption{Discrete basis approximations of $h(u)$ in the one dimensional $\mathrm{Ma}=8$ shock. Circles denote original UGKS; the dashed gray curve is $M_h(u)$.}
 \label{fig:shock-distribution}
\end{figure}

The $(0,0)$ and $(4,4)$ sums miss the bimodal structure. The $(29,9)$ expansion recovers the grid locations of both peaks and the intervening minimum, with a maximum absolute nodal error of $6.44\times10^{-5}$ relative to the original UGKS solution.

\clearpage
\FloatBarrier

\begingroup
\hyphenpenalty=200
\interlinepenalty=10000
\bibliographystyle{elsarticle-num}
\bibliography{references}
\endgroup
\end{document}